\newif\ifllncs  %
\newif\ifsubmit %
\newif\ifieeecs %
\newif\ifexabs  %

\submittrue

\ifieeecs
  \documentclass[10pt, conference, compsocconf]{IEEEtran}
\else\ifllncs
  \documentclass{llncs}

\else
  \documentclass[letterpaper,11pt]{article}
  
  \usepackage[in]{fullpage}
\fi
\fi

\usepackage{iftex}
\ifPDFTeX
  \usepackage[utf8]{inputenc}
  \usepackage[noTeX]{mmap}
  \usepackage[T1]{fontenc}
\fi
\ifLuaTeX
  \usepackage{luatex85}
  \usepackage[noTeX]{mmap}
\fi

\usepackage{amsmath, amsfonts, amsthm, amssymb, amsthm, color}
\usepackage[bookmarks]{hyperref}
\usepackage[nameinlink]{cleveref}
\ifllncs\else
\ifexabs
  \newtheorem{theorem}{Theorem}
\else
\ifieeecs
  \newtheorem{theorem}{Theorem}
\else
  \newtheorem{theorem}{Theorem}[section]
\fi
\fi
\newtheorem{definition}[theorem]{Definition}
\newtheorem{remark}[theorem]{Remark}
\newtheorem{lemma}[theorem]{Lemma}
\newtheorem{corollary}[theorem]{Corollary}
\newtheorem{proposition}[theorem]{Proposition}

\fi
\newtheorem{fact}[theorem]{Fact}
\newtheorem{openproblem}[theorem]{Open Problem}

\newtheorem*{theorem*}{Theorem}
\newtheorem*{lemma*}{Lemma}
\newtheorem*{remark*}{Remark}
\newtheorem{innercustomtheorem}{Theorem}
\newenvironment{customtheorem}[1]
  {\renewcommand\theinnercustomtheorem{#1}\innercustomtheorem}
  {\endinnercustomtheorem}
\newtheorem{innercustomlemma}{Lemma}
\newenvironment{customlemma}[1]
  {\renewcommand\theinnercustomlemma{#1}\innercustomlemma}
  {\endinnercustomlemma}

\ifieeecs\else
  \usepackage{appendix}
\fi
\usepackage{algorithm}
\usepackage{algpseudocode}
\usepackage{braket}
\usepackage{comment}
\usepackage{url}

\usepackage{pgfplots, tikz}
\pgfplotsset{compat=1.14}

\usepackage[style=alphabetic,minalphanames=3,maxalphanames=4,maxnames=99]{biblatex}
\renewcommand*{\labelalphaothers}{\textsuperscript{+}}
\renewcommand*{\multicitedelim}{\addcomma\space}
\usepackage{soul, xcolor, xparse}
\makeatletter
  \ExplSyntaxOn
    \cs_new:Npn \white_text:n #1
    {
      \fp_set:Nn \l_tmpa_fp {#1 * .01}
      \llap{\textcolor{white}{\the\SOUL@syllable}\hspace{\fp_to_decimal:N \l_tmpa_fp em}}
      \llap{\textcolor{white}{\the\SOUL@syllable}\hspace{-\fp_to_decimal:N \l_tmpa_fp em}}
    }
    \NewDocumentCommand{\whiten}{ m }
    {
      \int_step_function:nnnN {1}{1}{#1} \white_text:n
    }
  \ExplSyntaxOff
  
  \NewDocumentCommand{ \varul }{ D<>{5} O{0.2ex} O{0.1ex} +m } {%
    \begingroup
    \setul{#2}{#3}%
    \def\SOUL@uleverysyllable{%
      \setbox0=\hbox{\the\SOUL@syllable}%
      \ifdim\dp0>\z@
      \SOUL@ulunderline{\phantom{\the\SOUL@syllable}}%
      \whiten{#1}%
      \llap{%
        \the\SOUL@syllable
        \SOUL@setkern\SOUL@charkern
      }%
      \else
      \SOUL@ulunderline{%
        \the\SOUL@syllable
        \SOUL@setkern\SOUL@charkern
      }%
      \fi}%
    \ul{#4}%
    \endgroup
  }
\makeatother

\newcommand{\equad}{\mathrel{\phantom{=}}}
\newcommand{\E}{\mathop{\mathbb{E}}}
\newcommand{\Z}{\mathbb{Z}}

\newcommand{\As}{\mathcal{A}}
\newcommand{\Bs}{\mathcal{B}}

\newcommand{\Ds}{\mathcal{D}}

\newcommand{\owf}{{\sf OWF}}
\newcommand{\owfy}{{{\sf OWF}_y}}
\newcommand{\yb}{{\sf YB}}
\newcommand{\gyb}{{\sf YB0}}
\newcommand{\prg}{{\sf PRG}}

\newcommand{\crh}{{\sf CRH}}

\newcommand{\samp}{{\sf Samp}}
\newcommand{\query}{{\sf Query}}
\newcommand{\ver}{{\sf Ver}}
\newcommand{\ch}{{\sf ch}}
\newcommand{\ans}{{\sf ans}}
\renewcommand{\st}{{\sf st}}

\newcommand{\sto}{{\sf StO}}
\newcommand{\pho}{{\sf PhO}}

\newcommand{\csto}{{\sf CStO}}
\newcommand{\stddecomp}{{\sf StdDecomp}}
\newcommand{\aux}{{\sf aux}}
\newcommand{\err}{{\sf err}}

\newcommand{\vub}{{\mathbf{B}}}
\newcommand{\vux}{{\mathbf{X}}}
\newcommand{\vx}{{\mathbf{x}}}
\newcommand{\vuy}{{\mathbf{Y}}}
\newcommand{\vy}{{\mathbf{y}}}

\newcommand{\eps}{\varepsilon}

\newcommand{\Gmi}[1]{G^{\otimes #1}}
\newcommand{\Cmi}[1]{C^{\otimes #1}}

\newcommand{\poly}{{\sf poly}}
\newcommand{\negl}{{\sf negl}}
\DeclareMathOperator{\Tr}{Tr}

\newcommand{\innerprod}[1]{{\left\langle#1\right\rangle}}
\newcommand{\abs}[1]{{\left\lvert#1\right\rvert}}
\newcommand{\ltwonorm}[1]{{\abs{#1}_2}}
\newcommand{\norm}[1]{{\left\lVert#1\right\rVert}}
\newcommand{\trdist}[1]{{\norm{#1}_{\operatorname{tr}}}}

\newcommand{\qpmwinit}{{\sf QPMW.Init}}
\newcommand{\qpmwest}{{\sf QPMW.Est}}

\ifsubmit
    \newcommand{\luowen}[1]{}
    \newcommand{\qipeng}[1]{}
    \newcommand{\siyao}[1]{}
    \newcommand{\KM}[1]{}
\else
    \newcommand{\luowen}[1]{{\color{magenta} Luowen: #1}}
    \newcommand{\qipeng}[1]{{\color{red} Qipeng: #1}}
    \newcommand{\siyao}[1]{{\color{blue} Siyao: #1}}
    \newcommand{\KM}[1]{{\color{green} KM: #1}}
\fi

\newcommand{\good}{{\sf good}}

\begin{document}

\ifexabs
\else
\ifieeecs\else
\pagenumbering{gobble} %
\fi

\title{Tight Quantum Time-Space Tradeoffs for Function Inversion}
\fi
\newcommand{\email}[1]{\href{mailto:#1}{\texttt{#1}}}
\ifieeecs
  \author{\IEEEauthorblockN{Kai-Min Chung\IEEEauthorrefmark{1},
      Siyao Guo\IEEEauthorrefmark{2},
      Qipeng Liu\IEEEauthorrefmark{3} and
      Luowen Qian\IEEEauthorrefmark{4}}
    \IEEEauthorblockA{\IEEEauthorrefmark{1}Institute of Information Science, Academia Sinica\\
      Taipei, Taiwan\\
      \email{kmchung@iis.sinica.edu.tw}}
    \IEEEauthorblockA{\IEEEauthorrefmark{2} Department of Computer Science, New York University Shanghai\\
      Shanghai, China\\
      \email{siyao.guo@nyu.edu}}
    \IEEEauthorblockA{\IEEEauthorrefmark{3}
      Department of Computer Science, Princeton University \& NTT Research\\
      Princeton, USA \\
      \email{qipengl@cs.princeton.edu}}
    \IEEEauthorblockA{\IEEEauthorrefmark{4}Department of Computer Science, Boston University\\
      Boston, USA\\
      \email{luowenq@bu.edu}}}
\else
\author{
    Kai-Min Chung\footnote{Academia Sinica. Email: \email{kmchung@iis.sinica.edu.tw}}
    \and
    Siyao Guo\footnote{New York University Shanghai. Email: \email{siyao.guo@nyu.edu}}
    \and
    Qipeng Liu\footnote{Princeton University \& NTT Research. Email: \email{qipengl@cs.princeton.edu}}
    \and
    Luowen Qian\footnote{Boston University. Email: \email{luowenq@bu.edu}}
}
\date{}
\fi

\maketitle

\ifexabs
\else
\begin{abstract}
In function inversion, we are given a function $f: [N] \mapsto [N]$, and want to prepare some advice of size $S$, such that we can efficiently invert any image in time $T$. 
This is a well studied problem with profound connections to cryptography, data structures, communication complexity, and circuit lower bounds.
Investigation of this problem in the quantum setting was initiated by Nayebi, Aaronson, Belovs, and Trevisan (2015), who proved a lower bound of $ST^2 = \tilde\Omega(N)$ for random permutations against classical advice, leaving open an intriguing possibility that Grover's search can be sped up to time $\tilde O(\sqrt{N/S})$. %
Recent works by Hhan, Xagawa, and Yamakawa (2019), and Chung, Liao, and Qian (2019) extended the argument for random functions and quantum advice, but the lower bound remains $ST^2 = \tilde\Omega(N)$. 

In this work, we prove that even with quantum advice, $ST + T^2 = \tilde\Omega(N)$ is required for an algorithm to invert random
functions. %
This demonstrates that Grover's search is optimal for $S = \tilde O(\sqrt{N})$, ruling out any substantial speed-up for Grover's search even with quantum advice.
Further improvements to our bounds would imply new classical circuit lower bounds, as shown by Corrigan-Gibbs and Kogan (2019).

To prove this result, we develop a general framework for establishing quantum time-space lower bounds.
We further demonstrate the power of our framework by proving the following results.
\begin{itemize}
    \item Yao's box problem: We prove a tight quantum time-space lower bound for classical advice. %
          For quantum advice, we prove a first time-space lower bound using shadow tomography. %
          These results resolve two open problems posted by 
          Nayebi et al
           (2015).
    \item Salted cryptography: We show that ``salting generically provably defeats preprocessing,'' a result shown by Coretti, Dodis, Guo, and Steinberger (2018), also holds in the quantum setting.
          In particular, we prove quantum time-space lower bounds for a wide class of salted cryptographic primitives in the quantum random oracle model.
          This yields the first quantum time-space lower bound for salted collision-finding, which in turn implies that $\mathsf{PWPP}^{\mathcal O} \not\subseteq \mathsf{FBQP}^{\mathcal O}\mathsf{/qpoly}$ relative to a random oracle $\mathcal O$.%
\end{itemize}

\end{abstract}

\ifieeecs

  \let\paragraph\subsubsection
  \let\subsubsection\subsection
  \let\subsection\section
\else
  \newpage
  \tableofcontents
  \newpage
  \pagenumbering{arabic}
  
  \section{Introduction}
\fi
\fi

\subsection{Function Inversion}
\label{sec:intro-owf}

The task of function inversion asks that given a function $f: [N] \mapsto [N]$ and a point $y$, find an $x$ such that $f(x) = y$.
\ifexabs\else
It is easy to show that a classical inversion algorithm requires $\Omega(N)$ queries (which trivially lower bounds time) to succeed with constant probability in inverting a random function $f$, even with the help of randomness.
Grover~\cite{grover1996fast} considered the same problem in the context of database search, and showed that quantum computers can invert any function in time only $\tilde O(\sqrt N)$, which was subsequently shown to be tight~\cite{bennett1997strengths}.

\fi
The situation becomes intriguing when preprocessing is allowed.
Namely, we allow the algorithm to take the entire truth table of $f$ and arbitrarily preprocess an $S$-bit advice string $\alpha = \alpha(f)$, and then we give the algorithm a random $f(x)$ and ask the algorithm to invert it using at most $T$ queries.
Understanding the tradeoff between $S$ and $T$ is referred to as time-space tradeoffs for function inversion.
This tradeoff is an important problem in cryptography.
Recent works~\cite{GGHPV19,KP19,corrigan2019function} showed its connections to well studied problems in data structures, communication complexity, and circuit lower bounds.

\ifexabs
\else
For classical algorithms, the heuristic algorithm proposed by Hellman~\cite{hellman1980cryptanalytic}, and subsequently rigorously analyzed by Fiat and Naor~\cite{FN99}, uses $S$ bits of advice and $T$ queries to invert a random function with high probability for every $S,T$ satisfying $S^2T \geq \tilde{O}(N^2)$.
For the lower bound, Yao~\cite{yao1990coherent} and De et al.~\cite{DTT10} proved that any preprocessing algorithm that uses $S$ bits of advice and $T$ queries must satisfy $ST = \tilde{\Omega}(N)$, which remains the best general lower bound we know today.
Corrigan-Gibbs and Kogan~\cite{corrigan2019function} recently investigated possible improvements on the lower bound, and showed that any improvements on Yao's lower bound will lead to improved circuit lower bounds. 

The study of time-space tradeoffs in the quantum setting was initiated by Nayebi, Aaronson, Belovs and Trevisan~\cite{NABT15}.
When the preprocessing algorithm is quantum, it is natural to distinguish the cases of quantum versus classical advice, as analogous to the complexity classes of $\mathsf{BQP}/\mathsf{qpoly}$ versus $\mathsf{BQP}/\mathsf{poly}$ and $\mathsf{QMA}$ versus $\mathsf{QCMA}$. Nayebi et al.~\cite{NABT15} showed that any quantum preprocessing algorithm that uses $S$ bits of \emph{classical} advice and $T$ queries to invert a random permutation with a constant probability must satisfy $ST^2 = \tilde{\Omega}(N)$. Note that the $T^2$ term is necessary given Grover's search algorithm. Recently, with motivations from post-quantum cryptography, Hhan, Xagawa, and Yamakawa~\cite{hxy19} extended the lower bound to handle general function inversion (and other cryptographic primitives).
For the more challenging case of algorithms with \emph{quantum} advice, Hhan el al.~\cite{hxy19} and Chung et al.~\cite{chung2019lower} proved lower bounds for inverting random permutation and a restricted class of random functions, specifically, the lower bound only holds for functions with roughly the same domain and image size. However, these lower bounds remain $ST^2 = \tilde{\Omega}(N)$. As pointed out by Nayebi et al.~\cite{NABT15}, this leaves open the following intriguing possibility:  
\begin{quote}
 {\it Could a piece of preprocessed advice help speed up Grover's search algorithm?
 }
\end{quote}
\fi

In this work, we prove the following quantum time-space lower bound for function inversion, which shows that even quantum advice of size $S = O(\sqrt{N})$ does not help speed up Grover's search algorithm.
\begin{theorem}\label{thm:owf}
 Let $f: [N] \mapsto [N]$ be a random function. For any quantum oracle algorithm $\As$ with $S$-qubit oracle-dependent advice $\alpha = \alpha(f)$ and $T$ queries to $f$, 
 \begin{itemize}
     \item if $\alpha$ is classical (i.e. an $S$-bit string), then
     $$\Pr\Big[\As^f(\alpha,f(x))\in f^{-1}(f(x))\Big] = \tilde{O} \left(\frac{ST+T^2}{N}\right);$$
     \item if $\alpha$ is quantum (i.e. the general case), then
       $$\Pr\Big[\As^f(\alpha,f(x))\in f^{-1}(f(x))\Big] = \tilde{O} \left(\sqrt[3]{\frac{ST+T^2}{N}}\right),$$
 \end{itemize}
where both probabilities are over $f$, a uniformly random $x$ from $[N]$, and randomness of $\As$. 
\end{theorem}
Our lower bound implies that for a quantum preprocessing algorithm to invert a random function with a constant probability, it must satisfy $ST + T^2 = \tilde{\Omega}(N)$ even for the case of quantum advice.
This further shows that Grover's search is optimal for $S = \tilde O(\sqrt{N})$, ruling out any substantial speed-up for Grover's search even with quantum advice.
For $S = \tilde\omega(\sqrt{N})$, our lower bound matches Yao's lower bound for classical algorithms and is the best provable lower bound in light of the above-mentioned barrier results of~\cite{corrigan2019function}.

Furthermore, in the context of cryptography, a typical complexity measure to define the security is $S+T$ (corresponding to the program length and the running-time), with respect to which our lower bound implies a lower bound $S+T \geq \tilde\Omega(\sqrt{N})$. This matches Grover's search algorithm and gives a tight characterization for function inversion.

A comparison of our bounds with known upper and lower bounds is included in  Figure~\ref{fig:owf-bounds}.  We remark that our lower bounds can be extended to general functions $f:[N]\mapsto[M]$
\ifexabs\else\ifieeecs\else(see \Cref{thm:owfaiqrom} and \Cref{thm:owfqaiqrom})\fi\fi.%

\begin{figure}
    \centering
    \ifieeecs
      \input{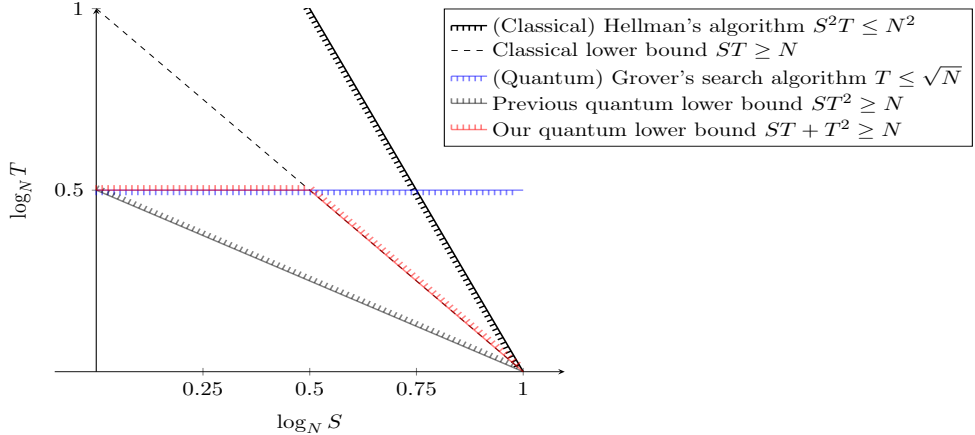}
    \else
      \resizebox{13cm}{5cm}{%
\begin{tikzpicture}[baseline]
\pgfplotsset{IneqStyleLess/.style = {%
    decoration={border,segment length=1mm, amplitude=1mm,angle=-90},
    postaction={decorate,draw},
    thick
  }
}

\pgfplotsset{IneqStyleGtr/.style = {%
    decoration={border,segment length=1mm, amplitude=1mm,angle=90},
    postaction={decorate,draw},
    thick
  }
}

\begin{axis}[
width=4in,
axis lines=middle,
axis equal,
xmin=0,xmax=1,
ymin=0,ymax=1,
domain=0:1,
xlabel=$\log_N S$,ylabel=$\log_N T$,
x label style={at={(axis description cs:.5,-.10)},anchor=north},
y label style={at={(axis description cs:-.03,.5)},rotate=90,anchor=south},
xtick={0, .25, .5, .75, 1},
ytick={0, .5, 1},
yticklabel={\ifdim\tick pt=0pt \else\pgfmathprintnumber{\tick}\fi},
legend style={at={(1.8,1)},cells={align=left}},
legend cell align=left,
clip=true,
]
\addplot[mark=none,IneqStyleLess,black!100!white]{2-2*x};
\addlegendentry{(Classical) Hellman's algorithm $S^2 T \le N^2$}
\addplot[mark=none,draw=black,dashed]{1-x};
\addlegendentry{Classical lower bound $ST \geq N$}
\addplot[mark=none,IneqStyleLess,blue!100!white,opacity=0.5]{0.5};
\addlegendentry{(Quantum) Grover's search algorithm $T \le \sqrt N$}
\addplot[mark=none,IneqStyleGtr,black!100!white,opacity=0.5]{1/2-x/2};
\addlegendentry{Previous quantum lower bound $ST^2 \ge N$}
\addplot[mark=none,IneqStyleGtr,red!100!white,opacity=0.5]
    coordinates{(0, 1/2) (1/2, 1/2) (1, 0)};
\addlegendentry{Our quantum lower bound $ST + T^2 \ge N$}
\end{axis}
\end{tikzpicture}
}
    \fi
    \caption{Time-space tradeoffs for inverting a random function $f:[N]\mapsto[N]$ with a constant successful probability.  For the classical setting,  the best upper bound is given by Hellman's algorithm~\cite{hellman1980cryptanalytic,FN99}, and the best lower bound is given by \cite{yao1990coherent, EC:DodGuoKat17,EC:CDGS18}. For the quantum setting, the best upper bound is given either by Grover's search algorithm~\cite{grover1996fast} or (classical) Hellman's algorithm, and the previous best lower bound is given by~\cite{hxy19,chung2019lower}. For simplicity, logarithmic terms and constant factors are omitted.} %
    \label{fig:owf-bounds}
\end{figure}

\subsection{Yao's Box Problem}

Yao's box problem~\cite{yao1990coherent} is another basic problem investigated in the literature of time-space tradeoffs that is closely related to the function inversion problem and non-uniform security of pseudorandom generators. In this problem, a preprocessing algorithm $\As$ can compute an $S$-bit advice for a function $f : [N] \mapsto \{0, 1\}$ in the preprocessing phase, and then in the online phase, it is required to compute the bit $f(x) \in \{0, 1\}$ for a random point $x\in [N]$ by making at most $T$ queries but without querying $f$ on $x$. 
Similar formulations of such problem in the classical setting are recently studied in the context of circuit complexity~\cite{ST18,MW19}, which led to a new result on depth $3$ circuits.

\ifexabs\else
In the classical setting, Yao~\cite{yao1990coherent} proved that for $\As$ to succeed with probability $2/3$, it must satisfy $ST = \Omega(N)$, which is known to be optimal.
In the quantum setting, Nayebi et al.~\cite{NABT15} showed an $S T^2 = \Omega(N)$ lower bound for solving Yao's box problem with \emph{classical} advice, which does not rule out the possibility that a preprocessing algorithm with $S$ bits of advice and $T = O(\sqrt{N/S})$ queries.
Subsequently, Hhan et al.~\cite[Lemma 6]{hxy19} refined the analysis and showed that for any $S, T, N$, any quantum algorithm with classical advice can only succeed with probability $1/2 + \tilde O(ST^2/N)^{1/6}$.
However, the following two problems posted by Nayebi et al. remained open.

\fi
\begin{openproblem}[{\cite[Section 3.4, Section 5]{nayebi2014quantum}}]
  Is there a quantum algorithm that solves Yao's box in time $T = O(\sqrt{N/S})$? Or equivalently, prove or disprove the optimality of the lower bound $ST^2 = \tilde\Omega(N)$ for Yao's box.
\end{openproblem}
\begin{openproblem}[{\cite[Section 5]{nayebi2014quantum}}]
  Extend the lower bound for Yao's box to the setting where the advice can be an arbitrary quantum state.
\end{openproblem}

We prove the following theorem for Yao's box problem, answering both open problems above.

\begin{theorem}
 Let $f: [N] \mapsto \{0,1\}$ be a random function. For any quantum oracle algorithms $\As$ with $S$-qubit oracle-dependent advice $\alpha = \alpha(f)$ and $T$ queries to $f$ except on the given challenge point $x$, 
 \begin{itemize}
     \item if $\alpha$ is classical, then
     $$\Pr[\As^f(\alpha,x)=f(x)] \le \frac{1}{2}+\tilde{O} \left(\sqrt[3]{\frac{ST}{N}}\right);$$
     \item if $\alpha$ is quantum, then
       $$\Pr[\As^f(\alpha,x)=f(x)] \le \frac{1}{2}+\tilde{O} \left(\sqrt[19]{\frac{S^5T}{N}}\right),$$
 \end{itemize}
where both probabilities are over $f$, a uniformly random $x$ from $[N]$, and randomness of $\As$. 
\end{theorem}

In particular, this theorem implies that any algorithm achieving success probability $2/3$ has to satisfy $ST = \tilde\Omega(N)$ for classical advice, which shows that the power of quantum query does not help solving Yao's box problem, and $S^5 T = \tilde\Omega(N)$ for quantum advice.

\ifexabs
\else
We note that while in the classical setting, time-space lower bounds for Yao's box problem and the function inversion problem can be proved using the same techniques, such as compression~\cite{DTT10, EC:DodGuoKat17} or presampling~\cite{C:Unruh07, EC:CDGS18}, proving quantum time-space lower bounds for Yao's box problem seems to be more challenging for quantum advice.
This is also the case for our framework since in Yao's box problem, the algorithm cannot verify the answer on its own.
For the case of quantum advice, we employ a novel use of online shadow tomography~\cite{AR19}, which enables us to prove the first time-space lower bound for algorithms with quantum advice.
We discuss the issue more carefully in \Cref{sec:tech-overview-yaos-box}.

Finally, we note that while we did not explicitly prove this, we believe that our techniques are sufficient for proving a lower bound $ST = \tilde\Omega(N)$ for quantum advice, if we restrict the algorithm such that it has to recover $f(x)$ for any $x \in [N]$ with probability $2/3$.
We discuss this further in \Cref{rem:yao-box-worst-case}.

\subsection{Post-Quantum Non-Uniform Security}

It turns out that our techniques for proving the two results above are fairly general, and we can use them to prove a variety of quantum non-uniform lower bounds.
We now turn our focus to proving non-uniform lower bounds in cryptographic ideal models, which is a topic that has gained a lot of momentum recently.

\paragraph*{Random oracle methodology and concrete security.}
While theoretically we can instantiate a lot of private-key cryptography assuming only the existence of any one-way function~\cite{Levin87}, the constructions are almost always way too inefficient to be any useful for practical purposes.
Practical cryptographic schemes, instead are usually designed under an ideal model and are proven secure under that model.
A popular model of choice is random oracle model (ROM)~\cite{BR93}, which ``heuristically'' models a function in question $f: [N] \mapsto [M]$, say a SHA3 hash function, as a truly random function that everyone has only oracle access to.
To determine the security parameter to use in the construction, usually concrete security bounds are derived in the ideal model, and we work out the calculations to make sure any adversary in the ideal model, using at most a reasonable amount of resources, can only succeed with small enough probability, usually $2^{-32}$~\cite{rfc5926,irtf-cfrg-xchacha-03}.

While the random oracle model ``heuristically'' captures all attacks that do not employ the structure of any particular instantiation of the random oracle, the model itself captures neither preprocessing attackers nor quantum attackers.
Moreover, security bounds obtained in the random oracle model are inaccurate or do not apply at all once preprocessing or quantum computation are allowed.

\paragraph*{Quantum random oracle model.}
Quantum algorithms are known to achieve various nontrivial speedups compared to classical algorithms, sometimes even an exponential speedup.
For example, Grover's search algorithm achieves $\sqrt N$ speedup over classical algorithms, and Shor's algorithm solves factoring in quantum polynomial time.

Motivated by assessing the post-quantum security of constructions in ROM, Boneh et al.~\cite{AC:BDFLSZ11} introduced the quantum random oracle model (QROM) where the attacker can make superposition queries to the oracle, as given a classical description of any function, a quantum algorithm can trivially perform such superposition queries.
Extending security proofs in classical ROM into the stronger QROM has been an active area of investigation~\cite{AC:BDFLSZ11,C:Zhandry12,C:BonZha13,EC:Unruh15,TarghiU16,AC:Unruh17,KiltzLS18,zhandry2019record,ambainis2019quantum,liu2019revisiting,don2019security}. %

\paragraph*{Auxiliary-input random oracle model.}
Consider the example where we use random oracle to instantiate an one-way function (OWF).
As shown by Hellman's algorithm~\cite{hellman1980cryptanalytic}, a preprocessing attacker can achieve a non-trivial saving of resources from $N$ to $N^{2/3}$ even in the random oracle model.
Because the instantiated function is usually public, there is no way of preventing the adversary from performing a heavy precomputation to speed up the online attack.

To address this mismatch, Unruh~\cite{C:Unruh07} introduced the auxiliary-input random oracle model (AI-ROM) where the adversary is allowed to obtain a bounded length advice about the random oracle before attacking the system.
Several works~\cite{C:Unruh07,EC:DodGuoKat17,EC:CDGS18,coretti2018non} have developed various techniques for analyzing the security in this model.

Hhan et al.~\cite{hxy19} first considered the auxiliary-input quantum random oracle model (AI-QROM) and quantum auxiliary-input quantum random oracle model (QAI-QROM), which is the natural generalization of the model above for quantum adversaries.
Despite the progress in QROM and AI-ROM respectively, proving post-quantum non-uniform security remains quite challenging.
Only a few security bounds have been proven against classical advice, i.e. under AI-QROM~\cite{hxy19}, and {\em no} security bounds are known against quantum advice, i.e. QAI-QROM, except for OWFs~\cite{hxy19,chung2019lower}.
The only known technique under these models is the (quantum) compression technique, introduced by Nayebi et al.~\cite{NABT15}, and its variants.
We note that this technique is only somewhat generic, and its effectiveness seems %
limited under QAI-QROM.

\paragraph*{New generic framework for security bounds in (Q)AI-QROM.}
We generalize our techniques for function inversion and Yao's box into a general framework for proving concrete post-quantum non-uniform security in AI-QROM and QAI-QROM.
In particular, we show the following theorem.

\begin{theorem*}[Informal]
  For any security game $G$, adversary time bound $T$, and any $g > 0$, consider its multi-instance game $\Gmi g$, which requires the adversary to break $g$ independent challenges \emph{sequentially}, and each instance is given time $T$.
  
  If the best winning probability for $\Gmi g$ is $\delta^g$ in the QROM, meaning that the best adversary can only win the game with probability at most $\delta^g$, then for adversaries with $S = g$ (qu)bits of advice:
  \begin{enumerate}
      \item $G$ is roughly $\tilde O(\delta)$-secure in AI-QROM;
      \item If $G$ is publicly verifiable, $G$ is $\negl(N)$-secure in QAI-QROM if $\delta = \negl(N)$;
      \item If $G$ is a decision game, $G$ is $1/2 + \negl(N)$-secure in QAI-QROM if $\delta = 1/2 + \negl(N)$.
  \end{enumerate}
\end{theorem*}

\ifieeecs\else
\noindent
The formal definition of security games and multi-instance games are given in \Cref{sec:security-games}, and the formal reduction is presented in \Cref{sec:advice2mis-reductions}.
\fi

We note that our bound for function inversion we show in \Cref{sec:intro-owf} directly translates to the concrete security bound for OWFs, and our techniques for OWFs can be additionally used to show security for pseudorandom generators (PRGs) with little efforts.
The concrete bounds,
\ifieeecs\else
proven in \Cref{thm:owfaiqrom}, \Cref{thm:owfqaiqrom}, and \Cref{thm:prg},
\fi
are summarized into the table below.
\begin{table}[H]
  \centering
  \begin{tabular}{ | l | c | c | }
    \hline
    & OWFs & PRGs \\ \hline
    AI-ROM\ifieeecs\else~\cite{EC:DodGuoKat17,EC:CDGS18}\fi   & $\frac{S T}{\alpha}$ & $\left(\frac{S T}{N}\right)^{1/2} + \frac{T}{N}$ \\ 
    AI-QROM~\cite{hxy19}   & $\left(\frac{S T^2}{\alpha}\right)^{1/2}$ & $\left( \frac{S T^4}{N} + \frac{T^4}{N} \right)^{1/6}$ \\ 
    AI-QROM (ours)  & $\frac{S T}{\alpha} + \frac{T^2}\alpha$ & $\left(\frac{S T}{N} + \frac{T^2}N\right)^{1/3}$ \\
    QAI-QROM (ours)  & $\left(\frac{S T}{\alpha} + \frac{T^2}\alpha\right)^{1/3}$ & $\left(\frac{S^5 T}{N} + \frac{S^4 T^2}N\right)^{1/19}$ \\
    \hline
  \end{tabular}
  \caption{Asymptotic security bounds on the security of OWFs and PRGs constructed from a random function $f:[N]\mapsto[M]$ against $(S, T)$-algorithms, where $\alpha := \min\{N, M\}$.}
\end{table}

\paragraph{Salting defeats preprocessing in the quantum setting.}
Instead of proving more concrete security bounds using the framework, we show that salting, a common mechanism used in practice for defeating auxiliary input, generically extends the security of applications proven in the QROM to the (Q)AI-QROM.
A similar statement was first shown to hold in the classical world by Coretti et al.~\cite{EC:CDGS18}, where they showed that the security in the ROM can be generically extended into salted security in the AI-ROM. %

In this work, we prove the hardness of salted multi-instance game under QROM.
\begin{lemma*}[Informal\ifieeecs\else statement of \Cref{lem:saltingmis}\fi]
  For any security game $G$ with security $\delta$ in the QROM against adversary running in time $T$, let $G_S$ denote its salted version with salt space $[K]$.
  Then the best winning probability for multi-instance salted game $\Gmi{g}_S$ is at most $(\delta + gT/K)^g$, which is tight\ifieeecs\else (\Cref{sec:salting-tight})\fi.
\end{lemma*}
\noindent
Combining this lemma with our reduction theorem above, we conclude that salting generically defeats preprocessing in the AI-QROM, and defeats preprocessing against publicly-verifiable and decision games in the QAI-QROM.
We believe our techniques are sufficient and both the reduction theorem and salted multi-instance bound can be generalized for proving quantum non-uniform salted security bounds under other idealized cryptographic models, e.g. the ideal-cipher model (which is by definition a salted permutation family), and the salted generic group model.
However, in order to simplify presentation, in this work we only focus on the ROM. 

Using this generic bootstrapping theorem, we can also easily obtain security bounds for a wide class of salted cryptographic primitives in the AI-QROM and QAI-QROM.
In particular, we use it to give the first bound for salted collision-resistant hash (CRH) in AI-QROM and QAI-QROM, resolving an open problem raised by Hhan et al.~\cite{hxy19}.
This in turn implies that $\mathsf{PWPP}^{\mathcal O} \not\subseteq \mathsf{FBQP}^{\mathcal O}\mathsf{/qpoly}$ relative to a random oracle $\mathcal O$\ifieeecs\else (\Cref{thm:oracle-separation})\fi.

Finally, we note that our proof can also be naturally downgraded to a classical reduction, which gives a new proof of the classical result first proven by Coretti et al.~\cite{EC:CDGS18}

\subsection{Technical Overview}

\ifieeecs
In this section, we give an overview of our techniques for proving the results.
Please refer to the full version for the formal proof.
\fi

\subsubsection{From Non-uniform Algorithms to Multi-Instance Games}
\label{sec:tech-overview-reduction}

We start by considering applying a general approach in the literature for proving lower bounds for non-uniform algorithms in the context of classical function inversion, and then generalize it to quantum non-uniformity.
In particular, the following argument is based on the work of Aaronson~\cite{Aaronson05}.

Let $f: [N] \mapsto [N]$ be a function.
Consider a classical (randomized) algorithm $\As$ with $S$-bit of advice $\alpha = \alpha(f)$, $T$ queries to $f$ (which is a lower bound on its running time), that can invert a random image for \emph{any} function\footnote{Since we are requiring the algorithm to invert any function, the lower bound we present here is slightly weaker. We intentionally make this omission for the overview to highlight the more important ideas in our proofs. Interested readers should refer to the formal proofs.} $f$ with probability at least $\delta$, i.e.
$$\delta := \Pr_{\As,x}[ \As^f(\alpha, f(x)) \in f^{-1}(f(x))],$$
where the probability is taken over the measurement randomness of the algorithm $\As$ and the random choice of $x$, and we want to give an upper bound on $\delta$.
At a high level, the approach is to reduce it to proving lower bound for the multi-instance version of the problem for algorithms \emph{without} advice.
Specifically, we reduce it to bound the success probability of an algorithm $\Bs$ with $gT$ queries to invert random $g$ inputs $f(x_1), \dots, f(x_g)$ simultaneously, for some parameter $g \in [N]$ on the number of instances.
For function inversion, it is easy to show that for any $g > 0$, the best success probability drops exponentially fast in $g$.
Specifically, for any (randomized) algorithm $\Bs$ and random functions $f$,
\ifieeecs
  \newcommand{\mathieeecs}[1]{#1}
\else
  \newcommand{\mathieeecs}[1]{}
\fi
\begin{equation}
  \label{eq:tech-overview-upper-bound}
  \begin{split}
    \Pr_{\Bs,f,x_1,\dots,x_g}[&\Bs^f(f(x_1),\dots, f(x_g)) \mathieeecs{\\
    &}\mbox{ inverts } f(x_1),\dots, f(x_g)] \leq O(gT/N)^g.
  \end{split}
\end{equation}

The reduction proceeds in two simple steps. We first use $\As$ to construct an algorithm $\Bs'$ using \emph{only one copy of the original advice} for the multi-instance problem (with decent success probability), and then get rid of the advice by simply guessing a uniformly random bitstring.
The algorithm $\Bs'^f(\alpha, f(x_1),\dots,f(x_g))$ simply invokes $\As$ to invert each $f(x_i)$, and succeeds when $\As^f(f(x_i))$ succeeds on all $g$ instances.
By independence, it is easy to see that
\begin{equation*}
  \begin{split}
    \Pr_{\Bs',x_1,\dots,x_g}[&\Bs'^f(\alpha,f(x_1),\dots, f(x_g)) \mathieeecs{\\
    &}\mbox{ inverts } f(x_1),\dots, f(x_g)] \ge \delta^g.
  \end{split}
\end{equation*}

Next we remove the advice by guessing. Consider an algorithm $\Bs$ first guesses a \emph{random} advice $\alpha \in \{0,1\}^S$ and then runs $\Bs'$. Clearly, $\Bs$ guesses the advice correctly with probability $2^{-S}$, in which case $\Bs$ simulates $\Bs'$ perfectly. Hence,
\begin{equation*}
  \begin{split}
    \Pr_{\Bs,x_1,\dots,x_g}[&\Bs^f(f(x_1),\dots, f(x_g)) \mathieeecs{\\
    &}\mbox{ inverts } f(x_1),\dots, f(x_g)] \geq 2^{-S} \delta^g.
  \end{split}
\end{equation*}
As this statement is for any function, the same conclusion holds for random functions.
Combining this with the above upper bound~\eqref{eq:tech-overview-upper-bound} on the success probability of $\Bs$ with $g = S$ shows that $\delta \leq O(ST/N)$, which matches the best known classical bound~\cite{yao1990coherent, EC:DodGuoKat17, EC:CDGS18}.

It should be clear from the above example that this approach is fairly general and reduces the non-uniform lower bound problems to analyze the success probability of the corresponding multi-instance games. Indeed, this approach and its variants have implicitly appeared in various contexts under the name of direct product theorems. We discuss this in more details in \Cref{sec:related-works}.

\paragraph*{First attempt using multi-instance problems.}
Seeing this as a promising start, we now consider the setting of quantum algorithms with \emph{classical} advice.
Indeed, the argument works out similarly, except that now we need to consider the success probability of the best \emph{quantum} algorithm that queries $gT$ locations and succeed inverting $g$ independently random images.
We could hope that analyzing the best success probability for quantum algorithms solving the multi-instance problem would lead us to the desired bound.

Unfortunately, it turns out that this approach is destined to fail to achieve any $\delta \ll ST^2/N$.
With $q$ quantum queries, the famous Grover's search can find one element in $\eta$ fraction of all elements  with probability $\approx \eta q^2$.
Consider the following algorithm $\tilde\Bs$, where it tries to find one pre-image among all the $g$ images using the first $T$ queries, which succeeds with probability $gT^2/N$. If it succeeds in finding the first pre-image, it then use the next $T$ queries to  find one pre-image among all remaining $(g-1)$ images, which succeeds with probability $(g-1) T^2 / N$ and so on. 
Using this algorithm, for any function, we can find all $g$ pre-images with probability at least roughly 
\begin{align*}
  (g T^2 / N) \cdot ((g-1) T^2/N) \cdot ... \cdot (T^2 / N) &\approx (T^2/N)^g \cdot g! \mathieeecs{\\
    &}\approx (gT^2/N)^g. 
\end{align*}
This implies that the best bound we can hope to achieve for the multi-instance problem would be $O(gT^2/N)^g$, which would in turn imply the bound $O(ST^2/N)$, which is the same bound as what we already had before. 

\paragraph*{Bypassing the barrier via multi-instance games.}
A natural question arises that whether we can go beyond $ST^2/N$.
We claim that this is actually the case.

To see this point, we first recall the high level ideas of the argument above -- we first bootstrap the best algorithm $\As$ with advice $\alpha$ for computing $f^{-1}$ with success probability $\delta$, into a multi-instance algorithm $\Bs$ with advice $\alpha$ with success probability $\delta^g$, and then remove the advice and incur a loss of $2^{-S}$, which we amortize into $\delta^g$.
The problem essentially reduces to proving a lower bound for the success probability of the resulting algorithm $\Bs$ that solves the multi-instance problem.

While it may seem like that the $(gT^2/N)^g$ algorithm above could be an upper bound for the new problem, we observe that this iterated Grover's search algorithm $\tilde\Bs$ actually never arises from our reduction from $\As$ to $\Bs$.
In particular, $\Bs$ always solves each instance one by one, while $\tilde\Bs$ in some sense solves \emph{all} $g$ instances at once.
To view this issue in terms of quantum queries, the first $T$ queries by $\Bs$ are only searching pre-image of $f(x_1)$, but the first $T$ queries by $\tilde\Bs$ are searching for all $g$ pre-images.

We formalize this intuition by strengthening the multi-instance problem into what we call a ``multi-instance game,'' where the algorithm (or adversary) is instead \emph{interacting} with a verifier (or challenger).
For each round $i \in [g]$, the challenger samples a new image $f(x_i)$, and the adversary is given $T$ queries to $f$, before producing an output $x'_i$, which the challenger checks whether $f(x_i) = f(x'_i)$.
The main change we make to the multi-instance \emph{games}, compared with multi-instance \emph{problems}, is that the adversary gets challenges one-by-one or \emph{sequentially}, rather than getting all challenges at once or in \emph{parallel}.

Observe that this multi-instance game seems to rule out the algorithm $\tilde\Bs$ above, as the adversary does not have any information about $f(x_i)$ until he has issued $(i - 1)T$ queries.

If we assume that the probability that any quantum adversary wins such multi-instance game for function inversion is $O\left(\frac{gT + T^2}N\right)^g$, using the same reduction as before, we would reach the conclusion that the best success probability for function inversion with $S$ bits of classical advice and $T$ quantum queries is $O(ST + T^2)/N$.
It turns out that we can indeed prove this assumption, but we will defer the discussion and consider quantum advice first.

\paragraph*{Beyond classical preprocessing.}
The reduction above requires the algorithm to solve multiple instances using only a single copy of the advice, which is problematic in the quantum setting due to no-cloning theorem.
We resolve this problem by constructing $\Bs$ similarly as before, and add gentle measurement to solve multiple instances.
To fill in the details, $\Bs$ does the following.
\begin{enumerate}
    \item Prepare $k = \Theta(\log g)$ copies of the quantum advice $\beta := \alpha^{\otimes k}$.
    \item Boost $\As$'s success probability from constant to $1 - o(1)$, by running $\As$ on each copy of the advice $\alpha$ $k$ times, and identify the correct answer using one additional query.
    \item To solve $g$ instances simultaneously, $\Bs$ simply runs boosted $\As$ for each instance, and applies measurement.
        As we have boosted the success probability high enough, the measurement will be ``gentle'' and we can recover an almost-as-good-as-new quantum advice for the next instance.
    \item Finally, to remove the quantum advice, we replace $\beta$ with a maximally mixed state, which gives us a multiplicative loss of $2^{-Sk}$ in success probability.
\end{enumerate}

However, for function inversion, this idea seems to fail as function inversion problem can have non-unique correct answers.
In particular, if $f^{-1}(f(x)) = \{x_1, x_2, ..., x_\ell\}$, and the algorithm somehow prepares the answer $\ket{x'_i}$ that is a uniform superposition over all the answers $x_1, ..., x_\ell$, while it succeeds with probability 1, performing a gentle measurement on this answer seems very difficult.
This difficulty of performing gentle measurements for function inversion was also acknowledged by the work of Hhan et al.~\cite{hxy19}, although under a different context.

We claim that using multi-instance \emph{games} (as opposed to multi-instance problems), there is actually a very elegant solution to this.
Instead of having the adversary submitting a classical answer $x'_i$, we will allow the adversary to submit a \emph{quantum state} $\ket{x'_i}$, and the challenger can compute in superposition whether $f(\ket{x'_i}) = f(x_i)$, and measure her decision (which will be gentle, as we boosted its success probability), and send back $\ket{x'_i}$.
On a high level, the idea is basically having the adversary and the challenger ``jointly'' perform this gentle measurement.
As the adversary cannot control challenger's behavior, this change should not impact the best winning probability of the multi-instance game.

\ifieeecs\else
We formally define multi-instance games in \Cref{sec:mis} and highlight the difference from multi-instance problems in more details there at \Cref{rem:comparison-mis-vs-dpt}.
\fi

\subsubsection{Analyzing Multi-Instance Game via ``Compressed Oracles''}

To complete our time-space tradeoff for function inversion, the only remaining step is to bound the best winning probability of the multi-instance game for function inversion.
We extends the techniques of Zhandry's compressed oracles~\cite{zhandry2019record} and combine with a new indistinguishability lemma to give a tight bound for this problem. 

\paragraph*{Compressed oracles.}
In the classical setting, there is a commonly used technique for arguing random functions called the lazy sampling of a random oracle.
The idea is that a simulator will maintain a partial truth table about the random function.  Upon an oracle query $x$, the simulator looks up $x$ in the table $D$ and returns it as the answer. If not found, the simulator freshly samples a new $y$ as the output of $x$ and inserts the pair $(x, y)$ into the table $D$. %

Zhandry observes that, if care is taken to implement the oracle correctly, a quantum analogy of the classical on-the-fly simulation is possible. Unlike the classical simulation, they simulate a random oracle as a superposition of tables, each of which partially instantiates a random function. Below is the high level idea of their simulation. The table is initialized as an empty table. Upon a quantum query is made by an algorithm, the simulator updates the database in superposition: for a query $\ket x$ and a table $\ket D$, the simulator look up $x$ in the table $D$; if not found, it initializes a superposition of all possible output $\sum_y \ket y$ (up to a normalization) as the value of $D(x)$ and updates $D$ in superposition to get $\ket {D \cup (x, y)}$; it then returns $D(x)$ as the output.

A major difference between quantum setting and classical setting is an algorithm may forget some query it made before. As an example, an algorithm can query the same input twice to un-compute everything and thus completely lose the information about the output. Therefore, to perfectly simulate a quantum random oracle, the simulator also checks after every query that if the algorithm loses all information about the query. With all these operations above, Zhandry shows a quantum random oracle can be efficiently simulated on-the-fly.

\paragraph*{Analyzing multi-instance game.}
To prove the success probability of multi-instance function inversion, we consider a stronger statement regarding the success probability of inverting for any round: assume an algorithm already makes $(i-1) T$ queries for the first $(i-1)$ rounds, and conditioned on it having passed the first $(i-1)$ rounds, what is the probability of succeeding in the $i$-th round by making $T$ queries? 
\begin{lemma}%
    For any quantum algorithm making $q_0 + q$ queries to a random function $f:[N] \to [N]$, if $f(x)$ is sampled and given after the $q_0$-th query, conditioned on arbitrary outcomes (with non-zero probability) of the algorithm's measurement during the first $q_0$ queries, the probability of inverting $f(x)$ is at most $O((q_0 + q^2)/N)$. 
\end{lemma}

We consider the lemma above as remarkable and perhaps even surprising, as intuitively, it is saying that quantum power can achieve a quadratic speed up for search \emph{only} if you know what you are looking for, and there is no classical analogue of this.

With the above lemma, the probability of succeeding in the $i$-th round is at most $O((i T + T^2)/N)$.   
Therefore, the probability of succeeding in inverting all random images is at most $O((g T + T^2)/N)^g$.

To prove this lemma, let us start by assuming that a uniformly random image $y$ is instead given as a challenge and without conditioning on the intermediate measurements.
By Zhandry's techniques, after making $q_0$ queries to a random oracle, the knowledge of an algorithm about the random oracle can be viewed as a superposition of some tables with at most $q_0$ entries which specifies  partial random functions over at most $q_0$ inputs.
Since $y$ is sampled uniformly, the amplitude (square root of probability) of database containing $y$ is about $\sqrt{q_0/N}$. After given the image $y$, each query can increase the amplitude by at most $\sqrt{1/N}$. Therefore, the final amplitude of of database containing $y$ is about $(\sqrt{q_0} + q)/\sqrt{N}$ which gives us the lemma above. 

There are two challenges we need to overcome for proving the full lemma.

The first challenge comes from the fact that our lemma statement requires an algorithm that is conditioned on some fixed measurement outcomes.
To address this, we extend Zhandry's techniques to such settings\ifieeecs\else in \Cref{sec:compressed-oracle} and complete the proofs in \Cref{appendix:compressed-oracle}\fi, which is a very natural but non-trivial extension and crucial for analyzing the multi-instance game.

The second challenge seems even more difficult.
For function inversion, the image is sampled by first sampling a random pre-image $x$ and computing $f(x)$.
The natural application of compressed oracle only gives us a probability bound when a random $y$ is sampled instead of $f(x)$.
This is a nontrivial issue as with high probability, the distribution of $f(x)$ is only supported on a constant fraction of $[N]$, so the statistical difference of the two distributions is significant.
The issue is more prominent when we consider the general random functions $f: [N] \mapsto [M]$ where $N \ll M$, where a uniform sample of $y$ is with high probability not an image of any $x$.

Towards this challenge, we prove the following indistinguishability lemma to bridge the gap. The proof of the lemma does not require the compressed oracle technique and we believe that both lemmas are of independent interest. 
\begin{lemma}[Indistinguishability]%
    For any quantum algorithm making $q_0 + q$ queries to a random function $f:[N] \to [N]$, if $f(x)$ or a uniformly  random $y$ is sampled and given after the $q_0$-th query, conditioned on arbitrary outcomes (with non-zero probability) of the algorithm's measurement during the first $q_0$ queries, the advantage of distinguishing is at most $O((q_0 + q^2)/N)$.
\end{lemma}
We can then assume a random $y$ is sampled instead of a random $f(x)$ with only an additive loss of $O((gT + T^2)/N)$ in each round, giving us the bound that we desire.

\paragraph*{Proving the indistinguishability lemma.}
We first convert the problem from distinguishing samples (random $f(x)$ or random $y$) into distinguishing oracles.
Assume the oracle is sampled as follows: first, a uniformly random input $x$ is sampled; then a random function $f_{-x}$ defined on all inputs except $x$ is sampled, together with two independently sampled $y_0, y_1$; define $f_{-x} || y$ as a  function that outputs $f_{-x}(x')$ on all inputs that are not $x$ and $y$ on input $x$; the distinguisher is ask to given either oracle access to $f_{-x} || y_0$ or $f_{-x} || y_1$ and the same challenge $y_0$ after the $q_0$-th query, distinguish which oracle is given (without knowing $x$). When the function $f_{-x}||y_0$ and challenge $y_0$ is given, it corresponds to the case a random $f(x)$ is given; while giving the second function corresponds to the second case that a random $y$ is given. 
It can be shown that the two problems have the exact same difficulty.

Intuitively, every quantum query to a function  entangles a quantum algorithm with one of the output  of that function in superposition. By making $q_0$ queries, the algorithm is entangled with at most $q_0$ outputs of the function in superposition. 
When $y_0$ is given, the only way to tell which oracle is given is by already making  a query on $x$ and entangling the algorithm with either $y_0$ or $y_1$ respectively. Since $y_0$ has not been given during the first $q_0$ queries,  $x$ is perfectly hidden and completely uniformly random from the algorithm's view. Thus, such entanglement only happens with probability $q_0/N$. 
For the remaining $q$ queries, knowing the information $y_0$ does help build the entanglement faster. One strategy is to use Grover's search to check if $y_0$ is an image of the function since with constant probability, $y_0$ is not an image of $f_{-x}||y_1$. By using Grover's search, the advantage of distinguish is about $q^2/N$ and we show such advantage is the best one can get for these $q$ queries. Combining with these two separate analysis, we conclude the indistinguishability lemma.

\subsubsection{Yao's Box Problem}
\label{sec:tech-overview-yaos-box}

We now focus our attention on Yao's box problem.
Assume that an algorithm $\As$, given any function $f: [N] \mapsto \{0, 1\}$, prepares an $S$-qubit advice $\alpha = \alpha(f)$ such that $\As$ can recover $f(x)$ for a random $x$ using $\alpha$ without ever querying $f(x)$ in time $T$ with probability $1/2 + \eps$.

We claim that for classical advice, our reduction from function inversion with advice to multi-instance game can also be generalized to Yao's box, with a more careful amortizing analysis; and the multi-instance game for Yao's box is fairly straightforward to argue using similar techniques as for function inversion.
Intuitively, for Yao's box, the only non-trivial strategy is that the adversary's first $(i - 1)T$ quantum queries predicted the challenge $x_i$, so the best advantage of any algorithm, i.e. the best winning probability minus $1/2$, can only be $O(\sqrt{gT/N})$ instead of $O(\sqrt{gT^2/N})$.
This leads us to the final bound $1/2 + \tilde O(ST/N)^{1/3}$, where the additional exponent loss comes from the new amortizing argument.

For quantum advice, things are a lot trickier.
While Aaronson~\cite{Aaronson05} proved a similar lower bound for a different problem against quantum advice, the techniques there only allow us to prove lower bounds against algorithms that find the correct answer with probability at least $2/3$ for \emph{all} $f$ and $x$ (and indeed under this setting, our techniques combined with \cite{Aaronson05} are sufficient to give query lower bound $ST \ge \tilde\Omega(N)$ even for quantum advice).
However, this is insufficient in our settings where we need to consider the stronger lower bound where $x$ is sampled randomly, and the algorithm can only predict some $f(x)$ and output a random guess for others.

We first revisit the idea of Aaronson~\cite{Aaronson05}, which is to prepare $\poly(\log g)$ copies of the advice, and use majority vote\footnote{This idea was implicitly given, where they called it ``boosting'' under the context of randomized algorithms in complexity theory.} to boost the success probability from $2/3$ to $1 - o(1)$ to make the measurement gentle.
We show that majority vote cannot possibly boost success probability under the average-case (over $x$) setting, by considering the following example: the algorithm $\As$ outputs the correct answer with probability $1$ on 40\% of the inputs, and with probability $0.45$ on the other 60\% of the inputs.
Overall, the success probability is 67\%, but with majority vote, the success probability goes down, and can go down arbitrarily close to 40\%, which is much worse than random guessing!

One way to resolve this is to instead ``gently'' measure the success probability of $\As$ for each instance $f(x)$, and throw a biased random coin according to this distribution as our answer.
We observe that the problem of ``gently'' measuring this probability can be reduced to shadow tomography, which is a problem introduced by Aaronson~\cite{aaronson2018shadow}.
As our multi-instance game requires that each challenge is given sequentially, we also require the shadow tomography to be able to handle online queries.
Aaronson and Rothblum~\cite{AR19} showed that online shadow tomography indeed can be done using $S^2 \log^2 g$ copies of the advice, which leads us to the final bound $S^5 T = \tilde\Omega(N)$ for constant $\eps > 0$.

\begin{remark}[Improving the lower bound for Yao's box with quantum advice]
    \label{rem:yao-box-worst-case}
    We note that if we make the restriction so that for most functions $f$, $\As$ succeeds to output $f(x)$ for \emph{every} input $x$ with probability at least $2/3$, then the idea of using majority vote still works, since in this case, boosting will give us the correct outcome with overwhelming probability.
    Using the same reduction, we can prove that $\As$ has to satisfy $ST = \tilde\Omega(N)$.
    Therefore, we think that $ST = \tilde\Omega(N)$ should be the optimal lower bound.
\end{remark}

\subsection{Related Works}
\label{sec:related-works}

In this section, we compare our techniques with other related works.

The approach of reducing time-space tradeoff lower bounds to multi-instance \emph{problems}, which is outlined in \Cref{sec:tech-overview-reduction}, has appeared implicitly in various works~\cite{Beame91, klauck2003quantum, KlauckSW07}, where they usually refer to the (exponential) hardness of multi-instance problems as ``(strong) direct product theorems.''
While the different approaches presented in different works share some similar high-level ideas, the context and details in each work are slightly different.
In this work, to avoid confusion, we use the term ``multi-instance problem'' instead of direct products.
Recently Hamoudi and Magniez~\cite{hamoudi2020quantum} independently applied the similar technique along with Zhandry's compressed oracles to prove quantum time-space tradeoffs for finding multiple collisions.  

Classically, this approach has also been considered in various works~\cite{goldreich1993existence, impagliazzo2011relativized} for proving non-uniform lower bounds.
Aaronson~\cite{Aaronson05} first showed how to employ such ideas when the non-uniform lower bounds need to hold even against quantum advice.
While the problem they consider is quite simple and somewhat arbitrary, the starting point we outline in \Cref{sec:tech-overview-reduction} is based on this work.
However, to the best of our knowledge this technique has not been explored in the AI-ROM literature (possibly due to the fact that same bounds can already be achieved using other -- possibly more complicated -- techniques), and the reductions presented in this work can also be easily ``dequantized'' to the classical setting.

As far as we are concerned, our work is the first one to consider the stronger variant ``multi-instance games'' and show a separation of the two variants for function inversion in \Cref{sec:tech-overview-reduction}.
Additionally, we present our reduction under a general framework for the stronger variant of multi-instance games.

The idea of using gentle measurements is almost ubiquitous for proving lower bounds against quantum advice.
To the best of our knowledge, Aaronson~\cite{Aaronson05} first showed how to combine boosting and gentle measurements for quantum advice lower bounds, which we briefly discuss in \Cref{sec:tech-overview-yaos-box}.
In particular, this technique was also employed by Hhan et al.~\cite{hxy19} to prove an asymptotic lower bound of $ST^2 \ge N$ for inverting random permutations, although under a different context.

\subsection{Open Problems}
\label{sec:open-problems}

\paragraph*{Quantum time-space tradeoff lower bounds for permutation inversion.}
While our work provides substantial evidence that the quantum time-space tradeoff bound for inverting permutations is $ST + T^2 = \tilde\Theta(N)$, which would have resolved the open problem posted by Nayebi et al.~\cite[Section 4.3]{NABT15}, we are not able to formally prove this due to a lack of ``compressed permutation oracles''.
This is especially interesting, considering that it is easier to argue about permutations than functions using the compression argument, which is used in prior works~\cite{NABT15,hxy19, chung2019lower}.

\paragraph*{Nontrivial quantum speed ups for function inversion.}
While our lower bound is best possible, it still leaves open the possibility that some nontrivial quantum speed ups exist under the following asymptotic regime:
\begin{equation*}
    \left\{
        \begin{aligned}
            ST &\gg N, \\
            T^2 &\ll N, \\
            S^2 T &\ll N^2.
        \end{aligned}
    \right.
\end{equation*}
An especially interesting case is that there might be a quantum algorithm with advice under this regime, but it might seem extremely hard to be ``dequantized.''

Our paper also gives many lower bounds, improving any of the bounds that are not tight (for example, the query bound for average-case Yao's box against quantum advice), or showing new non-trivial attacks, are all interesting possibilities.

\paragraph*{Suboptimal exponent on success probability for quantum advice, or is it?}
All the bounds we have achieved for quantum advice are not tight in terms of the exponent on the success probability for quantum advice.
In particular, we note that for function inversion, our classical advice lower bound has exponent $1$, which is tight -- while the quantum advice bound has exponent $1/3$.
Can this loss be avoided, or is there any speed up in terms of $S$ and $T$ for sub-constant success probability?

We make the observation that the loss in exponent ultimately comes from the use of gentle measurements.
Looking back at the literature, all the quantum advice lower bound techniques~\cite{Aaronson05,hxy19,chung2019lower} we have seen so far always require ``reusing'' of the advice, which in turns require gentle measurements.
Is there a way to avoid reusing the quantum advice to escape this cost?
We suspect that this is the case, as in the work of Chung et al.~\cite{chung2019lower}, they presented a reduction to quantum random access code, which by definition seems to avoid this issue, albeit in the end, the advice reusing issue somehow kicks back in.

\fi

\ifexabs\else
\ifieeecs\else

\section{Preliminaries}
For any $n \in \mathbb N$, we denote $[n]$ to be the set $\{1, 2, ..., n\}$.
We denote $\Z/n\Z = \{0, 1, ..., n - 1\}$ as the ring of integer modulo $n$, and $\mathbb F_2 = \{0, 1\}$ as the binary finite field.
For a complex vector $\vx \in \mathbb C^n$, we denote the $L^2$-norm $|\vx| = \ltwonorm{\vx} = \sum_{i \in [n]} x_i \overline{x_i}$.
In algorithms, we denote $a \gets_\$ A$ to be taking $a$ as a uniformly (fresh) random element of $A$.

Next, we recall some basic facts about quantum computation, and review the relevant literature on the quantum random oracle model.

\subsection{Quantum Computation}
A quantum system $Q$ is defined over a finite set $B$ of classical states.
A \textbf{pure state} over $Q$ is a unit vector in $\mathbb{C}^{|B|}$, which assigns a complex number to each element in $B$. In other words, let $|\phi\rangle$ be a pure state in  $Q$, we can write $|\phi\rangle$ as a column vector:
\begin{equation*}
    |\phi\rangle = \sum_{x \in B} \alpha_x |x\rangle
\end{equation*}
where $\sum_{x \in B} |\alpha_x|^2 = 1$ and $\{|x\rangle\}_{x \in B}$ is called 
the ``\emph{computational basis}'' of $\mathbb{C}^{|B|}$. The computational basis forms an orthonormal basis of $\mathbb{C}^{|B|}$. 
We define $\bra \phi$ to be the row vector that is the conjugate of $\ket \phi$.

Given two quantum systems $Q_1$ over $B_1$ and $Q_2$ over $B_2$, we can define a \emph{product} quantum system $Q_1 \otimes Q_2$ over the set $B_1 \times B_2$. Given $|\phi_1\rangle \in Q_1$ and $|\phi_2\rangle \in Q_2$, we can define the product state $|\phi_1\rangle \otimes |\phi_2\rangle \in Q_1 \otimes Q_2$. 

We say $|\phi\rangle \in Q_1 \otimes Q_2$ is \emph{entangled} if there does not exist 
$|\phi_1\rangle \in Q_1$ and $|\phi_2\rangle \in Q_2$ such that $|\phi\rangle = |\phi_1\rangle \otimes |\phi_2\rangle$. For example, consider $B_1 = B_2 = \{0,1\}$
and $Q_1 = Q_2 = \mathbb{C}^2$, $|\phi\rangle = \frac{|00\rangle + |11\rangle}{\sqrt{2}}$ is entangled. Otherwise, we say $|\phi\rangle$ is unentangled. 

A pure state $|\phi\rangle \in Q$ can be manipulated by a unitary operator $U \in \mathbb C^{|B| \times |B|}$. The resulting state $|\phi'\rangle = U |\phi\rangle$. 
We denote the trace norm $\trdist{U}$ to be $\frac 1 2 \Tr\sqrt{U^\dagger U}$.

We extract classical information from a quantum state $|\phi\rangle$ by performing a \emph{measurement}. A measurement is specified by an orthonormal basis, typically the computational basis, and the probability of getting result $x$ is $|\langle x | \phi \rangle|^2$. After the measurement, $|\phi\rangle$ ``collapses'' to the state $|x\rangle$ if the result is $x$. 
                
For example, given the pure state $|\phi\rangle = \frac{3}{5} |0\rangle + \frac{4}{5} |1\rangle$ measured under $\{|0\rangle ,|1\rangle \}$, with probability $9/25$ the result is $0$ and $|\phi\rangle$ collapses to $|0\rangle$; with probability $16/25$ the result is $1$ and $|\phi\rangle$ collapses to $|1\rangle$.

    We assume quantum circuits can  implement any unitary transformation (by using these basic gates, Hadamard, phase, CNOT and $\frac{\pi}{8}$ gates), in particular the following two unitary transformations:
        \begin{itemize}
            \item \textbf{Classical Computation:} Given a function $f : X \to Y$, one can implement a unitary $U_f$ over $\mathbb{C}^{|X|\cdot |Y|} \to \mathbb{C}^{|X| \cdot |Y|}$ such that for any $|\phi\rangle = \sum_{x \in X, y \in Y} \alpha_{x, y} |x, y\rangle$, 
            \begin{equation*}
                U_f |\phi\rangle = \sum_{x \in X, y \in Y} \alpha_{x, y} |x, y \oplus f(x)\rangle
            \end{equation*}
            
            Here, $\oplus$ is a commutative group operation defined over $Y$.
            In particular, if $f$ is given as a classical circuit $C$, there exists an efficient implementation of the unitary $U_f$ using $|C|$ ancillas, and each gate is evaluated at most twice.
            
            \item \textbf{Quantum Fourier Transform:} For every $n \in \mathbb N$, the quantum Fourier transform ${\sf QFT}_n$ is a unitary operation, that is given a quantum state $|\phi\rangle = \sum_{j \in \Z/n\Z} x_j |j\rangle$, %
	outputs $|\psi\rangle =  \sum_{k \in \Z/n\Z} y_k |k\rangle$ where the sequence $\{y_k\}_{k}$ is the Fourier transform to the sequence $\{x_j\}_{j}$, i.e.
		\begin{equation*}
			y_k = \frac{1}{\sqrt{n}} \sum_{j \in \Z/n\Z} \omega_n^{jk} x_j
		\end{equation*} 
	where $\omega_n = e^{2 \pi i / n}$, and $i = \sqrt{-1}$ is the imaginary unit. 
	
        \end{itemize}

\subsection{Quantum Random Oracle Model}

An oracle-aided quantum algorithm can perform quantum computation as well as quantum oracle query. 
A quantum oracle query for an oracle $f:[N] \to [M]$ is modeled as a unitary $U_f: \ket x \ket u = \ket x \ket {u + f(x)}$, where $+$ denotes addition in integer ring $\Z/M\Z$ (we take the natural bijection that $M \simeq 0$, but any bijection $[M] \leftrightarrow \Z/M\Z$ suffice for our purposes). 

A random oracle is a random function $H : [N] \to [M]$. The random function $H$ is chosen at the beginning. A quantum algorithm making $T$ oracle queries to $H$ can be modeled as the following: it has three registers $\ket x, \ket u, \ket z$, where $x \in [N], u \in \Z/M\Z$ and $z$ is the algorithm's internal working memory; it starts with some input state $|0\rangle \ket 0 |\psi\rangle$, then it applies a sequence of unitary to the state: $U_0$, $U_H$, $U_1$, $U_H$, $\cdots$, $U_{T-1}$, $U_H$, $U_T$ and a final measurement over computational basis. Each $U_H$ is the quantum oracle query unitary  $U_H: \ket x \ket u = \ket x \ket {u + H(x)}$ and $U_i$ is the local quantum computation that is independent of $H$. We can always assume there is only one measurement which is a measurement over computational basis and applied at the last step of the algorithm. 

\subsection{Compressed Oracle}
\label{sec:compressed-oracle}

In this subsection, we recall the technique introduced by Zhandry \cite{zhandry2019record}. We will explain how to purify a random oracle in the quantum setting first, and then give equivalent forms of a quantum random oracle, namely standard oracle ${\sf StO}$, phase oracle ${\sf PhO}$ and compressed standard oracle $\csto$.  All these oracles are equivalent in the sense that for every (even unbounded) algorithm making queries to one of these oracles, the output distribution of the algorithm is exactly identical regardless of which oracle is given. 
Then Zhandry shows dealing with compressed standard oracle is usually easier. 
Roughly speaking, Zhandry shows that with compressed standard oracle, one could quantify the amount of the information about the random oracle learned by any quantum algorithm, analogous to the lazy sampling technique that is very commonly used for classical random oracles.

Note that in Zhandry's work \cite{zhandry2019record}, they originally only considered output of size $M = 2^m$, and implementing a quantum random oracle as $U_H : \ket x \ket u = \ket x \ket {u \oplus H(x)}$ where $\oplus$ is bit-wise XOR, or the addition over $\mathbb F_2^m$. Therefore, their description of the compressed oracle technique is different since the range is defined as $\mathbb{F}_2^{m}$ instead of $\Z/M\Z$ considered in this paper.
Two oracles are equivalent as we can simulate one with the other using two queries.
For the completeness of the paper, we will reprove some of the useful lemmas under the integer ring $\Z/M\Z$.

Also note that Zhandry also showed that a compressed oracle can be efficiently implemented by a quantum computer, i.e. the running time is only polynomial in the number of queries and $\log N, \log M$.
In this work, since we mainly consider query complexity and for presentation, we ignore the issue of efficiency for a simpler presentation.

\paragraph{Purification: standard oracle.} Let $H$ be a random oracle $[N] \to [M]$. 
The function $H$ is sampled at the very beginning, or equivalently, initially we prepare a maximally mixed state $\eta \sum_H \ket H \bra H$ up to some normalization factor $\eta$, and each query can be implemented by another unitary $U$, which reads the function $H$ and applies $U_H$.
However, we can ``purify'' the random oracle, meaning that we can replace the mixed state of $\ket {H}$ with a uniform superposition of all possible functions, i.e. $\sqrt\eta \sum_H \ket H$.
Consider the truth table of $H$, that is $\ket H = \ket {H(1)} \ket {H(2)} \cdots \ket {H(N)}$. Let $\As$ be any quantum algorithm. We say the algorithm can query the standard oracle if we treat the algorithm's registers and $\ket H$ as a whole system, initialized as $\ket 0 \ket \psi \otimes \frac{1}{M^{N/2}} \sum_H \ket H$. An oracle query ${\sf StO}$ in this purified state is defined as, 
\begin{align*}
    {\sf StO} \ket {x} \ket {u} \ket {z} \otimes \ket {H} = \ket {x} \ket {u + H(x)} \ket {z} \otimes \ket {H},  
\end{align*}
where $\ket x, \ket u$ are the input and output register, $\ket z$ is an arbitrary working register and $\ket H$ is the random oracle. Each local quantum computation is $U_i \otimes I$ which only operates on $\As$'s registers $\ket x \ket u \ket z$. 
Therefore the computation of any $\As$ can be described as a sequence of:  $U_0 \otimes I$, $\sto$, $\cdots$, $U_{T-1} \otimes I$, $\sto$,  $U_T \otimes I$, and a final computational measurement over $\As$'s register. The following proposition tells that the output distribution using a standard oracle is exactly the same as using a random oracle. 

\begin{lemma}[{\cite[Lemma 2]{zhandry2019record}}] \label{prop:purifyoracle}
    Let $\As$ be an (unbounded) quantum algorithm making oracle queries. The output of $\As$ given a random function $H$ is exactly identical to the output of $\As$ given access to a standard oracle. 
    Therefore, a random oracle with quantum query access can be perfectly simulated as a standard oracle. 
\end{lemma}

\paragraph{Phase kickback: phase oracle.} %
Define a unitary $V$ as $(I_x \otimes {\sf QFT}_M^\dagger \otimes I_{H})$ which applies ${\sf QFT}_M^\dagger$ on the output register $\ket u$.  Define the phase oracle operator $\pho := V^\dagger \cdot \sto \cdot V$. 
\begin{align*}
    \pho \ket x \ket u \otimes \ket H &=  V^\dagger \cdot  \sto \cdot  \frac{1}{\sqrt{M}}  \sum_y  \omega_M^{- u y} \ket x \ket y \otimes \ket H  \\ 
    &= V^\dagger  \cdot   \frac{1}{\sqrt{M}}  \sum_y \omega_M^{- u y} \ket x \ket {y + H(x)} \otimes \ket H  \\ 
     &=   \frac{1}{{M}} \sum_{y, y'}  \omega_M^{- u y + (y + H(x)) y'}  \ket x \ket {y'} \otimes \ket H  \\ 
     &=   \frac{1}{{M}} \omega_M^{u H(x)} \sum_{y, y'}  \omega_M^{(y + H(x)) (y' - u)}  \ket x \ket {y'} \otimes \ket H  \\ 
    &=  \ket x \ket u \otimes \omega_M^{u H(x)} \ket {H}. 
\end{align*}
Similarly, we override the notation $\pho$ such that for any auxiliary register $\ket z$, $\pho \ket x \ket u \ket z \otimes \ket H = \ket x \ket u \ket z \otimes \omega_M^{u H(x)} \ket {H}$.

Observing that $V V^\dagger = I$, the following lemma tells that we can efficiently convert between a standard oracle algorithm and a phase oracle algorithm. 
\begin{lemma}[{\cite[Lemma 3]{zhandry2019record}}] 
    \label{lem:pho-simulate}
    Let $\As$ be an (unbounded) quantum algorithm making queries to a standard oracle. 
    Let $\Bs$ be the algorithm that is identical to $\As$, except it performs $V$ and $V^\dagger$ before and after each query.
    Then the output distributions of $\As$ (given access to a standard oracle) and $\Bs$ (given access to a phase oracle) are identical.
    Therefore, a quantum random oracle can be perfectly simulated as a phase oracle.
\end{lemma}

We then have the following lemma for the phase oracle, that formulates the behavior of a quantum algorithm making at most $T$ queries to the phase oracle. We have seen that every $\pho$ query will add a phase to the $\ket H$ register, i.e., $\pho \ket x \ket u \otimes \ket H = \ket x \ket u \otimes \omega_M^{u H(x)} \ket {H}$. Define $D$ as a truth table, or equivalently a vector in $(\Z/M\Z)^N$ and $D(x)$ be the $x$-th entry of $D$. Define $|D|$ be the number of non-zero entries in $D$. For any $D$, we define $\ket {\phi_D} = \frac{1}{M^{N/2}} \sum_{H} \omega_M^{\langle D, H\rangle}  \ket H$ for all $D \in (\Z/M\Z)^N$ where  $\langle D, H\rangle$ is defined to be the inner product $\sum_{x \in [N]} D(x) H(x)$. Note that we will only use this inner product on the exponent of $\omega_M$ so it is irrelevant whether we are computing it on the integer ring or the ring modulo $M$.
\newcommand{\statelemboundeddatabasephase}{
    Let $\As$ be a quantum algorithm making at most $T$ queries to a phase oracle. The overall state of $\As$ and the phase oracle can be written as $\sum_{z, D: |D| \leq T} \alpha_{z, D} \ket z \otimes \frac{1}{M^{N/2}} \sum_{H} \omega_M^{\langle D, H\rangle} \ket H = \sum_{z, D: |D| \leq T} \alpha_{z, D} \ket z \otimes \ket {\phi_D}$. 
    
    Moreover, it is true even if the state is conditioned on arbitrary outcomes (with non-zero probability) of $\As$'s intermediate measurements. 
}
\begin{lemma}
    \label{lem:bounded-database-phase}
    \statelemboundeddatabasephase
\end{lemma}
\noindent
For completeness, we provide the proof for this lemma in \Cref{appendix:compressed-oracle}.

\paragraph*{Compressed standard oracle.} Intuitively, compressed oracle is an analogy of classical lazy sampling method. Instead of recording all the information of $H$ in the registers (like what it does in the standard oracle or the phase oracle), Zhandry provides a better solution which is useful to argue the amount of the information an algorithm knows about the random oracle. 

The oracle register records a database/list that contains the output on each input $x$, the output is an element in $\Z/M\Z \cup \{\bot\}$, where $\bot$ is a special symbol denoting that the value is ``uninitialized''.
The database is initialized as an empty list $D_0$ of length $N$, in other words, it is initialized as the pure state $\ket{\emptyset} := \ket{\bot, \bot, \cdots, \bot}$. Let $|D|$ denote the number of entries in $D$ that are not $\bot$. Define $D(x)$ to be the $x$-th entry.%

For any $D$ and $x$ such that $D(x) = \bot$, we define $D \cup (x, u)$ to be the database $D'$, such that for every $x' \ne x$, $D'(x') = D(x)$ and at the input $x$, $D'(x) = u$. 

The compressed standard oracle is the unitary $\csto := \stddecomp \circ \csto' \circ \stddecomp$, where
\begin{itemize}
    \item $\csto'\, |x, u\rangle |D\rangle = |x, u + D(x) \rangle |D\rangle$ when $D(x) \ne \bot$, which writes the output of $x$ defined in $D$ to the $u$ register. This operator will never be applied on an $x, D$ where $D(x) = \bot$.
    \item $\stddecomp(\ket x \otimes \ket D) := \ket x \otimes \stddecomp_x \ket D$, where ${\sf StdDecomp}_x\, |D\rangle$ works on the $x$-th register of the database $D(x)$. Intuitively, it swaps a uniform superposition $\frac 1{\sqrt M} \sum_y \ket y$ with $\ket \bot$ on the $x$-th register. 
    Formally,
        \begin{itemize}
            \item If $D(x) = \bot$, $\stddecomp_x$ maps $\ket \bot$ to $\frac{1}{\sqrt{M}} \sum_y \ket y$, or equivalently, $\stddecomp_x |D\rangle = \frac{1}{\sqrt{M}} \sum_y |D \cup (x, y)\rangle$. Intuitively, if the database does not contain information about $x$, it samples a fresh $y$ as the output of $x$. 
            
            \item If $D(x) \ne \bot$, $\stddecomp_x$ works on the $x$-th register, and it is an identity on $\frac{1}{\sqrt{M}} \sum_y \omega_M^{u y} \ket y$ for all $u \ne 0$; it maps the uniform superposition $\frac{1}{\sqrt{M}} \sum_y \ket y$ to $\ket \bot$. 
            
            More formally, for a $D'$ such that $D'(x) = \bot$, 
            \begin{align*}
                \stddecomp_x  \frac{1}{\sqrt{M}} \sum_y \omega_M^{u y} |D' \cup (x, y)\rangle 
                            = \frac{1}{\sqrt{M}} \sum_y \omega_M^{u y}  |D' \cup (x, y)\rangle  \text{ for any } u \ne 0  ,
            \end{align*}
            and, 
            \begin{align*}   
                \stddecomp_x  \frac{1}{\sqrt{M}} \sum_y |D' \cup (x, y)\rangle  =|D'\rangle      .
            \end{align*}
        \end{itemize}
        Since all $\frac{1}{\sqrt{M}} \sum_y \omega_M^{u y} \ket y$ and $\ket \bot$ form a basis, these requirements define a unique unitary operation. %
\end{itemize}
Zhandry proves that, ${\sf StO}$ and \csto{} are perfectly indistinguishable by any \textit{unbounded} quantum algorithm. 
\begin{lemma}[{\cite[Lemma 4]{zhandry2019record}}] 
    Let $\As$ be an (unbounded) quantum algorithm making oracle queries. The output of $\As$ given access to the standard oracle is exactly identical to the output of $\As$ given access to a compressed standard oracle. 
\end{lemma}

Combining this lemma with \Cref{prop:purifyoracle}, we obtain the following corollary.
\begin{corollary}
    \label{cor:csto-simulate}
    A quantum random oracle can be perfectly simulated as a compressed standard oracle.
\end{corollary}

In this work, we only consider query complexity, and thus simulation efficiency is irrelevant to us.
Looking ahead, we simulate a random oracle as a compressed standard oracle to help us analyze security of different games with the help from the following lemmas.

The first lemma gives a general formulation of the overall state of $\As$ and the compressed standard oracle after $\As$ makes $T$ queries, analogous to \Cref{lem:bounded-database-phase} for phase oracle.
We also defer the proof for this lemma to \Cref{appendix:compressed-oracle}.
\newcommand{\statelemboundeddatabase}{
    If $\As$ makes at most $T$ queries to a compressed standard oracle, assuming the overall state of $\As$ and the compressed standard oracle is $\sum_{z, D} \alpha_{z, D}\, |z\rangle_{\As} |D\rangle_H$, then it has support on all $D$ such that $|D| \leq T$. In other words, the overall state can be written as,
    \begin{align*}
        \sum_{z, D: |D| \leq T} \alpha_{z, D}\, |z\rangle_{\As} \otimes |D\rangle_H.
    \end{align*}
    Moreover, it is true even if the state is conditioned on arbitrary outcomes (with non-zero probability) of $\As$'s intermediate measurements. 
}
\begin{lemma}
    \label{lem:bounded-database}
    \statelemboundeddatabase
\end{lemma}

The second lemma provides a quantum analogue of lazy sampling in the classical ROM.

\newcommand{\statelemzhandry}{
    Let $H$ be a random oracle from $[N] \to [M]$.
    Consider a quantum algorithm $\As$ making queries to the standard oracle and outputting tuples $(x_1, \cdots, x_k, y_1, \cdots, y_k, z)$. Supposed the random function $H$ is measured after $\As$ produces its output. Let $R$ be an arbitrary set of such tuples. Suppose with probability $p$, $\As$ outputs a tuple such that (1) the tuple is in $R$ and (2) $H(x_i) = y_i$ for all $i$. Now consider running $\As$ with the compressed standard oracle {\sf CStO}, and supposed the database $D$ is measured  after $\As$ produces its output. Let $p'$ be the probability that (1) the tuple is in $R$ and (2) $D(x_i) = y_i$ (in particular, $D(x_i) \ne \bot$) for all $i$. Then $\sqrt{p} \leq \sqrt{p'} + \sqrt{k / M}$.
    
    Moreover, it is true even if it is conditioned on arbitrary outcomes (with non-zero probability) of $\As$'s intermediate measurements. 
}
\begin{lemma}[{\cite[Lemma 5]{zhandry2019record}}]
    \label{lem:zhandry-lemma5}
    \statelemzhandry
\end{lemma}

\section{Security Games}
\label{sec:security-games}
\subsection{Security Game with Advice}

\begin{definition}[Algorithm with Advice]
  An $(S, T)$ (query) algorithm $\mathcal{A} = (\mathcal{A}_1, \mathcal{A}_2)$ with (oracle-dependent) advice consists of two procedures: 
  \begin{itemize}
    \item $\ket \alpha \gets \mathcal{A}_1(H)$, which is an arbitrary (unbounded) function of $H$, and outputs an $S$-qubit quantum state $\ket\alpha$;
    \item $\ket{\ans} \gets \mathcal{A}^{\tilde H}_2(\ket\alpha, \ch)$, which is an unbounded algorithm that takes advice $\ket\alpha$, a challenge $\ch$, makes at most $T$ quantum queries to $\tilde H$, and outputs an answer, which we measure in the computational basis to obtain the classical answer $\ans$ if needed.
  \end{itemize}
  Furthermore, we distinguish the following cases:
  \begin{itemize}
    \item If both the output of $\mathcal A_1$ and all queries of $\mathcal A_2$ are classical queries\footnote{Here, we do not distinguish whether $\mathcal A^{\tilde H}_2$ is a quantum algorithm making classical queries or a completely classical algorithm making classical queries, since a quantum algorithm can always be simulated in classical finite time.}, we call it a classical algorithm with (classical) advice, or an $(S, T)$ algorithm in the AI-ROM (auxiliary input random oracle model); 
    \item If only the output of $\mathcal A_1$ is classical, we call it a quantum algorithm with classical advice, or an $(S, T)$ algorithm in the AI-QROM (auxiliary input quantum random oracle model);
    \item Otherwise, we call it a quantum algorithm with (quantum) advice, or an $(S, T)$ algorithm in the QAI-QROM (quantum auxiliary input quantum random oracle model).
    \item If $S = 0$, we call it an classical/quantum algorithm without advice, or an algorithm in the ROM (random oracle model)/QROM (quantum random oracle model) respectively. 
  \end{itemize}
\end{definition}

\begin{remark}
  In the above definition, we assume the advice is a pure state without loss of generality. 
  Mentioned in \cite{Aaronson05}, by Kraus' Theorem, every $S$-qubit mixed state can be realized
  as half of a $2 S$-qubit pure state.
\end{remark}

Below, we will use the word ``adversary" and ``algorithm" interchangeably, especially when we consider interactive security games shortly after.

\begin{definition}[Security Game]
 \label{def:security-game}
  Let $H$ be a random oracle $[N] \to [M]$.
  A (non-interactive) (classical) security game $G = (C)$ is specified by a challenger $C = (\samp, \query, \ver)$, where:
  \begin{enumerate}
    \item $\ch \gets \samp^H(r)$ is a deterministic classical algorithm that takes randomness $r \in R$ as input, and outputs a challenge $\ch$.  
    \item $\query^H(r, \cdot)$ is a deterministic classical algorithm that hardcodes the challenge and provides adversary's online queries\footnote{For almost all applications below, $\query^H(r, \cdot) = H(\cdot)$. The only exception is Yao's box, where the adversary cannot query the challenge point.}.
    \item $b \gets \ver^H(r, \ans)$ is a deterministic classical algorithm that takes the input $\ans$ and outputs a decision $b$ indicating whether the game is won. 
  \end{enumerate}

  For every algorithm with advice, i.e.  $\As = (\As_1, \As_2)$ , we define
  \begin{align*}
  \As \Longleftrightarrow C(H) := \ver^H\left(r, \As^{\tilde H}_2(\As_1(H), \samp^H(r))\right)
  \end{align*}
  to be the binary variable indicating whether $\As$ successfully makes the challenger output 1, or equivalently if $\As$ wins the security game, where $\tilde H(\cdot) := \query^H(r, \cdot)$.
\end{definition}

\begin{definition}[Security in the AI-ROM/AI-QROM/QAI-QROM]
  We define the security in the AI-ROM, AI-QROM, QAI-QROM of a security game $G$ to be
  $$\delta = \delta(S, T) := \sup_{\As} \Pr_{H, r, \As} \left[ \As \Longleftrightarrow C(H) = 1\right],$$
  where $\As$ in the probability denotes the randomness of the algorithm, and supremum is taken over all $\As$ in the AI-ROM/AI-QROM/QAI-QROM (respectively). 
\end{definition}

Additionally, we can say a security game $G$ is $\delta$-secure if its security is at most $\delta$.

For some games, it is helpful to consider how much advantage an algorithm can gain by getting more queries or more advice.
This is formalized as below.

\begin{definition}
  We call the security game a \textbf{decision} game if $\ans \in \{0, 1\}$.
\end{definition}

\begin{definition}[Advantage against Decision Games]
  We define the advantage of $\As$ for a decision game $G$ to be
  \begin{align*}
    \eps = \eps(S, T) := \delta(S, T) - 1/2,
  \end{align*}
  if it has winning probability $\delta(S, T)$.
\end{definition}

\begin{definition}[Best Advantage of Decision Games]
  We define the best advantage of a decision game $G$ in model $\mathcal M$ to be $\eps(S, T) := \delta(S, T) - 1/2$ if $G$ has security $\delta(S, T)$ in $\mathcal M$, where the $\mathcal M$ could be AI-ROM, AI-QROM, or QAI-QROM.
\end{definition}

Another class of security games that we would like to consider is when the adversary can verify the answer by itself, which we formalize as below.

\begin{definition}
  We call a security game to be \textbf{publicly-verifiable} with verification time $T_\ver$, if $\ver^H(r, \cdot) = \widetilde\ver^{\tilde H}(\ch, \cdot)$ for some classical algorithm $\widetilde\ver^{\tilde H}$ where $\tilde H(\cdot) = \query^H(r, \cdot)$, and $T_\ver$ is the upper bound on the number of $\tilde{H}$ queries for computing $\ver^{\tilde{H}}(\ch, \cdot)$.
\end{definition}

Since by our definition of security games, an algorithm has only access to $\tilde{H}(\cdot)$ which may not be the same as $H(\cdot)$. Therefore, publicly-verifiable means an algorithm can verify an answer by only making queries to $\tilde{H}$.

\begin{remark}
  In general, interesting \emph{decision} games are not \emph{publicly-verifiable}, otherwise, there exists an adversary that breaks the game using only as many queries as it takes for the challenger to verify the answer.
\end{remark}

\subsection{Multi-Instance Security}
\label{sec:mis}

In this section, we introduce multi-instance game and multi-instance security.

\begin{definition}[Multi-Instance Security Game]
  \label{def:mis-game}
  For any security game $G = (C)$ and any positive integer $g$, we define the multi-instance game $\Gmi{g} = (\Cmi{g})$, where $C^{\otimes g}$ is given as follows:
  \begin{enumerate}
    \item For $i \in [g]$, do: \begin{enumerate}
      \item Sample fresh randomness $r_i \gets_\$ R$;
      \item Compute $\ch_i \gets C.\samp^H(r_i)$ and send it to the adversary;
      \item Give adversary access to oracle $C.\query^H(r_i, \cdot)$ until the adversary submits a quantum state (register) $\ket{\ans_i}$;
      \item Let $\{P_0, P_1\}$ be a projective measurement where $P_1$ defines all $\ans$ that $C.\ver^H(r_i, \ans) = 1$ and $P_0 = I - P_1$. Measure $\ket{\ans_i}$ in $\{P_0, P_1\}$ to get the quantum state $\ket{\ans'_i}$ and store the result in $b_i$;
      \item Send $\ket{\ans'_i}$ back to the adversary.
    \end{enumerate}
    \item Output $b_1 \land b_2 \land ... \land b_g$.
  \end{enumerate}
\end{definition}

\begin{definition}
  A $(g, S, T)$ adversary with advice for a multi-instance security game $\Gmi{g} = (\Cmi{g})$ is defined as $\As = (\As_1, \As_2)$, where the interactions between $\As_2(\ket \alpha)$  and $\Cmi{g}$ is defined as follows.
  \begin{enumerate}
    \item $\ket \alpha \gets \mathcal{A}_1(H)$, which is an arbitrary (unbounded) function of $H$, and outputs an $S$-qubit quantum state $\ket\alpha$ for $\As_2$;
    \item For each $i \in [g]$,
    \begin{enumerate}
        \item $\As_2$ is given a challenge $\ch_i$ and an oracle $\tilde H_i$ from $\Cmi{g}$; 
        \item $\As_2$ makes at most $T$ queries to $\tilde{H}_i$ and prepares $\ket{\ans_i}$; 
        \item $\As_2$ sends $\ket{\ans_i}$ to $\Cmi{g}$ and gets $\ket{\ans'_i}$ back;
    \end{enumerate}
    \item Finally, $\Cmi{g}$ outputs a bit $b$. 
  \end{enumerate}
  \noindent
  In particular, if $S = 0$, we also call it $(g, T)$ classical/quantum adversary (without advice), or a $(g, T)$ algorithm in the ROM/QROM respectively.

  For any  $\As$, which is a $(g, S, T)$ adversary with advice, we define
  \begin{align*}
  \As \Longleftrightarrow \Cmi{g}(H)
  \end{align*}
  to be the binary variable indicating whether $\As$ successfully makes the challenger output 1 in the game defined above, or equivalently if $\As$ wins the multi-instance security game.
\end{definition}

\begin{definition}[Multi-Instance Security]
  We say a  multi-instance game $\Gmi g$ is $\delta$-secure in the QROM (quantum random oracle model) if for any $(g, T)$ adversary $\As$,
  $$\Pr_{H, \As, \Cmi g}[\As \Longleftrightarrow \Cmi g(H)] \le \delta^g = \delta(g, T)^g,$$
  where $\As$ in the probability denotes the randomness of the algorithm, $\Cmi g$ in the probability denotes the randomness of the challenger.
\end{definition}

\begin{remark}[Comparison with multi-instance problems]
  \label{rem:comparison-mis-vs-dpt}
  We expand our discussion in \Cref{sec:tech-overview-reduction} by noting that our multi-instance game is different from multi-instance problems (also similar to direct product theorem) in two important aspects: \begin{enumerate}
    \item The adversary gets back $\ket{\ans'_i}$ after the challenger measures her decision.
          For a classical adversary (i.e. if $\ket{\ans'_i}$ is classical), this makes no difference.
          For a quantum adversary, this makes the game easier ($\delta$ might go up).
          Looking ahead, this change is only necessary for \Cref{thm:public-ver-reduction} but we state this for general games as it does not hurt the general bound in all the cases we consider.
    \item The adversary gets access to $\ch_{i + 1}$ only after submitting his $\ket{\ans_i}$.
          This makes the game harder, especially for quantum adversaries, as discussed in \Cref{sec:tech-overview-reduction}.
  \end{enumerate}
\end{remark}

\subsection{Helper Lemma for Multi-Instance Security}

We now show a helpful lemma, which shows that to prove multi-instance security, it suffices to prove conditional security for each round.

\begin{lemma}
  \label{lem:conditioning2mis}
  Let $G$ be any security game.
  Let $B_i$ be the random variable for $b_i$ for the multi-instance game $\Gmi g$ as defined in \Cref{def:mis-game}, and let $X := (B_1, B_2, ..., B_{i - 1}, Y)$ be the random variable that denote the outcomes of all the measurements during the first $(i - 1)$ rounds.
  For any $\delta$, if for any multi-instance adversary $\As$ and for any $x$ that $\Pr[X = x] > 0$, we have that
  \[
    \Pr[B_i = 1 | X = x] \le \delta,
  \]
  then $\Gmi{g}$ is $\delta$-secure.
\end{lemma}
\begin{proof}
  We note that the success probability of $\As$ is exactly
  \begin{align*}
    &\equad \Pr\left[ B_1 = B_2 = \cdots = B_g = 1 \right] \\
    &= \prod_{i=1}^g \Pr\left[ B_i = 1 \,|\, B_1 =  \cdots = B_{i-1} = 1 \right] \\
    &= \prod_{i=1}^g \E_Y \left[ \Pr\left[ B_i = 1 \,|\, B_1 =  \cdots = B_{i-1} = 1, Y \right] \right] \\
    &\le \delta^g.
  \end{align*}
  Here in the last inequality, we implicitly assume that the expectation is only taken over $\Pr[Y \,|\,\allowbreak B_1 = \cdots = B_{i-1} = 1] > 0$, as otherwise the term will trivially evaluate to zero.
\end{proof}

\section{Reducing Games with Advice to Multi-Instance Games}
\label{sec:advice2mis-reductions}
In this section, we give four reductions from security against $(g, T)$ multi-instance adversaries to security against $(S, T)$ adversaries with advice.

We first present the theorem that bootstraps security of multi-instance game to security against quantum adversaries with classical advice or completely classical adversaries.

\newcommand{\statethmreductionclassicaladvicegeneral}{
  Given a security game $G$.
  Let $\delta_0$ be the winning probability of an adversary that outputs a random answer without advice or making any query.
  Assume that multi-instance game $\Gmi g$ is $\delta(g, T)$-secure in QROM (or ROM), $G$'s security against $(S, T)$ adversaries in \varul{AI-QROM} (or \underline{AI-ROM}, respectively) is
  \[\delta' \le 4 \cdot \delta(S + \log(1/\delta_0) + 1, T).\]
}
\begin{theorem}
  \label{thm:reduction-classical-advice-general}
  \statethmreductionclassicaladvicegeneral
\end{theorem}

\begin{remark*}
  Looking ahead, we show in \Cref{lem:owfmis} that for OWF, its multi-instance security  in the QROM is $\delta(g, T) = O\left(\frac{gT + T^2}{\min\{N, M\}}\right)$. The winning probability of random guess $\delta_0$ is at least $1 / N$.
  Plug these parameters into the corollary above, we have that OWF's security in the AI-QROM is at most
  \[
    O\left(\frac{(S + T + \log N) \cdot T}{\min\{N, M\}}\right) = \tilde O\left(\frac{ST + T^2}{\min\{N, M\}}\right).
  \]
\end{remark*}
\medskip

While this theorem is very general, it yields no meaningful bound when $\delta \ge 1/4$, which is the case for decision games as a random guessing adversary achieves success probability $1/2$!
We show that this is not an issue, by showing that a more careful application of the reduction that is used for proving the theorem above, actually yields the following bound for decision games against classical advice, although with worse exponents on $\eps' = \delta' - 1/2$.

\newcommand{\statethmdecisiongameclassicaladvicereduction}{
  Let $G$ be a \varul{decision} game. Let $\eps_0$ be a lower bound on the advantage for an adversary with $T_0$ queries and $S_0$ bits of advice.
  Assume that multi-instance game $\Gmi g$ is $(1/2 + \eps(g, T))$-secure (or has best advantage $\eps(g, T)$) in QROM (or ROM). For any $(S, T)$ adversaries in \varul{AI-QROM} (or \underline{AI-ROM}, respectively), the best advantage against $G$ is  $\eps'$, which satisfies, 
  \[\eps' \le 4 \cdot \eps\left(\frac{10 \ln 2}{\eps'} \cdot (S + S_0 + \log(1/\eps_0) + 2), T + T_0 \right).\]
}
\begin{theorem}
  \label{thm:decision-game-classiacl-advice-reduction}
  \statethmdecisiongameclassicaladvicereduction
\end{theorem}

\begin{remark*}
  Looking ahead, we show in \Cref{lem:ybmis} that for Yao's box, its multi-instance best advantage is $\eps(g, T) = O\left(\sqrt{\frac{gT}N}\right)$. When $S_0 = 1$ and $T_0 = 0$, $\eps_0 \ge \frac 1{2N}$ (which is simply remembering $H(1)$).
  Plug these parameters into the theorem above, we have that Yao's box's best advantage in the AI-QROM $\eps'$ satisfies
  \[
    \eps' \le O\left(\sqrt{\frac{(S + 1 + \log N + 2) \cdot T}{\eps' N}}\right).
  \]
  Therefore, we conclude that
  \[
    \eps' \le O\left( \frac{(S + \log N) \cdot T}N \right)^{1/3} = \tilde O\left( \frac{S T}N \right)^{1/3}.
  \]
\end{remark*}
\medskip

These two theorems conclude the bounds for adversaries with classical advice.
Next, we consider bounds for quantum adversaries with advice, which is slightly worse than their classical counterparts.

We present a different reduction for quantum adversaries with advice, albeit the reduction only works for publicly-verifiable games.

\newcommand{\statethmpublicverreduction}{
  Given a \varul{publicly-verifiable} security game $G$ with verification time $T_\ver$.
  Assume that $\delta_0 = 1/\poly(N, M)$ is the winning probability by random guess, and assume that multi-instance game $\Gmi g$ is $\delta(g, T)$-secure  in the QROM. For any $(S, T)$ adversaries in the QAI-QROM, the best advantage against  $G$ is $\delta'$, which satisfies
  \[\delta' \le  8 \cdot \delta\left(\tilde O(S)/\delta', \tilde O(T + T_\ver)/\delta'\right),\]
  where $\tilde O$ absorbs $\poly(\log N, \log M, \log S)$ factors.
}
\begin{theorem}
  \label{thm:public-ver-reduction}
  \statethmpublicverreduction
\end{theorem}

\begin{remark*}
  Recall that for OWF, its multi-instance security is $\delta(g, T) = O\left(\frac{gT + T^2}{\min\{N, M\}}\right)$ and $T_\ver = 1$; for random guess, the winning probability is $\delta_0 \ge 1/N$.
  Plug these parameters into the corollary above, we have that OWF's security $\delta'$ in the QAI-QROM  satisfies
  \[
  \delta' \le \tilde O\left(\frac{ST + T^2}{\delta'^2 \cdot \min\{N, M\}}\right). 
  \]
  In conclusion,
  \[
  \delta' \le \tilde O\left(\frac{ST + T^2}{\min\{N, M\}}\right)^{1/3}.
  \]
\end{remark*}
\medskip

As noted before that publicly-verifiable games and decision games are essentially disjoint, we finally present the fourth reduction for decision games.

\newcommand{\statethmdecisiongamequantumadvicereduction}{
  Given a \underline{decision} security game $G$.
  Assume that $\eps_0 = 1/\poly(N, M)$ is a non-negligible lower bound on the advantage for an adversary with $T_0$ queries and $S_0$ qubits of advice, and assume that multi-instance game $\Gmi g$ is $\eps(g, T)$-secure in the QROM. For any  $(S, T)$ adversaries in QAI-QROM, the  best advantage against $G$ is $\eps'$, which  satisfies
  \[\eps' \le 8 \cdot \eps\left(\tilde O(S + S_0)^3/\eps'^9, \tilde O((T + T_0)(S + S_0)^2)/\eps'^8\right),\]
  where $\tilde O$ absorbs $\poly(\log N, \log M, \log S)$ factors.
}
\begin{theorem}
  \label{thm:decision-game-quantum-advice-reduction}
  \statethmdecisiongamequantumadvicereduction
\end{theorem}
\begin{remark*}
  Recall that for Yao's box, its multi-instance advantage is $\eps(g, T) = O\left(\sqrt{\frac{gT}N}\right)$; when $S_0 = 1$ and $T_0 = 0$, $\eps_0 \ge \frac 1{2N}$ (which is simply remembering $H(1)$).
  Plug these parameters into the corollary above, we have that Yao's box's best advantage $\eps'$ in the QAI-QROM satisfies
  \[
  \eps' \le \tilde O\left(\sqrt{\frac{S^5 T}{\eps'^{17} N}}\right),
  \]
  therefore, we conclude that
  \[
  \eps' \le \tilde O\left( \frac{S^5 T}N \right)^{1/19}.
  \]
\end{remark*}

\subsection{Reduction for Classical Advice}

\begin{fact}
  \label{lem:averaging-argument}
  Given events $E_1, E_2, ..., E_N$, any real number $B \geq 0$, $\eps$, let $I$ be a uniformly random integer in $[N]$.
  If $\Pr[E_I] \ge B + \eps$, there exists a subset $S \subseteq [N]$ of size $|S| \ge \eps N / 2$, such that for any $i \in S$, $\Pr[E_i] \ge B + \eps/2$.  
\end{fact}
\begin{proof}
  Without loss of generality, assume $\eps > 0$.
  Let $S^*$ be the subset of $[N]$ that contains all $i \in [N]$ whose $\Pr[E_i] \geq B + \eps / 2$.   Assume for contradiction, $|S^*| <  \eps N / 2$. We have, 
  \begin{align*}
    B + \eps
      &\le \Pr[E_I] \\
      &\le 1 \cdot \Pr[I \in S^*] + \Pr[E_I | I \not\in S^*] \cdot \Pr[I \not\in S^*] \\
      & < 1 \cdot \eps/2 + (B + \eps/2) \cdot 1 \\
      & = B + \eps,
  \end{align*}
  which is a contradiction. 
\end{proof}

\begin{proposition}
  \label{thm:classical-mis-reduction}
  Given a security game $G$.
  Given an $(S, T)$ adversary with \underline{classical} advice $\As$ for a security game $G$ with winning probability $\delta$, then for all real $B \ge 0$, let $\delta = B + \eps$, there exists an $(g, S, T)$ adversary with classical advice $\As'$ for the multi-instance game $\Gmi g$ with winning probability $\delta' \ge \eps/2 \cdot (B + \eps/2)^g$ for any $g > 0$.
  
  Furthermore, this reduction is classical, meaning that if $\As$ is a classical algorithm with advice, so is $\As'$.
\end{proposition}
\begin{proof}
  
  Without loss of generality, assume $B \le \delta$ or $\eps \ge 0$.
  Consider the following adversary with advice $\mathcal A'$:
  \begin{enumerate}
    \item Compute $\alpha \gets \As_1(H)$; %
    \item For $i \in [g]$, on receiving challenge $\ch_i$ and given oracle access to $\tilde H_i$, computes $\As_2^{\tilde H_i}(\alpha, \ch_i)$, sends its output, and ignores the answer sent back by the challenger.
  \end{enumerate}

  Consider the event $E_H$ to be the event when $\As$ wins the game conditioned on the random oracle is fixed as $H$.
  Using \Cref{lem:averaging-argument}, we can find a subset $H_\good$ of all possible functions such that whenever the function $H$ is in $H_\good$, $\Pr_{r, \mathcal A}[\As \Longleftrightarrow C(H) = 1] \ge B + \eps/2$.
  Since $H_\good$ takes up at least $\eps/2$ fraction of all possible functions, we have the theorem by observing that $\mathcal A'$ given above have the success probability lower bounded by $\eps/2 \cdot (B + \eps/2)^g$.
\end{proof}

\begin{proposition}
  \label{thm:mis-remove-advice}
  Given a security game $G$ and a $(g, S, T)$ adversary with (classical or quantum) advice $\As$ for the multi-instance game $\Gmi g$ with winning probability $\delta$, there exists an $(g, T)$ adversary (without advice) $\As'$ for the same game with winning probability at least $2^{-S} \cdot \delta$.
  
  Furthermore, the reduction is classical.
\end{proposition}
\begin{proof}
  $\As'$ works simply by invoking the second-stage adversary with a maximally mixed state $I_{2^S}$ as its advice.
  Note that a maximally mixed state is a mixed state in any basis, in particular, assume that the correct advice is $\ket\alpha$, we can write $I_{2^S}$ as a maximally mixed state in basis states $v_1, v_2, ..., v_{2^S}$ where $v_1 = \ket\alpha$.
  Using this observation, we conclude that $\As'$ has winning probability at least $2^{-S} \cdot \delta$.
  
  The reduction is classical, as picking a maximally mixed state $I_{2^S}$ is equivalent to choosing a $S$-bit bit-string uniformly at random. 
\end{proof}

Combining the two propositions above yields the following corollary.
\begin{corollary}
  \label{cor:reduction-classical-advice-general}
  Let $G$ be a security game.
  Given an $(S, T)$ adversary with \underline{classical} advice $\As$ for the security game $G$ with success probability $\delta$, then for all real $B \ge 0$, let $\delta = B + \eps$, there exists an $(g, T)$ adversary (without advice) $\As'$ for the multi-instance game $\Gmi g$ with winning probability $\delta' \ge 2^{-S} \cdot  \eps/2  \cdot (B + \eps/2)^g$ for any $g > 0$.
  
  Furthermore, the reduction is classical.
\end{corollary}

\begin{customtheorem}{\ref{thm:reduction-classical-advice-general}}
  \statethmreductionclassicaladvicegeneral
\end{customtheorem}
\begin{proof}
  Let $\As$ be the $(S, T)$ adversary with classical advice $\As$ for $G$ with the best winning probability $\delta'$, then using \Cref{cor:reduction-classical-advice-general} with $B = 0$, there exists an $(g, T)$ adversary (without advice) $\As'$ for the multi-instance game $\Gmi{g}$ with winning probability at least $2^{-S} \cdot \delta'/2 \cdot (\delta'/2)^g$ for any $g > 0$.
  By multi-instance security, we know that
  \[
    \delta(g, T)^g
      \ge 2^{-S} \cdot \delta'/2 \cdot (\delta'/2)^g
      \ge 2^{-S} \cdot \delta_0/2 \cdot (\delta'/2)^g.
  \]
  Take $g = S + \log(1/\delta_0) + 1$, we have
  \[
    \delta(g, T)^g \ge 2^{-g} \cdot (\delta'/2)^g,
  \]
  and thus we have the theorem.
\end{proof}

\subsection{Reduction for Decision Games with Classical Advice}

\begin{fact}
  \label{lem:unimportant-1}
  For any real $C \ge 0$, $0 \le x \le 1$, we have that $\frac C{C + x} \le 1 - \frac x{C + 1}$.
\end{fact}
\begin{proof}
  By monotonicity and convexity of $\frac C{C + x}$ on $x$.
\end{proof}

\begin{customtheorem}{\ref{thm:decision-game-classiacl-advice-reduction}}
  \statethmdecisiongameclassicaladvicereduction
\end{customtheorem}
\begin{proof}
  Let $\As$ be the $(S, T)$ adversary with classical advice $\As$ for $G$ with the best advantage $\eps'$. 
  By security of multi-instance game and \Cref{cor:reduction-classical-advice-general} with $B = 1/2$, we know that there exists an $(g, T)$ adversary (without advice) $\As'$ for the multi-instance game $\Gmi{g}$ with winning probability at least $2^{-S} \cdot \eps' / 2  \cdot (1/2 + \eps'/2)^g$ for any $g > 0$. Moreover, if $S \ge S_0$ and $T \ge T_0$, we have $\eps' \geq \eps_0$. Therefore, 
  \[
    (1/2 + \eps(S, T))^g
    \ge 2^{-S} \cdot (\eps'/2) \cdot (1/2 + \eps'/2)^g
    \ge 2^{-S} \cdot (\eps_0/2) \cdot (1/2 + \eps'/2)^g.
  \]
  Therefore, in general,
  \[
    (1/2 + \eps(g, T + T_0))^g \ge 2^{-S - S_0} \cdot (\eps_0/4) \cdot (1/2 + \eps'/2)^g.
  \]
  Take $g = \frac{10 \ln 2}{\eps'} \cdot (S + S_0 + \log(1/\eps_0) + 2)$, using \Cref{lem:unimportant-1}, and the fact that, we have that $(1 - 1/x)^x \le 1/e$ for any real $x \ge 1$, we have
  \begin{align*}
    \left(\frac{1/2 + \eps'/4}{1/2 + \eps'/2}\right)^g
      &\le \left(1 - \frac{\eps'/4}{1/2 + \eps'/4 + 1}\right)^g \\
      &=   \left(1 - \frac 1{4(1/2 + 1)/\eps' + 1}\right)^g \\
      &\le \left(1 - \frac 1{5(1/2 + 1)/\eps'}\right)^g \\
      &\le \left(1 - \eps'/10\right)^g \\
      &\le 2^{-(S + S_0 + \log(1/\eps_0) + 2)} \\
      &=   2^{-S - S_0} \cdot \eps_0/4.
  \end{align*}
  Combining the two formulae above, we obtain that
  \[
    (1/2 + \eps(g, T + T_0))^g \ge (1/2 + \eps'/4)^g,
  \]
  where  $g = \frac{10 \ln 2}{\eps'} \cdot (S + S_0 + \log(1/\eps_0) + 2)$. 
  Thus we have the theorem.
\end{proof}

\subsection{Reduction for Publicly-Verifiable Games}

\begin{lemma}[{Gentle measurement lemma~\cite[Lemma 2.2]{Aaronson05}}]
  \label{lem:gentle-measurement}
  Suppose a 2-outcome measurement of a mixed state $\rho$ yields outcome 0 with probability $1 - \eps$. Then after the measurement, we can recover a state $\tilde\rho$ such that $\trdist{\tilde\rho - \rho} \le \sqrt\eps$.
  This is true even if the measurement is a POVM (that is, involves arbitrarily many ancilla qubits).
\end{lemma}

\begin{lemma}[{Quantum union bound~\cite[Corollary 11]{AR19}}]
  \label{lem:quantum-union-bound}
  Let $\rho$ be a mixed state and let $S_1, ..., S_m$ be any quantum operations.
  Suppose that for all $i$, we have
  \[\trdist{S_i(\rho) - \rho} \le \eps_i.\]
  Then
  \[\trdist{S_m(S_{m - 1}(... S_1(\rho) ...)) - \rho} \le \eps_1 + ... + \eps_m.\]
\end{lemma}

\begin{lemma}[{\cite[Lemma 3.6]{BV97}}]
  \label{lem:trdist2prob}
  Let $\mathcal D(\psi)$ denote the probability distribution that results form a measurement of $\ket\psi$ in the computational basis.
  If $\trdist{\ket\phi - \ket\psi} \le \eps$, then $\norm{\mathcal D(\phi) - \mathcal D(\psi)}_1 \le 4\eps$.
\end{lemma}

\begin{proposition}
  There exists a universal constant $c > 0$.
  Given a \varul{publicly-verifiable} security game $G$ with verification time $T_\ver$.
  Given an $(S, T)$ adversary with advice $\As$ for $G$ with winning probability $\delta$, there exists an $(g, S', T')$ adversary with advice $\As'$ for the multi-instance game $\Gmi g$ with winning probability at least $(\delta/4)^{g+1}$ for any $g > 0$, where
  \begin{equation*}
    \begin{cases}
      S' = kS, \\
      T' = 2k(T + T_\ver), \\
      k = c \cdot \log(g + 1) / \delta.
    \end{cases}
  \end{equation*}
\end{proposition}
\begin{proof}
  Consider the following adversary $\As'$:
  \begin{enumerate}
    \item At the pre-processing phase, compute $\ket{\alpha_j} \gets \As_1(H)$ for all $j \in [k]$.
    \item At the online phase, on receiving challenge $\ch_i$ and given oracle access to $\tilde H(\cdot) := \query^H(\ch_i, \cdot)$, do: \begin{enumerate}
      \item Run $\ket{\ans_j} \gets \As_2^{\tilde{H}}(\ket{\alpha_j}, \ch_i)$ for all $j \in [k]$;
      \item Compute in superposition $\ket{v_j} \gets \widetilde\ver^{\tilde H}(\ch_i, \ket{\ans_j})$ for all $j \in [k]$. 
      \item In superposition, find the first $j$ such that $v_j = 1$, and store $\ans_j$ into an empty register $\ans$. 
      \item Send $\ket{\ans}$ to challenger, and receives $\ket{\ans'}$ back;
      \item Un-compute (c), (b), (a). %
    \end{enumerate}
  \end{enumerate}
  
  First, it is easy to see that for $\As'$, its advice is of length $S' = k S$ and it requires to make at most $T' = 2 k (T + T_\ver)$ oracle queries ($T + T_\ver$ queries to compute each $v_j$ and $T + T_\ver$ queries to un-compute the previous round). 
  
  Using \Cref{lem:averaging-argument}, we can find a subset $H_\good$ that is at least $\delta/2$ fraction of all functions, such that any function $H$ in $H_\good$, $\Pr_{r, \As}[\As \Longleftrightarrow C(H) = 1] \ge \delta/2$.
  
  Fix any function in $H_\good$.
  Using \Cref{lem:averaging-argument} again, we can find a subset $R_\good$ that is of size at least $\delta|R|/4$ (which consequently is of at least $\delta/4$ fraction), such that for any randomness $r$ in $R_\good$, $\Pr_{\As}[\As \Longleftrightarrow C(H) = 1] \ge \delta/4$.
  
  Assume $\ch_1$ is generated using the randomness in $R_\good$, then we know that the challenger will reject the first answer with probability at most $(1 - \delta/4)^k \le \frac 1{100g^4}$ if $c$ is a large enough constant.
  Let $\rho$ be $\ket{\alpha_1}, ..., \ket{\alpha_k}$.
  By \Cref{lem:gentle-measurement}, after one round, the advice state we recovered from un-computing $\rho_1$ satisfies $\trdist{\rho_1 - \rho} \le 1/\sqrt{100g^4} = 1 / 10 g^2$.
  
  If all challenges are generated using randomness in $R_\good$, by \Cref{lem:quantum-union-bound}, the overall advice state starting in each round $\rho_j$ satisfies $\trdist{\rho_j - \rho} \le g \cdot \frac 1{10g^2}$.
  Therefore, in each round, since every operation besides the measurement is unitary, by \Cref{lem:trdist2prob}, the challenger accepts with probability at least $1 - \frac 1{2g}$.
  By union bound, the challenger accepts with probability at least $1/2$, finishing the proof that the winning probability is at least $\delta/2 \cdot (\delta/4)^g \cdot 1/2$.
\end{proof}

Combining this with \Cref{thm:mis-remove-advice} (to guess a quantum advice by a maximally mixed state), we obtain the following.

\begin{corollary}
  There exists a universal constant $c > 0$.
  Given a \varul{publicly-verifiable} security game $G$ with verification time $T_\ver$.
  Given an $(S, T)$ adversary with advice $\As$ for $G$ with winning probability $\delta$, there exists an $(g, T')$ adversary $\As'$ for the multi-instance game $\Gmi g$ with winning probability at least $2^{-kS} \cdot   (\delta/4)^{g+1}$ for any $g > 0$, where $T' = 2k(T + T_\ver)$ and $k =c \cdot \log(g + 1) / \delta$.
\end{corollary}

\begin{fact}
  \label{lem:unimportant-2}
  Given any real $C \ge 0, D \ge 2$, if $g_0 = C + D + 14$ and $g = 2 g_0 \log g_0$, $g \ge C \log(g + 1) + D$.
\end{fact}
\begin{proof}
Since $C \ge 0, D \ge 2$, we have $g_0 \ge 16$ and $\log g_0 \geq 4$. 
  \begin{align*}
      g - (C \log (g+1) + D) & \geq g - (C \log (2 g) + D) \\
                            & \geq 2 g_0 \log g_0 - (C \log (4 g_0 \log g_0) + D) \\
                            & \geq 2 (C + D + 14) \log g_0  - (C \log (4 g_0 \log g_0) + D) \\
                            & \geq C (2 \log g_0 - (2 + \log g_0 + \log \log g_0)) + D (2 \log g_0 - 1).
  \end{align*}
  We complete the proof by noting that coefficients of $C$ and $D$ are both non-negative when $\log g_0 \geq 4$. 
\end{proof}

\begin{customtheorem}{\ref{thm:public-ver-reduction}}
  \statethmpublicverreduction
\end{customtheorem}
\begin{proof}
  Let $\As$ be the $(S, T)$ adversary with advice for $G$ with the best advantage $\delta'$, then there exists an $(g, T)$ adversary (without advice) $\As'$ for the multi-instance game $\Gmi g$ with winning probability at least $2^{-kS} \cdot (\delta'/4)^{g+1}$ for any $g > 0$, where $T' = 2k(T + T_\ver)$ and $k = c \cdot \log(g + 1) / \delta'$. Then we have, 
  \begin{align*}
  \delta(g, T')^g
    & \ge 2^{-kS} \cdot (\delta'/4)^{g+1} \\
    & \ge 2^{-kS} \cdot (\delta_0 / 4) \cdot (\delta'/4)^{g}
  \end{align*}
  Take $g_0 = \frac c{\delta'} \cdot S + \log(1/\delta_0) + 16$ and $g = 2 g_0 \cdot \log g_0$. By \Cref{lem:unimportant-2}, let $C = \frac{c}{\delta'} S$ and $D = \log(1/\delta_0) + 2$, we have $g \ge C \log(g+1) + D = c \log(g + 1) S/ \delta' + \log(1/\delta_0) + 2$. 
  
  Therefore, we have
  \begin{align*}
  \delta(g, T')^g & \ge 2^{-k S} \cdot (\delta_0 / 4) \cdot (\delta'/4)^{g} \\
        & =   2^{- c \cdot \log(g + 1) S / \delta'} \cdot 2^{- \log(1/\delta_0) - 2} \cdot (\delta'/4)^{g}  \\
        &\ge (\delta'/8)^g.
  \end{align*}
  Therefore, we conclude that
  \[
  \delta' \le 8 \cdot \delta(g, T').
  \]
  where $g = \tilde{O}(S) / \delta'$ and $T' = \tilde{O}(T + T_\ver) / \delta'$. 
\end{proof}

\subsection{Reduction for Decision Games with Quantum Advice}

\begin{theorem}[{Online shadow tomography~\cite[Theorem 9]{AR19}}]
  \label{thm:shadow-tomography}
  There exists an explicit procedure that performs shadow tomography, namely on input $n$ copies of an arbitrary unknown $d$-dimensional mixed state $\rho$, and any two-outcome measurements $E_1, ..., E_m$, it estimates $\Pr[E_i(\rho) = 1]$ within an additive error of $\pm \eps$ for every $i \in [m]$ with overall success probability $1 - \beta$, if
    \[n = \Omega\left(\frac{\log^2 m \cdot \log^2 d \cdot \log \frac 1 \beta}{\eps^8}\right).\]
  Furthermore, this procedure satisfies the following properties:
  \begin{enumerate}
    \item It is online, meaning that it can be given by a pair of quantum algorithms: an initialization algorithm $\ket \st \gets \qpmwinit(n, \rho^{\otimes n}, m, \eps, \beta)$, and an estimation algorithm $(\ket \st, p_i) \gets \qpmwest(\ket \st, i, E_i)$. $\ket \st$ will be updated by the estimation algorithm $\qpmwest$. 
    \item If the measurement $E_i$ can be written as (1). prepare a state $\ket \ans$ which is initialized as $0$, (2). apply a unitary $U_i$ over the joint state $\rho$ and $\ket \ans$, (3). measure  $\ket\ans$ in the computational basis. Then $U_i$ (or its inverse) will be run at most $4n$ times\footnote{This can be verified by looking at \cite[Figure 1]{AR19}.}.
    
  \end{enumerate}
\end{theorem}

First, we will present the following useful fact which says that additive error estimation is robust with respect to arbitrary distribution.

\begin{fact}
  \label{lem:additive-error-robustness}
  Given $2N$ events $E_1, ..., E_N$, and $E'_1, ..., E'_N$, and some real $\eps$, such that for any $i \in [N]$, $|\Pr[E_i] - \Pr[E'_i]| \le \eps$.
  Then for arbitrary distribution $\mu$ over $[N]$ that is independent of the $2N$ events, $|\Pr[E_\mu] - \Pr[E'_\mu]| \le \eps$.
\end{fact}
\begin{proof}
  \begin{align*}
      \left|\Pr[E_\mu] - \Pr[E'_\mu]\right|
        &= \left|\sum_{i \in [N]} (\Pr[E_\mu \land \mu = i] - \Pr[E'_\mu \land \mu = i])\right| \\
        &\le \sum_{i \in [N]} \left|\Pr[E_\mu \land \mu = i] - \Pr[E'_\mu \land \mu = i]\right| \\
        &= \sum_{i \in [N]} \Pr[\mu = i] \cdot \left|\Pr[E_i | \mu = i] - \Pr[E'_i | \mu = i]\right| \\
        &= \sum_{i \in [N]} \Pr[\mu = i] \cdot \left|\Pr[E_i] - \Pr[E'_i]\right| \\
        &\le \sum_{i \in [N]} \Pr[\mu = i] \cdot \eps \\
        &= \eps.
  \end{align*}
\end{proof}

\begin{proposition}
  There exists a universal constant $c > 0$.
  Given a \underline{decision}  game $G$.
  Given an $(S, T)$ adversary with advice $\As$ for $G$ with winning probability $1/2 + \eps$, there exists an $(g, S', T')$ adversary with advice $\As'$ for the multi-instance game $\Gmi g$ with winning probability at least $\eps/4 \cdot (1/2 + \eps/4)^g$ for any $g > 0$, where
  \begin{equation*}
    \begin{cases}
      S' = kS, \\
      T' = 4kT, \\
      k = \frac{c}{\eps^8} \cdot S^2 \log^2 g.
    \end{cases}
  \end{equation*}
\end{proposition}
\begin{proof}
  Consider the following adversary $\As'$:
  \begin{enumerate}
    \item At the pre-processing phase, compute $\ket{\alpha_j} \gets \As_1(H)$ for all $j \in [k]$.
          We denote $\rho^\otimes = \ket{\alpha_1}, ..., \ket{\alpha_k}$ as the advice.
    \item At the online phase, compute $\st \gets \qpmwinit(k, \rho^\otimes , g, \eps/4, 1/2)$;
    \item On receiving challenge $\ch_i$ and given oracle access to $\tilde H(\cdot) := \query^H(\ch_i, \cdot)$, do: \begin{enumerate}
      \item Compute $(\ket \st, p_i) \gets \qpmwest(\ket \st, i, E_i)$, where the measurement $E_i$ is written as the unitary that computes $\As_2^{\tilde H}(\rho, \ch_i)$;
      \item Flip a biased coin that yields 1 with probability $p_i$ and 0 with probability $1 - p_i$, and send the result to the challenger;
      \item Ignore challenger's response.
    \end{enumerate}
  \end{enumerate}
  
  Note that each variable in \Cref{thm:shadow-tomography} is set to be: $m = g$, $d = 2^S$, $\beta = 1/2$, additive error is set to be $\eps / 4$, and $\beta = 1/2$. Therefore the number of copies $n$ in this theorem, or in other words, $k$ in our proposition is $\frac{c}{\eps^8} S^2 \log^2 g$. The advice is of length $S' = k S$ and the number of queries is at most $T' = 4 k T$ because for each $E_i$, the unitary (or its inverse) that computes $E_i$ will be applied at most $4 k$ times (by \Cref{thm:shadow-tomography}). 

  Using \Cref{lem:averaging-argument}, we can find a subset $H_\good$ of all possible functions such that whenever the function $H$ is in $H_\good$, $\Pr_{r, \mathcal A}[\As \Longleftrightarrow C(H) = 1] \ge 1/2 + \eps/2$.
  Since $H_\good$ takes up $\eps/2$ fraction of all possible functions, and this shadow tomography algorithm succeeds with probability $\geq 1 - \beta = 1/2$.
  We have the theorem by observing that $\mathcal A'$ given above has the success probability lower bounded by $ \eps/2 \cdot 1/2 \cdot (1/2 + \eps/2 - \eps/4)^g$ (by \Cref{lem:additive-error-robustness} and the fact that each challenge is generated independently).
\end{proof}

Combining this theorem with \Cref{thm:mis-remove-advice} (to replace a quantum advice with a maximally mixed state), we obtain the following. 

\begin{corollary}
  There exists a universal constant $c > 0$.
  Given a \underline{decision} security game $G$.
  Given an $(S, T)$ adversary with advice $\As$ for $G$ with winning probability $1/2 + \eps$, there exists an $(g, T')$ adversary $\As'$ for the multi-instance game $\Gmi g$ with winning probability at least $2^{-kS} \cdot  \eps/4  \cdot (1/2 + \eps/4)^g$ for any $g > 0$, where $T' = 4kT$ and $k = \frac{c}{\eps^8} \cdot S^2 \log^2 g$.
\end{corollary}

\begin{fact}
  \label{lem:unimportant-3}
  Given any real $C \ge 0, D \ge 2$, if $g_0 = C + D$ and $g = 16 g_0 \log g_0$, $g \ge C \log^2 g + D$.
\end{fact}
\begin{proof}
First, since $C \ge 0, D \ge 2$, we have $g_0 \ge 2$. 
\begin{align*}
    \log g = \log (2 g_0 \log^2 g_0) = 1 + \log g_0 + 2 \log \log g_0 \leq 3.9 \cdot \log g_0
\end{align*}
The above inequality is true because for $g_0 \geq 2$, $2 \log \log g_0 \leq 1.9 \cdot \log g_0$ and $1 \leq \log g_0$. Then we have, 
  \begin{align*}
      g - (C \log^2 g + D) & \geq 16 \cdot g_0 \log^2 g_0 - (C  \cdot (3.9)^2 \log^2 g_0 + D) \\
            & \geq C (16 - 3.9^2) \log^2 g_0 + D (16 \log^2 g_0 - 1)
  \end{align*}
  We complete the proof by noting that coefficients of $C$ and $D$ are both non-negative when $\log g_0 \geq 1$. 
\end{proof}

\begin{customtheorem}{\ref{thm:decision-game-quantum-advice-reduction}}
  \statethmdecisiongamequantumadvicereduction
\end{customtheorem}
\begin{proof}
  Let $\As$ be the $(S, T)$ adversary with advice  for $G$ with the best advantage $\eps'$, then there exists an $(g, T)$ adversary (without advice) $\As'$ for the multi-instance game $\Gmi g$ with winning probability at least $2^{-kS} \cdot  \eps'/4 \cdot (1/2 + \eps'/4)^g$ for any $g > 0$, where $T' = 4kT$ and $k = \frac{c}{\eps^8} \cdot S^2 \log^2 g$.
  
  By the security of multi-instance game, if $T' \ge T_0$ and $S \ge S_0$, we know that
  \[
    (1/2+\eps(g, T'))^g
      \ge 2^{-kS} \cdot  \eps'/4  \cdot (1/2 + \eps'/4)^g
      \ge 2^{-kS} \cdot  \eps_0/4  \cdot (1/2 + \eps'/4)^g.
  \]
  Therefore, in general,
  \[
  (1/2+\eps(g, T'))^g \ge 2^{-k(S + S_0)} \cdot  \eps_0/4  \cdot (1/2 + \eps'/4)^g.
  \]

  Let $C = \frac{15 \ln 2}{\eps'} \cdot c/\eps'^8 \cdot S^2 (S+S_0)$ and $D = \frac{15 \ln 2}{\eps'} (\log(1/\eps_0) + 2) \ge 2$. Let 
  \begin{align*}
      g_0  & =  \frac{15 \ln 2}{\eps'}\cdot (c/\eps'^8 \cdot S^2 (S+S_0)  + \log(1/\eps_0) + 2) = C + D, \\
      g & = 16 g_0 \log^2 g_0 \geq C \log^2 g + D = \frac{15 \ln 2}{\eps'} (k (S + S_0) + \log(1/\eps_0) + 2),
  \end{align*}
  where the inequality for $g$ comes from \Cref{lem:unimportant-3}. 
  
  \begin{align*}
      \left( \frac{1/2 + \eps' / 8}{1/2 + \eps' / 4} \right)^g & \leq \left( 1 - \frac{ \eps' / 8}{1/2 + \eps' / 8 + 1} \right)^g   \\
      & = \left( 1 - \frac{ 1 }{ 8 (1/2 + 1) / \eps' + 1} \right)^g \\
      & \leq \left( 1 - \frac{ \eps' }{ 15 } \right)^g \\
      & \leq 2^{- k (S + S_0)}\cdot \eps_0 / 4,
  \end{align*}
  where the first inequality comes from \Cref{lem:unimportant-1} and the last inequality comes from $g \geq  \frac{15 \ln 2}{\eps'} (k (S + S_0) + \log(1/\eps_0) + 2)$. Overall, we have, 
    \begin{align*}
  (1/2+\eps(g, T'))^g \ge 2^{-k(S + S_0)} \cdot  \eps_0/4  \cdot (1/2 + \eps'/4)^g \geq (1/2 + \eps'/8)^g.
  \end{align*}
  In conclusion, $\eps'(S, T) \leq 8 \cdot \eps(g, T')$ where $g = \tilde O(S + S_0)^3/\eps'^9$ and $T' = 4 k T = \tilde O((T + T_0)(S + S_0)^2)/\eps'^8$.
\end{proof}

\section{One-Way Function}
\label{sec:owf}

In this section, we show the security of OWF in the AI-QROM and QAI-QROM. Let us start with reformalizing the security of OWF in QROM in our language of security games.

\begin{definition}[OWF Security Game]
    Let $H$ be a random oracle $[N] \to [M]$. The security game $G_\owf = (C_\owf)$ is specified by three procedures $(\samp, \query, \ver)$, where: 
    \begin{enumerate}
        \item $\samp^H(x)$ takes a random element $x \in [N]$, and outputs $y = H(x)$.
        \item $\query^H(x, x')$ ignores the challenge point $x$ and simply outputs $H(x')$.
        \item $\ver^H(x, x')$ outputs $1$ if and only if $H(x) = H(x')$. 
    \end{enumerate}
\end{definition}
To fill in the rest of the details from \Cref{def:security-game}, the security game defined here is that given a random oracle, the challenger samples a uniformly random element $x \in [N]$ and computes $y = H(x)$ as the challenge. The adversary is allowed to make oracle queries to $H$ and it wins if and only if it finds any pre-image of $y$.
This is exactly the security of OWF in QROM.
This is also exactly the problem of inverting a random image.

In the rest of the section, we are going to prove the following lemma.
\newcommand{\lemstateowfmis}{$\Gmi{g}_\owf$ is $\delta(g, T) = O\left(\frac{g T + T^2}{\min\{N, M\}} \right)$-secure in the QROM.}
\begin{lemma} \label{lem:owfmis}
  \lemstateowfmis
\end{lemma}
Combining this lemma with \Cref{thm:reduction-classical-advice-general} (or \Cref{thm:public-ver-reduction}), we can show the security of $G_\owf$ in the AI-QROM (or QAI-QROM, respectively).
\begin{theorem} \label{thm:owfaiqrom}
    $G_\owf$ is $\delta(S, T) = O\left( \frac{(S + T + \log N) T}{\min\{N, M\}} \right) =  \tilde{O}\left( \frac{S T + T^2}{\min\{N, M\}} \right)$-secure in the \varul{AI-QROM}. 
\end{theorem}
\begin{proof}
    In \Cref{thm:reduction-classical-advice-general}, $\delta_0$ is at least $1 / N$ because there is at least one pre-image of $y$, so random guess has winning probability at least $1/N$. Plugging in $\delta_0 \ge 1/N$, and $\Gmi{g}_\owf$ is $O\left(\frac{g T + T^2}{\min\{N, M\}}\right)$-secure from \Cref{lem:owfmis}, we have $G_\owf$ is $\delta(S, T)$-secure where, 
    \begin{align*}
        \delta(S, T) &\leq 4 \cdot O\left(\frac{(S + \log(1/\delta_0) + 1) T + T^2}{\min\{N, M\}}\right) \\
                     & = O\left( \frac{(S + \log N) \cdot T + T^2}{\min\{N, M\}} \right).%
    \end{align*}
\end{proof}

\begin{theorem} \label{thm:owfqaiqrom}
    $G_\owf$ is $\delta(S, T) = \tilde{O}\left( \left(\frac{S T + T^2}{\min\{N, M\}} \right)^{1/3} \right)$-secure in the QAI-QROM, where $\tilde O$ absorbs $\poly(\log N, \log M, \log S)$ factors.
\end{theorem}
\begin{proof}
    It is a publicly-verifiable game. 
    For $G_\owf$, we can verify that $T_\ver = 2$ and $\delta_0$ is again at least $1 / N$.
    Since $\Gmi{g}_\owf$ is $\delta'(S, T) = O\left(\frac{g T + T^2}{\min\{N, M\}}\right)$-secure from \Cref{lem:owfmis}, we have $G_\owf$ is $\delta(S, T)$-secure where, 
    \begin{align*}
        \delta(S, T) &\leq 8 \cdot\delta'\left(  {\tilde O (S)} / {\delta}, {\tilde O (T + T_\ver)} / {\delta} \right) \\ 
        &\leq 8 \cdot \frac{ \tilde{O} \left(S T + T^2 \right)}{\delta^2 \cdot \min\{N, M\}}.
    \end{align*}
Therefore, $\delta(S, T) = \tilde{O}\left( \left(\frac{S T + T^2}{\min\{N, M\}} \right)^{1/3} \right)$. 
\end{proof}

In the following two subsections, we will first show multi-instance security for a slightly different game with $\delta(g, T) = O(\frac{g T + T^2}{M})$. Then in the second subsection, we show that these two games are almost indistinguishable even for a quantum adversary, and conclude that the  multi-instance security for $\Gmi{g}_{\owf}$  is $\delta(g, T) = O\left(\frac{g T + T^2}{\min\{N, M\}}\right)$.

\subsection{Multi-Instance Security for Inverting Random Image}

Consider the following OWF game variant, where the main difference is that instead of uniformly sampling $x \in [N]$ and telling the adversary $H(x)$, we are directly sampling a uniformly random $y \in [M]$ and telling the adversary $y$.

\begin{definition}[$\text{OWF}_y$ Security Game]
    Let $H$ be a random oracle $[N] \to [M]$. The security game $G_{\owfy} = (C_{\owfy})$ is specified by the challenger $C_{\owfy} = (\samp, \query, \ver)$, where: 
    \begin{enumerate}
        \item $\samp^H(y) = y$ takes a randomness $y \in [M]$, and outputs $y$ as the challenge.
        \item $\query^H(y, x') = H(x')$. 
        \item $\ver^H(y, x')$ outputs $1$ if and only if $H(x') = y$. 
    \end{enumerate}
    
    In other words, given a random oracle, the challenger samples a random output $y$ as the challenge. The adversary is allowed to make oracle queries to $H$ and it wins if and only if it finds a pre-image of $y$. 
\end{definition}

In $\text{OWF}_y$ problem, a challenge $y$ may not have a valid answer. For example, when $N \ll M$, a random element $y$ in $[M]$ is not the image of any inputs with high probability. 

\begin{lemma}
    ${G}^{\otimes g}_{\owfy}$ is $\delta(g, T) = O\left(\frac{g T + T^2}{M} \right)$-secure in the QROM.
\end{lemma}
\begin{proof}
    Let $\As$ be any $(g, T)$ adversary.
    By definition of security, to prove the lemma, we would need to establish an upper bound on $\Pr[\As \Longleftrightarrow C^{\otimes g}_\owfy(H)]$.
    Without loss of generality, we can write $\As$ as a sequence of unitary operators $(U_{i, j})_{i \in [g], j \in [T + 1]}$ operating on registers $\ket x, \ket y, \ket\aux$, initialized to empty, where
    \begin{itemize}
        \item $x \in [N], u \in [M]$ and $\aux$ is the auxiliary register whose size is unbounded.
        \item $\ket u$ serves as the challenge as input to $U_{i, 1}$.
        \item For $1 \le t \le T$, $U_{i, t}$ prepares a query to $\query^H(y, \cdot) = H(\cdot)$ at $\ket x$ and receives the output at $\ket u$.
        \item $U_{i, T + 1}$ prepares the final output for $i$-th game in register $\ket x$, clears the register $\ket u$, and receives the measured answer also back in the register $\ket x$ and the next challenge in the register $\ket u$.
    \end{itemize}
    
    Consider the simulator that instead simulates the interaction of $\As$ and the challenger with the compressed oracle.
    Formally, the simulator initializes the state $\ket{\psi'_{1, 0}} := \ket 0_\As \otimes {\ket \emptyset}_{H}$.
    For each round $i \in [g]$:
    \begin{enumerate}
        \item Starting at $\ket{\psi'_{i, 0}}$, sample the next challenge $y_i \gets_\$ [M]$, and let $\ket u \gets y_i$.
            Let the resulting state \textit{conditioned on} the challenge being $y_i$ be $\ket{\psi_{i, 0}}$.
        \item We define $\ket{\psi_{i, t}}$ and $\ket {\psi'_{i, t}}$:
            \begin{alignat*}{3}
                \ket {\psi'_{i, t + 1}} &= U_{i, t+1}  \ket {{\psi}_{i, t}}  &\qquad& \text{for all $t = 0, \cdots, T$;}  \\
                \ket {{\psi}_{i, t}} &= U_H \ket {\psi'_{i, t}}  && \text{for all $t = 1, \cdots, T$;}  
            \end{alignat*}
            where $U_H$ is $\csto$ operating on registers $\ket x, \ket u, \ket H$.
        \item By our assumption on $U_{i, T+1}$, in state $\ket{\psi'_{i, T+1}}$, $\ket x$ holds the adversary's answer and $\ket u = \ket 0$.
          Verify the answer by doing the following: apply $U_H$ on $\ket x, \ket u$. Let us call the resulting state as $\ket{\psi_{i, T+1}}$, then
        \begin{align*}
            \ket{\psi_{i, T+1}} = \sum_{{x, u, \aux, D}} \gamma_{x, u, \aux, D} \ket {x, u, \aux}_\As \ket D_H
        \end{align*}
        \item Compute in superposition $b_i = [u = y_i]$ and measure $b_i$.
        \item Uncompute $U_H$, let the resulting state \textit{conditioned on} measurement result being $b_i$ be $\ket{\phi'_{i + 1, 0}}$.
    \end{enumerate}
    By \Cref{cor:csto-simulate}, the probability that this algorithm obtains $b_1 = \cdots = b_g = 1$ is exactly the probability that $\Pr[\As \Longleftrightarrow C^{\otimes g}_\owfy(H)]$.
    We have the following proposition.

    \begin{proposition} \label{lem:owf-single-stage}
        For all $i \in [g]$, let $Y_1, Y_{2}, ..., Y_{i-1}$ be the random variable of the first $(i-1)$ challenges and $B_1, ..., B_{i}$ be the random variable for the decision bits. For any $(g, T)$ algorithm in the QROM, $y_1, y_2, ..., y_{i-1}$ and $b_1, b_2, ..., b_i$, we have, 
        \begin{align*}
            \Pr\left[ B_i = 1 \,|\, Y_1 = y_1,  ..., Y_{i-1} = y_{i-1}, B_1 = b_1, ..., B_{i-1} = b_{i - 1} \right] \leq O\left( \frac{i T + T^2}{M} \right),
        \end{align*}
        given that $\Pr[Y_1 = y_1,  ..., Y_{i-1} = y_{i-1}, B_1 = b_1, ..., B_{i-1} = b_{i - 1}] > 0$.
    \end{proposition}
    \begin{proof}
        We start by observing that these conditioning correspond to the pure state $\ket{\psi'_{i, 0}}$ that we have described above.
        The state is well defined as these classical outcomes occur with non-zero probability as we have assumed.
        Since the total number of queries up to this point is at most $(i-1) (T+2)$, by \Cref{lem:bounded-database}, there exists some complex numbers $\alpha_*$ such that
        \[
          \ket{\psi'_{i, 0}} = \sum_{z, D : |D| \leq (i-1) (T+2)} \alpha_{z, D} \ket z \ket D.
        \]

        For any $y_i \in [M]$, define $P_{y_i}$ to be a projection of finding $y_i$ in the database: 
        \begin{align*}
            P_{y_i} = \sum_{D: \exists x \in [N], D(x) = y_i} \ket D \bra D.
        \end{align*}
        Let $p^{(t)}_{y_i}$ be the probability that after $t$ queries to $\csto$, measuring the database register gives a database containing $y_i$, or formally,
        \begin{align*}
            p^{(t)}_{y_i} := \left| P_{y_i} \ket{\psi_{i, t}} \right|^2.
        \end{align*}
        
        Since $\ket {\psi'_{i, 0}}$ only has non-zero amplitude over $D$ with size at most $(i-1) (T+2)$, for each ${z, D}$ pair, $|\alpha_{z, D}|^2$ contributes to at most $(i-1) (T+2)$ different $y_i$'s.
        We observe that the same conclusion also holds for $\ket{\psi_{i, 0}}$, as the only difference between these two states is that $u$ is initialized to some element in $[M]$, and $H$ register is unchanged.
        In other words,
        \begin{equation} \label{eq:unimportant-baddatabase}
            \sum_{y_i=1}^M p^{(0)}_{y_i} \leq (i-1) (T+2). 
        \end{equation}
        
        We then have the following lemma that bounds $p^{(t)}_{y_i}$ for all $t > 0$, which is the same as \cite[Theorem 1]{zhandry2019record} and similar to \cite[Lemma 8 and Lemma 9]{liu2019finding}. 
        \begin{lemma}
            For all $t = 0, 1, \cdots, T$ and all $y_i \in [M]$, $\sqrt{p^{(t + 1)}_{y_i}} \leq \sqrt{p^{(t)}_{y_i}} + \sqrt{1/M}$. 
        \end{lemma}
        \begin{proof}
           Let us define $q^{(t)}_{y_i} = \left| P_{y_i} \ket{\psi'_{i, t}} \right|^2$. First, similar to \cite[Lemma 8]{liu2019finding}, we have $q^{(t+1)}_{y_i} = p^{(t)}_{y_i}$, since
           \begin{align*}
               \sqrt{q^{(t + 1)}_{y_i}} = \left| P_{y_i} \ket{\psi'_{i, t + 1}} \right| = \left| P_{y_i} U_{i, t+1} \ket{\psi_{i, t}} \right| \overset{(*)}{=} \left| P_{y_i} \ket{\psi_{i, t}} \right| = \sqrt{p^{(t)}_{y_i}}.
          \end{align*}
           The equality $(*)$ follows from the fact that $U_{i, t+1}$ only applies on $\As$'s register but $P_{y_i}$ only applies on the database register and therefore they commute. 
           
          Second, we have $\sqrt{p^{(t+1)}_{y_i}} \leq \sqrt{q^{(t+1)}_{y_i}} + \sqrt{1/M}$~\cite[Lemma 9]{liu2019finding}:
          \begin{align*}
              \sqrt{p^{(t + 1)}_{y_i}} = \left| P_{y_i} \ket{\psi_{i, t + 1}} \right| &= \left| P_{y_i} U_H \ket{\psi'_{i, t + 1}} \right|  \\
              &= \left| P_{y_i} U_H \,(P_{y_i} + (I - P_{y_i})) \ket{\psi'_{i, t + 1}} \right| \\
              &\leq \left| P_{y_i} U_H  P_{y_i} \ket{\psi'_{i, t + 1}} \right| + \left| P_{y_i} U_H  (I - P_{y_i}) \ket{\psi'_{i, t + 1}} \right|.
          \end{align*}
          
          Intuitively, the first term corresponds to the case that before making the $(t+1)$-th query, the database already has $y_i$. The second term corresponds to the case that before making the $(t+1)$-th query, the database does not have $y_i$.
          
          The first term $\left| P_{y_i} U_H  P_{y_i} \ket{\psi'_{i, t + 1}} \right| \le \left| U_H  P_{y_i} \ket{\psi'_{i, t + 1}} \right| = \left| P_{y_i} \ket{\psi'_{i, t + 1}} \right| = \sqrt{q^{(t+1)}_{y_i}}$. 
          
          For the second term, let $\ket \phi := (I - P_{y_i}) \ket{\psi'_{i, t + 1}}$, by definition of $P_{y_i}$, there exists some complex numbers $\beta_*$ such that
          \begin{align*}
              \ket \phi = \sum_{\substack{x, u, \aux, \\ D : \forall x', D(x') \ne y_i}} \beta_{x, u, \aux, D} |x, u, \aux\rangle \ket D.
          \end{align*}
          We can equivalently write the second term as
          \begin{align*}
              &\equad \left| P_{y_i} U_H \ket \phi \right| \\
              &= \Bigg| P_{y_i} \csto   \sum_{\substack{x, u, \aux, \\ D : \forall x', D(x') \ne y_i}} \beta_{x, u, \aux, D} |x, u, \aux\rangle \ket D \Bigg| \\
              &= \Bigg| P_{y_i} \sum_{\substack{x, u, \aux, \\ D : \forall x', D(x') \ne y_i}} \beta_{x, u, \aux, D}\,  \stddecomp_x \circ \csto' \circ \stddecomp_x \, |x, u, \aux\rangle \ket D \Bigg| \\
              &= \Bigg| P_{y_i} \,  \sum_{\substack{x, u, \aux, \\ D : \forall x', D(x') \ne y_i}} \sum_y \frac{\beta_{x, u, \aux, D}}{\sqrt{M}}\,  \stddecomp_x  \, |x, u + y, \aux\rangle \ket {D \cup (x, y)} \Bigg| \\
              &\overset{(*)}{\leq} \Bigg| P_{y_i} \,  \sum_y \sum_{\substack{x, u, \aux, \\ D : \forall x', D(x') \ne y_i}} \frac{\beta_{x, u, \aux, D}}{\sqrt{M}}\, |x, u + y, \aux\rangle \ket {D \cup (x, y)} \Bigg| \\
              &= \Bigg| \sum_{\substack{x, u, \aux, \\ D : \forall x', D(x') \ne y_i}} \frac{\beta_{x, u, \aux, D}}{\sqrt{M}}\, |x, u + y_i, \aux\rangle \ket {D \cup (x, y_i)} \Bigg| \\
              &= \frac{1}{\sqrt{M}} |\ket \phi|
              \leq \frac{1}{\sqrt{M}}.
          \end{align*}
          The inequality (*) comes from the fact that $D \cup (x, y)$ already contains $x$, therefore $\stddecomp_x$ will either remove the entry for $x$ or not change the database at all. 
        \end{proof}
        Using an induction on the lemma, we have
        \begin{equation}
          \label{eq:owf-bound-ptplus1}
          \sqrt{p^{(T+1)}_{y_i}} \le \sqrt{p^{(0)}_{y_i}} + (T+1) / \sqrt{M}.
        \end{equation}
        
        Define $\tilde{p}_{y_i}$ as the probability that $B_i = 1$ conditioned on the challenge is $y_i$. In other words, $\tilde{p}_{y_i}$ is the probability that the final state $x, u$ when measured gives $u = H(x) = y_i$.
        By \Cref{lem:zhandry-lemma5}, we have $\sqrt{\tilde{p}_{y_i}} \leq \sqrt{{p}^{(T+1)}_{y_i}} + 1/\sqrt{M}$.
        Combining this with \eqref{eq:unimportant-baddatabase}, \eqref{eq:owf-bound-ptplus1}, and Cauchy-Schwarz inequality,
        \begin{align*}
            \frac{1}{M} \sum_{y_i=1}^{M} \tilde{p}_{y_i} \leq \frac{2}{M} \, \sum_{y_i=1}^{M} \left({{p}^{(0)}_{y_i}} + (T+2)^2 / M \right) \leq \frac{O(i T + T^2)}{M}.
       \end{align*}
       We complete the proof by observing that the left hand side is exactly the probability in the statement.
    \end{proof}
    
    Combine \Cref{lem:owf-single-stage} and \Cref{lem:conditioning2mis}, we have proven the lemma.
\end{proof}

\subsection{Multi-Instance Indistinguishability of Sampling Image}

The analysis for  multi-instance security of $\Gmi{g}_\owfy$ is not enough for $\Gmi{g}_\owf$. This is because in the standard one-way function game, $y$ is sampled in the way that it is correlated with the random oracle. Therefore, we can not use \Cref{lem:zhandry-lemma5}, since the challenge image $y$ is almost always in the database.
In this subsection, we will instead show that the advantage of indistinguishing whether $y_i$ is sampled uniformly at random or from $H(x_i)$ for a uniformly random $x_i$ is at most $O((i T + T^2)/N)$. 

First, we need the following lemma, which says for any quantum algorithm $\As$ making at most $q$ queries to a random oracle before getting a challenge, and at most $q'$ queries after, it can not distinguish if it gets a random image $H(x)$ or a random element $y$ in the range. 
Furthermore, the statement is true even after conditioning on arbitrary intermediate measurements.
\begin{lemma} \label{lem:prgind}
Let $\Ds$ be a quantum algorithm that makes at most $q + q'$ queries to a random oracle. Let $\vuy_{[p]} = Y_1, Y_{2}, ..., Y_{p}$ be the random variable of the outcomes of $p$ arbitrary intermediate measurements, all of which occur before the $q$-th query. After $\Ds$ making $q$ queries, it is given a challenge $y$, which we crudely denote as $\Ds^H(y)$. It then makes another $q'$ queries and simply outputs a bit.
For any list of outcomes $\mathbf{y}_{[p]} = y_1, ..., y_p$, 
\begin{align*}
    \Bigg|& \sqrt{\Pr_{H, y \gets_{\$} [M]}  \left[ \Ds^{H}(y) = 1\,|\, \vuy_{[p]} = \vy_{[p]}  \right]} \\
    & \quad\quad\quad -  \sqrt{\Pr_{H, x \gets_{\$} [N]} \left[\Ds^{H}(H(x)) = 1\,|\, \vuy_{[p]} = \vy_{[p]} \right]} \Bigg| \leq 2 (\sqrt{q} + q')  / \sqrt{N}.
\end{align*}
\end{lemma}
\begin{proof}
    
    We will first present two slightly different oracles or experiments, and argue that each of them perfectly simulates $\Ds^H(y)$ and $\Ds^H(H(x))$ respectively.
    
    \paragraph*{Case 1.}
    The challenge is a uniformly random $y$.
    We implement the following changes to the oracle:
    \begin{enumerate}
      \item We replace the random oracle with a phase oracle.
            By \Cref{lem:pho-simulate}, this change perfectly preserves the probability that $\Ds$ outputs 1 as he only interacts with $\ket H$ using the phase oracle.
      \item We consider the modified oracle $\tilde H: [N + 1] \mapsto [M]$, where $\tilde H := H || h'$ for some random element $h' \in [M]$, and we pick $h' = \tilde H(N + 1)$ as our challenge.
            As $\Ds$ only gets phase oracle access to $H$, and we will only use $h'$ as classical randomness, this change again perfectly preserves $\Ds$'s functionality.
      \item We additionally prepare a random $x \in [N]$ at a new register $x$ but we do not use it.
            This change is clearly only conceptual.
    \end{enumerate}
    Conditioned on $\vuy_{[p]} = \vy_p$, after $\Ds$ has made $q$ queries to the phase oracle and before generating the challenge, by \Cref{lem:bounded-database-phase}, we know that the overall state of $\Ds$ and the oracle internal state can be written as\footnote{Here we write $\ket x$ for simplicity, but to see correctness, you could really think of $\ket{x, x}$ so we have two computational-basis ``copies'' of the same uniform superposition $1/\sqrt N \sum_x \ket x$. As we will never touch the second copy, the first copy behaves exactly like a classical random string to the rest of the algorithm.}
    \begin{align*}
        \sum_{z, D: |D| \leq q} \alpha_{z, D}\, |z\rangle_{\Ds} \ket 0_y \otimes \frac{1}{\sqrt{N}} \sum_x |x\rangle \otimes \, \frac{1}{M^{(N+1)/2}} \sum_{H, h'} \omega_{M}^{\langle D, H\rangle} |H, h'\rangle,
    \end{align*}
    where $x$ and $h'$ have not been used and thus is completely un-entangled with $\Ds$.
    After generating the challenge, the overall state becomes
    \begin{align*}
          \ket {\phi_1} := \frac{1}{N^{1/2} M^{(N+1)/2}} \sum_{z, D: |D| \leq q} \sum_{x} \sum_{H, h'} \alpha_{z, D}\, \ket z_\Ds \ket{h'}_y \otimes |x\rangle \otimes \omega_{M}^{\langle D, H\rangle} |H, h'\rangle.
    \end{align*}
    
    \paragraph*{Case 2.}
    The challenge is the image for a uniformly random pre-image.
    We implement the following changes to the oracle:
    \begin{enumerate}
      \item We again replace the random oracle with a phase oracle using \Cref{lem:pho-simulate}.
      \item We again augment $H$ to be $\tilde H: [N + 1] \mapsto [M]$ for a new random element $h' \in [M]$, and prepare a new random $x \in [N]$.
            For now, this is only a conceptual change (we only added two new independent registers).
      \item Define a new classical oracle $H'_x: [N] \mapsto [M]$ to be
            \[
              H'_x(x') := H_{< x} || h' || H_{> x} = \begin{cases}
                \tilde H(N + 1) = h', & x' = x; \\
                \tilde H(x') = H(x'), & x' \neq x.
              \end{cases}
            \]
            Instead of giving $\Ds$ phase oracle to $H$, we are giving $\Ds$ phase oracle to $H'_x$, where $x$ is pulled from the new register $\ket x$ that we just add.
            For any fixed $x$, this change is perfectly consistent with the original oracle distribution; and as $\ket x$ is just a classical random string, the overall change is also perfectly consistent.
      \item Instead of sampling a new pre-image for generating the challenge, we instead use $x$ as the pre-image, that is, we ``compute'' $H'_x(x) = h'$ and send $h'$ as the challenge for $\Ds$.
            It is easy to see that $\Ds$ cannot obtain any information on $x$ for the first $q$ queries, thus this change still perfectly preserves $\Ds$'s functionality.
    \end{enumerate}
    Therefore, conditioned on $\vuy_{[p]} = \vy_p$, before giving $\Ds$ the challenge, the overall state of $\Ds$ and the oracle internal memory is
    \begin{align*}
        \sum_{z, D: |D| \leq q} \alpha_{z, D}\, |z\rangle_{\Ds} \ket 0_y \otimes \frac{1}{\sqrt{N}} \sum_x |x\rangle \otimes \, \frac{1}{M^{(N+1)/2}} \sum_{H, h'} \omega_{M}^{\langle D, H'_x\rangle} |H, h'\rangle.
    \end{align*}
    When the challenge is given, the overall state becomes
    \begin{align*}
          \ket {\phi_2} := \frac{1}{N^{1/2} M^{(N+1)/2}} \sum_{z, D: |D| \leq q} \sum_{x} \sum_{H, h'} \alpha_{z, D}\, |z\rangle_{\Ds} \ket{h'}_y \otimes |x\rangle \otimes \omega_{M}^{\langle D, H'_x\rangle} |H, h'\rangle.
    \end{align*}

    By now, we have shown two perfect simulator for $\Ds(y)$ and $\Ds(H(x))$ respectively, but instead of the two cases getting different challenges ($y$ vs $H(x)$), we have designed the simulator in the way that the difference now is instead which oracle $\Ds$ is interacting with ($H$ vs $H'_x$).
    Immediately after $\Ds$ is given the challenge, the overall states of the two simulators after conditioning are $\ket{\phi_1}, \ket{\phi_2}$ respectively.
    Intuitively, the problem now reduces to distinguishing these two oracles.
    
    \begin{lemma}
       $\ltwonorm{\ket {\phi_1} - \ket {\phi_2}} \leq 2 \cdot \sqrt{q/N}$. 
    \end{lemma}
    \begin{proof}
        We compare two vectors component-wise. Fixing $z, D, H, h'$, since $H$ and $H'_x$ only differ at point $x$, if $x$ satisfies $D(x) = 0$, then by how we defined the inner product,
        \[
          \omega_M^\innerprod{D, H} = \omega_M^\innerprod{D, H'_x}.
        \]
        Therefore, we can write these two vectors as $\ket {\phi_1} = \ket {\phi} + \ket {\err_1}$ and  $\ket {\phi_2} = \ket {\phi} + \ket {\err_2}$, where $\ket {\phi}$ only sums over $x$ such that $D(x) = 0$.  
        Their difference is bounded by the following, 
        \begin{align*}
            \ltwonorm{\ket {\phi_1} - \ket {\phi_2}}
            &= \ltwonorm{\ket {\err_1} - \ket {\err_2}} \\
            &\leq \ltwonorm{\frac{1}{N^{1/2} M^{(N+1)/2}} \sum_{z, D: |D| \leq q} \sum_{H, h'} \sum_{x : D(x) \ne 0} \alpha_{z, D} \, \omega_{M}^{\langle D, H\rangle} \, \ket {z, h'} \ket x \ket {H, h'}} + \\
            &\equad \ltwonorm{\frac{1}{N^{1/2} M^{(N+1)/2}} \sum_{z, D: |D| \leq q} \sum_{H, h'} \sum_{x : D(x) \ne 0} \alpha_{z, D} \,  \omega_{M}^\innerprod{D, H'_x} \, \ket {z, h'} \ket x \ket {H, h'}} \\ 
            &\leq 2 \cdot \sqrt{\frac{1}{M^{N+1}} \sum_{z, D: |D| \leq q} \sum_{H, h'}  \frac{|D|}{N}\cdot  |\alpha_{z, D}|^2} \\
            &\leq 2 \cdot \sqrt{\frac{q}{N} \cdot \sum_{z, D: |D| \leq q} |\alpha_{z, D}|^2} \\
            &= 2 \sqrt{q/N},
        \end{align*}
        where the last equality is due to the fact that the $\alpha_*$'s are defined as \Cref{lem:bounded-database-phase}.
    \end{proof}

    Therefore, before making the next $q'$ queries, the overall state is $\ket {\phi_1}$ if it is in \emph{Case 1} or $\ket {\phi_2}$ if it is in \emph{Case 2}. Let $U_H$ be the oracle query to $H$ and $U_{H'_x}$ be the oracle query to $H'_x$.
    We can write the rest of $\Ds^H$ as $U_{q'} U_H \cdots U_1 U_H U_0$. %
    First, by the fact that unitary preserves $L^2$-norm and triangle inequality, we have,
    \begin{align*}
        &\equad \ltwonorm{U_{q'} U_H \cdots U_1 U_H U_0 \ket {\phi_1} - U_{q'} U_{H'_x} \cdots U_1 U_{H'_x} U_0 \ket {\phi_2}} \\
        &\leq \ltwonorm{U_{q'} U_H \cdots U_1 U_H U_0 \ket {\phi_1} - U_{q'} U_{H'_x} \cdots U_1 U_{H'_x} U_0 \ket {\phi_1}} + 2 \sqrt{q/N} \\
        &\leq \sum_{j=0}^{q'-1} \ltwonorm{\ket {\psi_j} - \ket {\psi_{j+1}}} + 2 \sqrt{q/N},
    \end{align*}
    where $\ket {\psi_j}$ is defined as applying $U_{H}$ on $\ket {\phi_1}$ for the first $j$ oracle queries and then $U_{H'_x}$ for the rest $(q' - j)$ oracle queries. 
    \begin{lemma}
        For all $j = 0, ..., q' - 1$, $\ltwonorm{\ket {\psi_j} - \ket {\psi_{j+1}}} \leq 2 / \sqrt{N}$. 
    \end{lemma}
    \begin{proof}
        Before making the $(j+1)$-th oracle query, let $z = (x', u, \aux)$ where $x', u$ corresponds to the input to the phase oracle, by \Cref{lem:bounded-database-phase}, the overall state can be written as
        \begin{align*}
            \frac{1}{M^{N/2}} \sum_{z, D: |D| \leq q + i}  \sum_{H} \beta_{z, D}\, |x', u, \aux\rangle_{\Ds}  \otimes \omega_{M}^{\langle D, H\rangle} |H, h'\rangle \otimes  \frac{1}{\sqrt{N} } \sum_x  \ket x.
        \end{align*}
        Note that we can make this reordering since $H'_x$ is never invoked, and $\ket x$ is never used and therefore unentangled with other registers.
        
        For the $(j+1)$-th oracle query, either $U_H$ or $U_{H'_x}$ is applied. We find that two oracle queries are different if and only if $x = x'$. Therefore,  the difference is bounded by two times the norm of the component whose $x = x'$. Since $\ket x$ is un-entangled with other registers, the norm is at most $1/\sqrt{N}$. Therefore, $\left| \ket {\psi_j} - \ket {\psi_{j+1}} \right| \leq 2 / \sqrt{N}$. 
    \end{proof}
    
    Since each  $\left| \ket {\psi_j} - \ket {\psi_{j+1}} \right| \leq 2 / \sqrt{N}$, we conclude that 
    \begin{align*} 
        \ltwonorm{U_{q'} U_H \cdots U_1 U_H U_0 \ket {\phi_1} - U_{q'} U_{H'} \cdots U_1 U_{H'} U_0 \ket {\phi_2}} \leq 2 (q' + \sqrt{q}) / \sqrt{N}.
    \end{align*} 
    By perfect simulation, we have
    \begin{align*} 
        \sqrt{\Pr_{y \gets_{\$} [M]}  \left[ \Ds^{\ket H}(y) = 1\,|\, \vuy_{[p]} = \vy_{[p]}  \right]} &= \ltwonorm{P_1 U_{q'} U_H \cdots U_1 U_H U_0 \ket {\phi_1}}\\ 
        \sqrt{\Pr_{x \gets_{\$} [M]}  \left[ \Ds^{\ket H}(H(x)) = 1\,|\, \vuy_{[p]} = \vy_{[p]}  \right]} &= \ltwonorm{P_1 U_{q'} U_{H'} \cdots U_1 U_{H'} U_0 \ket {\phi_2}}.
    \end{align*}
    Our lemma follows.
\end{proof}

\subsection{Multi-Instance Security of OWF}

Now with \Cref{lem:owf-single-stage} and \Cref{lem:prgind}, we are ready to prove the multi-instance security of $\Gmi{g}_\owf$ that we have claimed at the beginning.

\begin{customlemma}{\ref{lem:owfmis}}
    \lemstateowfmis
\end{customlemma}
\begin{proof}
    Let $B_i$ be the random variable for $b_i$, indicating if $\As$ passes the $i$-th round of the multi-instance game. Let $Y_i$ be the random variable for the $i$-th challenge. Define $B'_i$ be the random variable that $\As$ finds a pre-image of the $i$-th challenge if $Y_i$ is instead distributed uniformly at random from $[M]$.

    From \Cref{lem:prgind} and setting $q = (i-1) (T+2)$ and $q' = T+1$, we have, for any challenge $\vy_{[i-1]}$ (with non-zero probability conditioned on $\vub_{[i-1]} = \mathbf{1}^{i-1}$), 
    \begin{align*}
        \bigg| & \sqrt{\Pr  [B_i = 1\,|\, \vuy_{[i-1]} = \vy_{[i-1]}, \vub_{[i-1]} = \mathbf{1}^{i-1}]} \\
         & \quad\quad\quad\quad - \sqrt{\Pr[B'_i = 1\,|\, \vuy_{[i-1]} = \vy_{[i-1]}, \vub_{[i-1]} = \mathbf{1}^{i-1}]} \bigg| \leq O\left(\frac{\sqrt{i T} + T}{\sqrt{N}}\right).
    \end{align*}
    Since $B'_i = 1$ is the single stage $G_\owfy$ game, by \Cref{lem:owf-single-stage}, we have, 
    \begin{align*}
        \Pr[B'_i = 1\,|\, \vuy_{[i-1]} = \vy_{[i-1]}, \vub_{[i-1]}] = O\left(\frac{i T + T^2}{M} \right).
    \end{align*}
    Then, by Cauchy-Schwarz, we have, 
    \begin{align*}
        &\equad \Pr[B_i = 1\,|\, \vuy_{[i-1]} = \vy_{[i-1]}, \vub_{[i-1]} = \mathbf{1}^{i-1}] \\
        &\leq 2  \Pr[B'_i = 1\,|\, \vuy_{[i-1]} = \vy_{[i-1]}, \vub_{[i-1]}] + O\left(\frac{i T + T^2}{N} \right) \\
        &\leq O\left(\frac{i T + T^2}{M} + \frac{i T + T^2}{N} \right) = O\left(\frac{g T + T^2}{\min\{M, N\}} \right).
    \end{align*}
    Finally, we complete the proof by invoking \Cref{lem:conditioning2mis}.
\end{proof}

\subsection{Security of PRG with Advice}

We note that \Cref{lem:prgind} essentially implies a security bound for PRG with advice.
Let us start by reformalizing PRG security in our security game framework.

\begin{definition}[PRG Security Game]
    Let $H$ be a random oracle $[N] \to [M]$ where $N < M$. The security game $G_\prg = (C_\prg)$ is specified by three procedures $(\samp, \query, \ver)$, where: 
    \begin{enumerate}
        \item $\samp^H(b, x, y)$ takes a random bit $b \in \{0, 1\}$, an element $x \in [N]$ and an $y \in [M]$, and outputs $H(x)$ if $b = 0$, and $y$ otherwise.
        \item $\query^H((b, x, y), x') = H(x')$ simply queries $H$ directly.
        \item $\ver^H((b, x, y), b')$ outputs $1$ if and only if $b = b'$. 
    \end{enumerate}
\end{definition}
To fill in the rest of the details from \Cref{def:security-game}, the security game defined here is that given a random oracle, the adversary is asked to distinguish a random $y \in [M]$ from a random image $H(x)$.
The adversary is allowed to make oracle queries to $H$ and it wins if and only if he distinguishes the two distributions correctly.
This is exactly the security of PRG in QROM.

Using \Cref{lem:prgind}, we have the following lemma:
\begin{lemma} \label{lem:prgmis}
  $\Gmi{g}_\prg$ is $\delta(g, T) = \frac 1 2 + O\left(\sqrt{\frac{gT+T^2}N}\right)$-secure in the QROM.
\end{lemma}
\begin{proof}
    Consider in round $i$, conditioned on random outcomes in the previous rounds $(b_j, x_j, y_j, b'_j)_{j < i}$, which we simply denote as event $E$.
    Let $p_{b, b'}$ be the probability that challenge is $b$ and the algorithm outputs $b'$.
    By \Cref{lem:prgind} and Cauchy-Schwarz, for any $b'$ we have that
    \begin{align*}
        |p_{0, b'} - p_{1, b'}|
          &\le |\sqrt{p_{0, b'}} - \sqrt{p_{1, b'}}| \cdot |\sqrt{p_{0, b'}} + \sqrt{p_{1, b'}}| \\
          &\le 2 \frac{\sqrt{(i - 1)(T + 1)} + T}{\sqrt N} \cdot 2 \\
          &\le 4 \sqrt 2 \cdot \sqrt{\frac{g(T + 1) + T^2}N}.
    \end{align*}
    Therefore, $|\Pr[b = b' | E] - \Pr[b \neq b' | E]| \le |p_{0, 0} - p_{1, 0}| + |p_{0, 1} - p_{1, 1}| \le 8 \sqrt 2 \cdot \sqrt{\frac{g(T + 1) + T^2}N}$, and we conclude that $\Pr[b = b' | E] \le \frac 1 2 + 4 \sqrt 2 \cdot \sqrt{\frac{g(T + 1) + T^2}N}$.
\end{proof}

Consider the algorithm that queries $H(1)$, and outputs $0$ if $H(1) = \ch$ and $1$ otherwise, then it will succeed with probability $1/2 + \Omega(1/N)$.
Combining these with \Cref{thm:decision-game-classiacl-advice-reduction} (or \Cref{thm:decision-game-quantum-advice-reduction}), we obtain a security bound for PRG in the AI-QROM (or QAI-QROM, respectively).

\begin{theorem}\label{thm:prg}
    $G_\prg$ is $\frac 1 2 + \tilde O\left(\frac{ST+T^2}N\right)^{1/3}$-secure in the AI-QROM, and $\frac 1 2 + \tilde O\left(\frac{S^5 T+S^4 T^2}N\right)^{1/19}$-secure in the QAI-QROM, where $\tilde O$ absorbs $\poly(\log N, \log M, \log S)$ factors.
\end{theorem}

\section{Yao's Box}
\label{sec:yaobox}

In this section, we show the security of Yao's box problem in the AI-QROM and QAI-QROM.
Let us start with reformalizing Yao's box problem using our language of security games.

\begin{definition}[Yao's box security game]
    Let $H$ be a random oracle $[N] \to [2]$. The security game $G_\yb = (C_\yb)$ is specified by three procedures $(\samp, \query, \ver)$, where: 
    \begin{enumerate}
        \item $\samp^H(x) = x$ takes some randomness $x \in [N]$, and outputs $x$ as the challenge pre-image.
        \item $\query^H(x, \cdot)$: for all input $x' \ne x$, the output is $H(x')$; $\query^H(x, x) = 1$. 
        \item $\ver^H(x, b)$ outputs $1$ if and only if $H(x) = b$. 
    \end{enumerate}
\end{definition}
    
To fill in the rest of the details from \Cref{def:security-game}, given a random oracle with binary range, the challenger samples a random input $x$ as the challenge. The adversary is allowed to make oracle queries to $H$ except on input $x$ and it wins if and only if it finds $H(x)$ without making any queries on $x$ (or equivalently, even if he makes query to $x$, he will just get a constant 1).  

We are going to prove the following lemma, the multi-instance security of $\Gmi{g}_\yb$.
\begin{lemma} \label{lem:ybmis}
    $\Gmi{g}_\yb$ is $\delta(g, T) = 1/2 + O\left(\sqrt{g T / N}\right)$-secure in the QROM.
\end{lemma}

Combining this lemma together with \Cref{thm:decision-game-classiacl-advice-reduction} (or \Cref{thm:decision-game-quantum-advice-reduction}), we can show the security of $G_\yb$ in the AI-QROM (or QAI-QROM, respectively).
\begin{theorem} \label{thm:ybaiqrom}
    $G_\yb$ is $\delta(S, T) = 1/2 + \tilde{O}\left( \left(S T / N\right)^{1/3} \right)$-secure in the AI-QROM. 
\end{theorem}
\begin{proof}
    Let a trivial algorithm with $0$ queries and $1$ bit of advice do the following: store $H(1)$ as its advice; answer $H(1)$ if the challenge is $1$ and answer a random bit otherwise. The advantage of this algorithm is $1/2 + 1/(2 N)$ because with probability $1 - 1/N$, it succeeds with probability $1/2$; with probability $1/N$, it succeeds with probability $1$. 

    Therefore, in \Cref{thm:decision-game-classiacl-advice-reduction}, we have $\eps_0 = 1/(2 N), S_0 = 1, T_0 = 0$ and $\Gmi{g}_\yb$ is $1/2+\eps(g, T) = 1/2+\tilde{O}(\sqrt{g T / N})$-secure. 
    We have $G_\yb$ is $1/2 + \eps'(S, T)$-secure where, 
    \begin{align*}
        \eps'(S, T) &\leq 4 \cdot \eps\left(\frac{10 \ln 2}{\eps'(S, T)} (S + S_0 + \log(1/\eps_0) + 2, T + T_0)  \right) \\
        & = \tilde{O}\left( \sqrt{ \frac{S T}{\eps'(S, T) \, N}  } \right).
    \end{align*}
    We conclude that $\eps'(S, T) = \tilde{O}\left( \left(S T / N\right)^{1/3} \right)$.
\end{proof}

\begin{theorem} \label{thm:ybqaiqrom}
    $G_\yb$ is $\delta(S, T) = 1/2 + \tilde{O}\left( \left(S^5 T / N\right)^{1/19} \right)$-secure in the QAI-QROM. 
\end{theorem}
\begin{proof}
    Let a trivial algorithm with $0$ queries and $1$ bit of advice do the following: store $H(1)$ as its advice; answer $H(1)$ if the challenge is $1$ and answer a random bit otherwise. The advantage of this algorithm is $1/2 + 1/(2 N)$ because with probability $1 - 1/N$, it succeeds with probability $1/2$; with probability $1/N$, it succeeds with probability $1$. 
    
    In \Cref{thm:decision-game-quantum-advice-reduction}, we have $S_0 = 1, T_0 = 0, \eps_0 = 1/(2 N)$ and $\eps(g, T) = \tilde{O}(\sqrt{g T / N})$. We have $G_\yb$ is $1/2 + \eps'(S, T)$-secure where, 
    \begin{align*}
        \eps'(S, T) &\leq 8 \cdot \eps\left( \tilde{O}\left(S + S_0 \right)^3 / \eps'^9, \tilde{O}\left((T + T_0)(S + S_0)^2 \right) / \eps'^8 \right) \\
        &= \tilde{O}\left( \sqrt{\frac{S^5 T}{\eps'(S, T)^{17} \cdot N}} \right).
    \end{align*}
     We conclude that $\eps'(S, T) = \tilde{O}\left( \left(S^5 T / N\right)^{1/19} \right)$.
\end{proof}

Now the problem reduces to proving its multi-instance security.

\begin{proof}[Proof of \Cref{lem:ybmis}]
    Let $\As$ be any $(g, T)$ adversary in the QROM.
    We again simulate the random oracle $H$ with a compressed standard oracle, which by \Cref{cor:csto-simulate} perfectly perserves the functionality of $\As \Longleftrightarrow \Cmi g_\yb$.
    We will also implement the projective measurement by directly measuring $\As$'s output, as the output in Yao's box is binary.

    Assume after passing the first $i-1$ challenges, conditioned on the challenges are $x_1, \cdots, x_{i-1}$ and the algorithm passes all the challenges, the overall state of the algorithm and compressed oracle is $|\phi_0\rangle = \sum_{z, D} \alpha_{z, D} |z\rangle |D\rangle$.  By \Cref{lem:bounded-database}, we know every possible $D$ with non-zero weight is of size at most $(i - 1)(T + 1) < g (T+1)$, which is the total number of queries to $\csto$. 
    
    Fixing an $x_i$ (which is chosen with probability $1/N$), let $p^{(t)}_{x_i}$ be the probability that after $T$ queries, measuring the database register, it gives a database containing $x_i$. In other words, define 
    \begin{align*} 
        P_{x_i} = \sum_{\substack{z\\ D: D(x_i) \ne \bot}} |z, D\rangle \langle z, D|.
    \end{align*}
    Then $p^{(0)}_{x_i} = |P_{x_i} |\phi_0\rangle|^2$. 
    Similar to \eqref{eq:unimportant-baddatabase}, we have, 
    \begin{align*}
        \sum_{x_i=1}^{N} p^{(0)}_{x_i} \leq g (T+1).
    \end{align*}
    
    As defined by $\query^H(x_i,\cdot)$, each query made by $\As$ will be sent to the challenger first and the challenger will apply the oracle query only on inputs that are not $x_i$. Therefore, assume before the first query the state is $|\phi_0\rangle = \sum_{x, u, \aux, D} \alpha_{x, u, \aux, D} |x, u, \aux\rangle |D\rangle$. After the first query, it becomes
    \begin{align*}
        |\phi_1\rangle := U_H \, \sum_{\substack{x \ne x_i, u, \aux \\ D}} \alpha_{x, u, \aux, D} |x, u, \aux\rangle |D \rangle
        + \sum_{\substack{u, \aux \\ D}} \alpha_{x_i, u, \aux, D} |x_i, u, \aux\rangle |D\rangle. 
    \end{align*}
    where $U_H$ is defined as $\csto$ applied on $x, u$ and $D$ register. 
    Recall that $p^{(0)}_{x_i} = |P_{x_i} |\phi_0\rangle|^2$ and $p^{(1)}_{x_i} = |P_{x_i} |\phi_1\rangle|^2$. We want to show the following lemma. 
    \begin{lemma}
        Let $p^{(0)}_{x_i}$ and $p^{(1)}_{x_i}$ be the probability defined above. Then $p^{(0)}_{x_i} = p^{(1)}_{x_i}$. 
    \end{lemma}
    \begin{proof}
    This can be shown by a direct calculation. 
    \begin{align*}
            p^{(1)}_{x_i} & = \left| P_{x_i} |\phi_1\rangle \right|^2 \\
            & = \left| \sum_{\substack{x \ne x_i, u, \aux \\ D}} \alpha_{x, u, \aux, D}\, P_{x_i} U_H |x, u, \aux\rangle |D\rangle  + \sum_{\substack{u, \aux \\ D}} \alpha_{x_i, u, \aux, D} P_{x_i}\, |x_i, u, \aux\rangle |D\rangle  \right|^2 \\
            & = \left| \sum_{\substack{x \ne x_i, u, \aux \\ D}} \alpha_{x, u, \aux, D}\, P_{x_i} U_H |x, u, \aux\rangle |D\rangle \right|^2 + \left|\sum_{\substack{u, \aux \\ D}} \alpha_{x_i, u, \aux, D} P_{x_i}\, |x_i, u, \aux\rangle |D\rangle  \right|^2 \\
            & \overset{(*)}{=}   \sum_{\substack{x \ne x_i, u, \aux, D \\ D(x_i) \neq \bot}}|\alpha_{x, u, \aux, D}|^2 +  \sum_{\substack{ u, \aux, D \\ D(x_i) \neq \bot}}|\alpha_{x_i, u, \aux, D}|^2 \\
            & = \sum_{\substack{x, u, \aux, D \\ D(x_i) \ne \bot}} |\alpha_{x, u, \aux, D}|^2  \,=\,  p^{(0)}_{x_i}.
    \end{align*}
    The equation (*) comes from the fact that if $U_H$ applies on $|x, u, D\rangle$ such that $x \ne x_i$, it does not change $D(x_i)$; in other words, if $D(x_i) =\bot$, then after applying $U_H$ it is still $\bot$; if $D(x_i) = v$, then after applying $U_H$ it is still $v$. 
    \end{proof}
    Let $p^{(t)} := |P_{x_i} \ket{\phi_i}|$, where $\ket{\phi_i}$ is the overall state after $i$-th query.
    As we do not assume any structure of $\ket{\phi_0}$, and the algorithm's local computation does not impact the projector value, we can use this lemma $t$ times and conclude that $p^{(t)}_{x_i} = p^{(0)}_{x_i}$ for all $t$.
    
    By \Cref{lem:zhandry-lemma5}, let the probability that $\As$ wins the game in the $i$-th round is $p_{x_i}$, after $T$ queries we have, $\sqrt{p_{x_i}}  \leq \sqrt{p^{(T)}_{x_i}} + \sqrt{1/M} = \sqrt{p^{(T)}_{x_i}} + \sqrt{1/2} = \sqrt{p^{(0)}_{x_i}} + \sqrt{1/2}$. Therefore, the overall probability is, 
    \begin{align*}
        \frac{1}{N} \sum_{x_i=1}^N p_{x_i} &= \frac{1}{N} \sum_{x_i=1}^N \left(\sqrt{p^{(0)}_{x_i}} + \sqrt{1/2}\right)^2  \\
        &= \frac{1}{2} + \sum_{x_i = 1}^N p^{(0)}_{x_i} / N + \sqrt{2} \cdot  \sum_{x_i = 1}^N \sqrt{p^{(0)}_{x_i}} / N  \\
        &\leq \frac{1}{2} + (g (T+1)) / N + \sqrt{2 g (T+1) / N} \\
        &\leq \frac{1}{2} + (1 + \sqrt{2}) \sqrt{g (T+1) / N}.
    \end{align*}
    The above inequalities come from $\sum_{x_i=1}^N p^{(0)}_{x_i} \leq g (T+1)$, Cauchy-Schwarz, and assuming $g (T+1) / N \le 1$ (which implies $g (T+1) / N \leq \sqrt{g (T+1) / N}$).
    On the other hand, if $g(T+1)/N > 1$, the bound above trivially holds.
    
    Let $B_i$ be the random variable for $b_i$, indicating if $\As$ passes the $i$-th round of the multi-instance game. Let $X_i$ be the random variable for the $i$-th challenge. The above statement says for any $\vx_{[i-1]}$, 
    \begin{align*}
        \Pr\left[ B_i = 1 \,|\, \vub_{[i-1]} = \mathbf{1}^{i-1} ,  \vux_{[i-1]} = \vx_{[i-1]}\right] \leq 1/2 + (1 + \sqrt{2}) \sqrt{g (T+1) / N}
    \end{align*}
    
    It implies $\Pr[B_1 = 1, B_2 = 1, \cdots, B_g = 1] \leq (1/2 + (1 + \sqrt{2}) \sqrt{g (T+1) / N})^g$. 
    Therefore, for any $(g, T)$ algorithm $\As$ in the QROM, the succeeding probability is at most $(1/2+O(\sqrt{g T / N})^g$. 
\end{proof}

\section{Salting Defeats Preprocessing}
\label{sec:salting}

In this section, we show generic salting defeats quantum pre-processing with classical/quantum advice. There exist schemes that are secure in the QROM but not in the AI-QROM. One example is collision-resistant hash function. For a random oracle $H: [N] \to [M]$ with large enough $N$, finding a collision without any pre-processing is hard \cite{aaronson2004quantum, liu2019finding}, which requires $T = \Omega(M^{1/3})$ queries. However, in the AI-QROM, as long as $S \geq 2 \cdot \log N$, the adversary can find a pair of collision in the pre-processing procedure and store it as the advice/leakage; in the online phase, it simply answers the pair. Therefore,  collision-resistance is broken in the AI-QROM.

\begin{definition}[Public-Salt Security Game]
  Given a security game $G = (C)$, where $C = (\samp, \allowbreak \query, \ver)$.
  Consider a ``salted'' oracle where we enlarge the input space: $H_S: [K] \times [N] \mapsto [M]$, for some parameter $K$.
  The salted security game $G_S = (C_S)$ defined on $H_S$, where $C_S = (\samp_S, \query_S, \ver_S)$ is specified as below:
  \begin{enumerate}
    \item $\samp_S^{H_S}(s, r)$: on input uniformly random salt $s \in [K]$ and $r \in R$, runs $\ch \gets \samp^{H_S(s, \cdot)}(r)$, returns $\ch_S := (s, \ch)$.
    \item $\query_S^{H_S}((s, r), (s', x))$: on input salt $s$ and challenger randomness $r$, query salt $s'$ and some input $x$, if $s = s'$, output $\query^{H_S(s, \cdot)}(r, x)$; otherwise, output $H_S(s', x)$.
    \item $\ver_S^{H_S}((s, r), \ans)$: output $\ver^{H_S(s, \cdot)}(r, \ans)$.
  \end{enumerate}
\end{definition}

\begin{lemma} \label{lem:saltingmis}
    For \emph{any} security game $G$ with security $\delta' = \delta'(T)$ in the QROM. For any integer $g$ and salt space $[K]$, the multi-instance game $\Gmi{g}_S$ is $\delta$-secure, where
    \begin{equation}
      \label{eq:saltingmis-general}
      \delta = \delta(g, T) = 2\delta' + \frac{g (T + 1)}K.
    \end{equation}
    $\Gmi{g}_S$ is also $\tilde\delta$-secure, where
    \begin{equation}
      \label{eq:saltingmis-decision}
      \tilde\delta = \tilde\delta(g, T) = \delta' + 3 \sqrt{\frac{g(T + 1)}K}.
    \end{equation}
\end{lemma}
We usually want to use \eqref{eq:saltingmis-decision} for decision games as a trivial random guessing adversary already wins with probability $1/2$, and the first bound trivially holds.

Combining  the first part of  \Cref{lem:saltingmis} with \Cref{thm:reduction-classical-advice-general}, we have the following theorem. Roughly speaking, the theorem says if a security game $G$ is $\delta(T)$-secure in the QROM, then the salted game $G_S$ with salt space $[K]$ is $2(\delta(T) + S T / K)$-secure against $(S, T)$ adversary in the AI-QROM. In other words, as long as $K \gg S T$, pre-processing does not help break the salted game $G_S$.

\begin{theorem} \label{thm:generalsalting-aiqrom}
 Let $G$ be a security game with security $\delta'(T)$ in the QROM. 
 Let $G_S$ be the salted security game with parameter $K$. Let $\delta_0$ be the winning probability of an adversary that outputs a random answer in the game $G_S$ without advice or making any query. Then $G_S$'s security $\delta(S, T)$ against $(S, T)$-adversaries in the AI-QROM is 
 \begin{align*}
     \delta(S, T) \leq 8 \cdot \left( \frac{(S + \log (1/\delta_0) + 1) (T + 1)}{K} + \delta'(T)  \right)
 \end{align*}
\end{theorem}
\begin{proof}
    If $G$ has security $\delta'(T)$, then $\Gmi{g}_S$ is $\delta''(g, T) = 2(\delta'(T) + g (T + 1)/K)$-secure. By \Cref{thm:reduction-classical-advice-general}, $G_S$ is $\delta(S, T)$-secure in the AI-QROM, where 
    $\delta(S, T) \leq 4 \cdot \delta''(S + \log 1/\eps_0 + 1, T)$. 
\end{proof}

Combining the first part of \Cref{lem:saltingmis} with \Cref{thm:public-ver-reduction}, we have the following theorem about salted public-verifiable security games in the QAI-QROM. 
\begin{theorem}
 Let $G$ be a \varul{publicly-verifiable} security game with security $\delta'(T)$ in the QROM. Let $T_\ver = \poly(\log M, \log N)$ be the upper bound on the number of $\tilde{H}$ queries for computing $\ver^{\tilde{H}}(\ch, \cdot)$, i.e., the number of queries required to do the public verification. 
  Let $G_S$ be the salted security game with parameter $K$.
  Let $\delta_0 = 1/\poly(N, M)$ be a non-negligible lower bound on the winning probability of an adversary that outputs a random answer in the $G_S$ without advice or making any query. 
  Then $G_S$'s security $\delta = \delta(S, T)$ against $(S, T)$ adversaries in the QAI-QROM satisfying 
 \begin{align*}
     \delta \leq  \tilde{O}\left( \frac{S T}{K \delta} + \delta'\left(  \tilde O(T) /\delta \right)  \right),
 \end{align*}
 where $\tilde O$ absorbs $\poly(\log N, \log M, \log S)$ factors.
\end{theorem}

Combining the second part of \Cref{lem:saltingmis} with  \Cref{thm:decision-game-classiacl-advice-reduction}, we have the following theorem about salted decision security games in the AI-QROM. Roughly speaking, the theorem says if a decision game $G$ is $1/2+\eps(T)$-secure in the QROM, then the salted game $G_S$ with salt space $[K]$ is $1/2+O(\eps(T) + (S T / K)^{1/3})$-secure against $(S, T)$  adversary in the AI-QROM. In other words, as long as $K \gg S T$, pre-processing does not help break the salted game $G_S$.  
\begin{theorem}
  Let $G$ be a \underline{decision} game with security $1/2 + \eps'(T)$. 
  Let $G_S$ be the salted decision game with parameter $K$.  
  Let $\eps_0$ be a lower bound on the advantage for an adversary for $G_S$ with $T_0$ queries and $S_0$ bits of advice. Then $G_S$'s security $\delta(S, T) = 1/2 + \eps(S, T)$ against $(S, T)$-adversaries in the \varul{AI-QROM} satisfies that
  \begin{align*}
      \eps(S, T) = O\left(\eps'(T + T_0) + \left(\frac{(S + S_0 + \log(1/\eps_0)) (T + T_0)}{K}\right)^{1/3} \right).
  \end{align*}
\end{theorem}

Combining the second part of \Cref{lem:saltingmis} with  \Cref{thm:decision-game-quantum-advice-reduction}, we have the following theorem about salted decision security games in the QAI-QROM.
\begin{theorem}
  Let $G$ be a \underline{decision} game with security $1/2 + \eps'(T)$.
  Let $G_S$ be the salted decision game with parameter $K$.
  Let $\eps_0 = 1/\poly(N, M)$ be a non-negligible lower bound on the advantage for an adversary for $G_S$ with $T_0$ queries and $S_0$ bits of advice.
  Then $G_S$'s security $\delta(S, T) = 1/2 + \eps(S, T)$ against $(S, T)$-adversaries in the QAI-QROM satisfies that
  \begin{align*}
      \eps = \tilde O\left(\eps'\left(\frac{\tilde O((T + T_0) (S + S_0)^2)}{\eps^8}\right) + \sqrt{\frac{(S + S_0)^5 (T + T_0)}{K \eps^{17}}} \right),
  \end{align*}
  where $\tilde O$ absorbs $\poly(\log N, \log M, \log S)$ factors.
\end{theorem}

\subsection{Multi-Instance Security of Salted Game}

In this section, we prove \Cref{lem:saltingmis}.
We first argue some general facts about the salting game, and then show that \eqref{eq:saltingmis-general} and \eqref{eq:saltingmis-decision} holds respectively, and therefore concluding the proof.

Let $\As$ be any $(g, T)$ adversary in the QROM for $\Gmi{g}_S$.
We again start by considering the following algorithm that perfectly simulates the original interaction:
\begin{enumerate}
  \item Instead of directly sampling the oracle distribution $H_S: [K] \times [N] \mapsto [M]$, we instead consider the equivalent sampling procedure where we instead sample a random function $H_s: [K] \mapsto \mathcal R_N$, where $\mathcal R_N$ denotes the randomness for sampling the original oracle $H: [N] \mapsto [M]$.
        Let $\mathcal S: \mathcal R_N \mapsto ([N] \mapsto [M])$ be the original oracle sampler, that is $\mathcal S$ on input some randomness, samples (according to the distribution) the entire truth table of $H: [N] \mapsto [M]$.
        $H_S([K], [N])$ has the exact same distribution as $\mathcal S(H_s([K]))([N])$ as long as the sampler $\mathcal S$ is correct.
        
        For random functions or random oracle model, $\mathcal R_N = [M^N]$, and any arbitrary bijection $[M^N] \leftrightarrow [M]^N$ is a correct sampler.
        
  \item By definition we have $\ver_S^{H_S}((s, r), \ans) = \ver^{H_S(s, \cdot)}(r, \ans)$, therefore, the verification algorithm will only ever query $H_S(s, \cdot) = \mathcal S(H_s(s))(\cdot)$.
        Using this observation, whenever we need to run $\ver_S^{H_S}((s, r), \ans)$, we instead first obtain the (entire truth table of) oracle $H' \gets \mathcal S(H_s(s))$, and run $b \gets \ver^{H'}(r, \ans)$ directly.
        
        Note that after this change, each round will only ever issue $(T + 1)$ queries to the oracle $H_s$.
        
  \item We replace the oracle access to $H_s$ with the compressed standard oracle operator $\csto$, which preserves the functionality by \Cref{cor:csto-simulate}.
\end{enumerate}

Conditioned on passing the first $(i-1)$ challenges, and the challenges being $(s_1, \ch_1)$, $\cdots$, $(s_{i-1}, \ch_{i-1})$ and the algorithm having passed all the challenges (we again only consider the $s$'s and $\ch$'s that have non-zero probability), we denote the overall state of the algorithm and compressed oracle to be $|\phi_0\rangle = \sum_{z, D} \alpha_{z, D} |z\rangle |D\rangle$.
By \Cref{lem:bounded-database}, since the total number of queries made by $\As$ and $\Cmi{g}$ is at most $(i-1)(T + 1)$, we know every possible $D$ with non-zero weight is of size at most $g (T+1)$.

Fixing a salt $s$ and challenger randomness $r$. Let $p_{s, r}$ be the probability that $\As$ wins conditioned on salt $s$ and randomness $r$.  Let $q_{s}$ be the probability that measuring the database register gives a database containing $s$.
In other words, define 
\begin{align*} 
    Q_{s} = \sum_{z, D : D(s) \ne \bot} |z, D\rangle \langle z, D|
\end{align*}
Then $q_{s} = |Q_{s} |\phi_0\rangle|^2$. 
Similar to \eqref{eq:unimportant-baddatabase}, we have, $\sum_{s\in[K]} q_{s} \leq g (T+1)$. 

Let $U_{s, r}$ be the unitary describing the rest of the computation for this round, that is including: $\As$'s oracle queries (which also includes the oracle expansion in superposition using $\mathcal S$) and local computation, our added query and $\mathcal S$ for computing $H'$ in change 2, and the final local computation $\ver^{H'}(r, \ans)$.
Let $P_{s, r}$ be the projective measurement that $\ver$ outputs 1. Using triangular inequality, we then have
\begin{align*}
    \sqrt{p_{s, r}} = & \left| P_{s, r} U_{s, r} \ket \phi_0 \right| \\
            \leq & \left| P_{s, r} U_{s, r} Q_s \ket \phi_0  \right| + \left| P_{s, r} U_{s, r} (I-Q_s) \ket \phi_0  \right| \\
            \leq & \left|Q_s \ket \phi_0  \right| + \left| P_{s, r} U_{s, r} (I-Q_s) \ket \phi_0  \right| \\
            = & \sqrt{q_s} + \left| P_{s, r} U_{s, r} (I-Q_s) \ket \phi_0   \right|.
\end{align*}    

\begin{proof}[Proof of \eqref{eq:saltingmis-general}]
    Taking expectation over the randomness of $s, r$, the probability of $\As$ succeeding in round $i$ is
    \begin{equation}
    \label{eq:saltingmis-general-uninteresting}
    \begin{aligned}
        \frac{1}{K \cdot R} \sum_{s, r} p_{s, r} &\leq \frac{1}{K \cdot R} \sum_{s, r} \left(\sqrt{q_s} + \left| P_{s, r} U_{s, r} (I-Q_s) \ket \phi_0   \right|\right)^2 \\
                    &\leq \frac{2}{K \cdot R} \sum_{s, r} \left(q_s + \left| P_{s, r} U_{s, r} (I-Q_s) \ket \phi_0   \right|^2  \right) \\
                    &\leq \frac{2 \cdot g (T + 1)}{K} +   \frac{2}{K \cdot R}\, \sum_{s, r}  \left| P_{s, r} U_{s, r} (I-Q_s) \ket \phi_0   \right|^2.
    \end{aligned}
    \end{equation}
    
    Next, we are going to show the second term is at most $2 \, \delta'(T)$. We are going to construct an algorithm $\Bs$ making at most $T$ queries which takes $\As$ as a subroutine and wins the game $G$ with probability at least $\frac{1}{K \cdot R}\, \sum_{s, r}  \left| P_{s, r} U_{s, r} (I-Q_s) \ket \phi_0   \right|^2$, which can not be greater than $\delta'(T)$ (the security of $G$).
    
    Consider the following algorithm $\Bs$ breaking $G = (C)$, $\Bs$ is going to simulate the multi-instance game for the first $(i-1)$ rounds and simulate the $i$-th round by making queries to $C$. Therefore, if $\As$ win the $i$-th round of $\Gmi{g}$, then $\Bs$ wins $G$.
    \begin{enumerate}
        \item $\Bs$ takes $\As$ as a subroutine. Let $H$ be the oracle in the game $G$. 
        \item $\Bs$ simulates the random oracle $H_S$ exactly the same as our simulator (in particular, using $\csto$ on $H_s$ and does not use $H$) and interacts with $\As$ up to round $(i - 1)$.%
        \item $\Bs$ checks that $\As$ passes the first $(i - 1)$ rounds.
        \item $\Bs$ samples a salt $s \gets [K]$ and it checks (or measures) whether $D(s) = \bot$.
        \item If any of the two checks fails, $\Bs$ clears all registers and start back at the beginning.
              
              If both checks pass, note that the overall state now conditioned on $(s_i, r_i, b_i)$'s is exactly $\frac{1}{\ell_s} (I - Q_s) \ket {\phi_0}$ where $\ell_s = |(I - Q_s) \ket {\phi_0}| \leq 1$ is the normalization factor.

        \item $\Bs$ simulates the $i$-th round using a modified oracle $H^*$, where $H^*$ is consistent with $H_S$, except on salt $s$ where the queries are instead forwarded to $\query^H$.
              In particular,
        \begin{itemize}
            \item 
            The challenge sent to $\As$ is instead $\samp^{H^*}_S(s, r) = \samp^{H^*(s, \cdot)}(r) =  \samp^H(r) = \ch$ 
            where $\ch$ is the challenge received from $C$. 
            
            \item When $\As$ makes a quantum oracle query $\query_S^{H^*}((s,r),(s',x))$, if $s \ne s'$, $\Bs$ can simulate the query using $\csto$ on $H_s$; if $s = s'$, 
            $\Bs$ makes an oracle query $(r, x)$ to $C$,  since $\query_S^{H^*}((s,r),(s',x)) = \query^H(r, x)$ if $s = s'$. 
        \end{itemize}
        Note that this simulation does not require knowledge of $r$ and only uses query access to $\query^H$, as desired.
        
        \item Finally, $\Bs$ output $\ket\ans$ from the overall state $\frac{1}{\ell_s} U_{s, r} (I - Q_s) \ket {\phi_0}$. 
    \end{enumerate}
    Since we condition on step 4 passing ($D(s) = \bot$), $H_S$ and $H^*$ are perfectly indistinguishable for $\As$.
    Therefore the succeeding probability of $\Bs$ is, 
    \begin{align*}
        \frac{1}{K \cdot R} \sum_{s, r} \left| \frac{1}{\ell_s} P_{s, r} U_{s, r} (I - Q_s) \ket {\phi_0} \right|^2 \geq  \frac{1}{K \cdot R} \sum_{s, r} \left|P_{s, r} U_{s, r} (I - Q_s) \ket {\phi_0} \right|^2.
    \end{align*}
    By security of the unsalted game, left hand side is bounded by $\delta'$, therefore,
    \begin{equation}
      \label{eq:missalting-simulator-bound}
      \frac{1}{K \cdot R} \sum_{s, r} \left|P_{s, r} U_{s, r} (I - Q_s) \ket {\phi_0} \right|^2 \leq \delta'.
    \end{equation}
    Putting this together with \eqref{eq:saltingmis-general-uninteresting}, we conclude that the algorithm wins $i$-th round with probability at most $2 g (T + 1)/K + 2 \delta'(T)$.
    Therefore, invoking \Cref{lem:conditioning2mis}, for any $(g, T)$ algorithm in the QROM, the succeeding probability is at most 
    $(2 g (T + 1)/K + 2 \delta'(T))^g$.
\end{proof}
    
\begin{proof}[Proof of \eqref{eq:saltingmis-decision}]
    Using \eqref{eq:missalting-simulator-bound}, the probability of $\As$ winning round $i$ conditioned on $(s_i, r_i, b_i)$'s is, 
    \begin{align*}
        \frac{1}{K \cdot R} \sum_{s, r} p_{s, r}
                    &\leq \frac{1}{K \cdot R} \sum_{s, r} \left(\sqrt{q_s} + \left| P_{s, r} U_{s, r} (I-Q_s) \ket {\phi_0}   \right|\right)^2 \\
                    &= \frac{1}{K \cdot R} \sum_{s, r} \left(q_s + \left| P_{s, r} U_{s, r} (I-Q_s) \ket \phi_0   \right|^2 + 2 \sqrt{q_s} \left| P_{s, r} U_{s, r} (I-Q_s) \ket{\phi_0}   \right|  \right) \\
                    &\leq \frac{g (T + 1)}{K} + \delta'(T) + \frac{2}{K \cdot R} \sum_{s, r} \sqrt{q_s} \cdot 1  \\ 
                    &\leq \frac{g (T + 1)}{K} + \delta'(T) + 2 \sqrt{\frac{g (T + 1)}{K}} \\
                    &\leq \delta'(T) + 3 \sqrt{\frac{g (T + 1)}{K}},
    \end{align*}
    The last inequality holds assuming $g (T + 1) / K \leq 1$, since otherwise \eqref{eq:saltingmis-decision} also trivially holds.
    Therefore, using \Cref{lem:conditioning2mis}, for any $(g, T)$ algorithm in the QROM, the succeeding probability is at most 
    $\left(\delta'(T) + 3 \sqrt{\frac{g (T + 1)}{K}}\right)^g$.
\end{proof}

\subsection{Hardness of Salted Collision-Resistant Hash Functions}

\begin{definition}[Collision-Resistant Hash Security Game]
    Security game $G_\crh = (C_\crh)$ is specified by three procedures $(\samp, \query, \ver)$, where:
    \begin{enumerate}
        \item $\samp^H(\bot) = \bot$.
        \item $\query^H(\bot, \cdot) = H(\cdot)$ provides query access to $H$.
        \item $\ver^H(\bot, (x_1, x_2))$ outputs $1$ if and only if $x_1 \neq x_2$ and $H(x_1) = H(x_2)$. 
    \end{enumerate}
\end{definition}
It is easy to see that $S = 2\log M$ suffices to break this security game with optimal probability.
However, without advice, the problem is hard even against quantum computations.

\begin{proposition}[{\cite[Corollary 2]{zhandry2019record}}]
    $G_\crh$ has security
    $\Theta(T^3/M)$ in QROM.
\end{proposition}

It is easy to verify that this game is publicly-verifiable, as $\ver^H(\bot, (x_1, x_2)) = \ver^{\query^H}(\bot, (x_1, x_2))$.
Denote the salted collision-resistant hash game as $G_{\crh, S}$.
Combining this lemma with \Cref{thm:generalsalting-aiqrom}, and the fact that random guessing has winning probability $1/M$, we obtain the following corollary.

\begin{corollary}
    $G_{\crh,S}$ is $O(T^3/M + (S + \log M) T / K)$-secure in the AI-QROM, and $\tilde O(T^3/M + ST/K)^{1/4}$-secure in the QAI-QROM, where $\tilde O$ absorbs $\poly(\log N, \log M, \log K, \log S)$ factors.
\end{corollary}

We now proceed to show the implication of this corollary in complexity theory for function problems.

\begin{definition}[Polynomial Weak Pigeonhole Principle]
    A binary relation $P(x, y)$ is in $\mathsf{PWPP}$ if and only if $P$ is Karp-reducible to $\mathsf{COLLISION}(\mathcal C, (x_1, x_2))$, where $\mathcal C$ is a circuit with $n$ inputs and $m < n$ outputs, and $(\mathcal C, (x_1, x_2))$ is in $\mathsf{COLLISION}$ if and only if $\mathcal C(x_1) = \mathcal C(x_2)$.
\end{definition}

\begin{definition}[Function Bounded-error Quantum Polynomial time with Polynomial-size Quantum advice]
    A binary relation $P(x, y)$ is in $\mathsf{FBQP/qpoly}$ if and only if there exists a polynomial-time uniform family of quantum circuits $\{C_n\}_{n \in \mathbb N}$ such that for any $x$, there exists some quantum state $\rho$ of dimension $2^{\poly(|x|)}$, such that $\Pr[P(x, C_{|x|}(\rho, x))] \ge 2/3$.
\end{definition}

\begin{theorem}
    \label{thm:oracle-separation}
    Relative to a random oracle $\mathcal O$, $\mathsf{PWPP}^\mathcal O \not\subseteq \mathsf{FBQP}^\mathcal O\mathsf{/qpoly}$.
\end{theorem}
\begin{proof}
    For any integer $n$, let $\mathcal O_n: [2^n] \times [2^{n+1}] \mapsto [2^n]$ be a random oracle.
    Consider the relation $P^\mathcal O(k, (x_1, x_2))$ that corresponds to finding collision in $\mathcal O(k, \cdot)$ (which is a circuit consisting of one oracle gate) for any $k \in [2^n]$ is a problem in $\mathsf{PWPP}^\mathcal O$ as $2^{n+1} > 2^n$, however, by the corollary above, any quantum polynomial-time algorithm with polynomial-size advice can succeed with probability $\exp(-\Omega(n))$ over a random $k$, in particular, this implies that there exists some hard instance $k$ such that the algorithm cannot succeed with non-negligible probability, and therefore proving that this problem is not in $\mathsf{FBQP}^\mathcal O\mathsf{/qpoly}$.
\end{proof}

\subsection{Tightness of Salted Game Lower Bound}
\label{sec:salting-tight}

In this subsection, we establish the tightness of \Cref{thm:generalsalting-aiqrom}.
In particular, we show that the bound is asymptotically tight up to poly-logarithmic factors for the following simple game. 
\begin{definition}[Prediction Security Game]
    Let $H$ be a random oracle $[1] \to [M]$, i.e. $N = 1$.
    The security game $G_\gyb = (C_\gyb)$ is specified by three procedures $(\samp, \query, \ver)$, where: 
    \begin{enumerate}
        \item $\samp^H(\bot) = \bot$. 
        \item $\query^H(\bot, \cdot) = \bot$.
        \item $\ver^H(\bot, y)$ outputs $1$ if and only if $H(1) = y$. 
    \end{enumerate}
\end{definition}
For this game, the adversary is asked to predict $H(1)$, and therefore its security $\delta'(T) = 1/M$.
Intuitively, the salted prediction game, denoted as $G_{\gyb,S}$ is simply Yao's box over a large alphabet $[M]$.
Combining this with \Cref{thm:generalsalting-aiqrom}, we immediately obtain the following corollary.

\begin{corollary}
    $G_{\gyb,S}$ is $\delta(S, T) = O(1/M + (S + \log M) T / K)$-secure in the AI-QROM. 
\end{corollary}

\begin{proposition}
    For any $S, T$ such that $S (T+1) \leq K \log M$, there exists an $(S, T)$-adversary in the AI-QROM that breaks $G_{\gyb}$ with winning probability at least $\Omega(1/M + S (T+1) / (K  \log M))$. 
\end{proposition}
\begin{proof}
    For the purpose of this proof, we will instead view $H_S: [K] \mapsto \Z/M\Z$ as $[K] \times [1] \simeq [K]$ and $[M] \simeq \Z/M\Z$.
    In the pre-processing stage, the algorithm stores the following $S/\log M$ elements in $\Z/M\Z$: for the $j$-th element ${\sf sum}_j$, 
    \begin{align*}
        {\sf sum}_j = \sum_{k=1}^{T+1} H(j (T+1) + k)  ,\quad\quad \forall j = 0, 1, \cdots, S/\log M - 1.
    \end{align*}
    
    In the online state, given a salt $s \in [K]$, if $s \not\in \{1, \cdots, S (T+1) / \log M\}$, the algorithm randomly guesses a $y$.
    Otherwise, let $j$ be the unique integer such that $j(T + 1) < s \le (j + 1)(T + 1)$. The algorithm queries $H(s')$ for all $s' \in \{j (T+1) + 1, \cdots, (j+1) (T+1)\}, s' \ne s$, and subtract them from ${\sf sum}_j$ and output the result, which succeeds with probability 1.
\end{proof}

\fi

\section*{Acknowledgements}

Kai-Min Chung is partially supported by the  Academia Sinica Career Development Award under Grant no. 23-17, and MOST QC project under Grant no. MOST 108-2627-E-002-001-.

Siyao Guo is supported by Shanghai Eastern Young Scholar Program.

Qipeng Liu is supported in part by NSF and DARPA. Opinions, findings and conclusions or recommendations expressed in this material are those of the author(s) and do not necessarily reflect the views of NSF or DARPA.

Luowen Qian is supported by DARPA under Agreement No. HR00112020023.

\fi

\printbibliography

\ifexabs\else
\ifieeecs\else
\appendix
\section{Extended Preliminaries for Compressed Oracle}
\label{appendix:compressed-oracle}
\begin{customlemma}{\ref{lem:bounded-database-phase}}
    \statelemboundeddatabasephase
\end{customlemma}
\begin{proof}
    We begin by assuming that there are no intermediate measurements.
    Note that if $\As$ makes some local quantum computation $U \otimes I$ on his on register $\ket z$, it does not affect the statement, as
    \begin{align*}
        (U \otimes I) \sum_{z, D: |D| \leq T} \alpha_{z, D} \ket z  \ket {\phi_D} = \sum_{D : |D| \leq T}  \left( U \sum_{z} \alpha_{z, D} \ket z\right) \ket {\phi_D}.
    \end{align*} 
    
    We prove the lemma by induction. The base case is before $\As$ makes any query.
    The overall state by definition is $\sum_{z} \alpha_{z} \ket z \otimes \frac{1}{M^{N/2}} \sum_{H}  \ket H$. 
    If we view the register $\ket H$ under the Fourier basis $\ket{\phi_D}$, there is only non-zero amplitude on $D_0$ which is an all-zero vector. 
    
    If $\As$ makes an oracle query to phase oracle, using the ``phase kickback'' trick,
    \begin{align*}
        & \equad \pho \sum_{\substack{z = (x, u, \aux) \\ D: |D| \leq T}} \alpha_{z, D} \ket {x, u, \aux}  \frac{1}{M^{N/2}} \sum_{H} \omega_M^{\langle D, H\rangle}  \ket H \\
        & =  \sum_{\substack{z = (x, u, \aux) \\ D: |D| \leq T}} \alpha_{z, D} \ket {x, u, \aux}  \frac{1}{M^{N/2}} \sum_{H} \omega_M^{\langle D, H\rangle + u H(x)}  \ket H  \\
        & =  \sum_{\substack{z = (x, u, \aux) \\ D: |D| \leq T}} \alpha_{z, D} \ket {x, u, \aux}  \ket{\phi_{D\oplus (x,u)}},
    \end{align*}
    where $D\oplus (x,u)$ is defined as a new vector that is almost identical to $D$ but the $x$-th entry is updated to $D(x) + u$. %
    Therefore, it is easy to see that $|D \oplus (x, u)| \leq |D| + 1$ for all $D$, which implies that the resulting state has only non-zero amplitude on $\ket z \ket {D'}$ whose $|D'| \leq T + 1$. 
    Using induction, we reach our conclusion for the case where there are no intermediate measurements.
    ~\\
    
    Now consider the case where $\As$ makes intermediate measurements. Without loss of generality, we can delay all the measurements to the last step. Before making these measurement, we know the overall state can be written as 
    \begin{align*}
        \sum_{z, b_1, \cdots, b_k, D: |D| \leq T} \alpha_{z, b_1\cdots, b_k, D} \ket z \ket {b_1, \cdots, b_k} \otimes \ket {\phi_D},
    \end{align*}
    where each $\ket {b_i}$ is going to be measured by $\As$. Therefore, it is easy to see that conditioned on any outcomes of the measurements (of non-zero probability), the resulting state has only non-zero amplitude on $\ket z \ket {D}$ whose $|D| \leq T$.
\end{proof}

\begin{customlemma}{\ref{lem:bounded-database}}
    \statelemboundeddatabase
\end{customlemma}
\begin{proof}

    First,  $\As$ applies a local unitary/measurement over its register. 
    Since the operation does not operate on $D$, it cannot change the probability distribution of the outcomes on measuring $D$.
    As a consequence, this operation cannot increase $\max_{\alpha_{z, D} > 0} |D|$. 
    
    Second,  $\As$ makes an oracle query on state $\ket x \ket u$.
    We claim that $\csto \ket x \ket u \ket D$ is a state with support on database $D'$ such that $|D'| \leq |D| + 1$, which is implicitly stated in the proof for Lemma 4, Hybrid 4 in~\cite{zhandry2019record}, and can also be directly derived from the definition of $\csto$ and $\stddecomp$.
    We reach our conclusion by using an induction if there are no intermediate measurements.
    
    Finally, $\As$ makes intermediate measurements. Without loss of generality, all these measurements are delayed to the last step. Before making these measurement, by the above discussion, the overall state can be written as 
    \begin{align*}
        \sum_{z, b_1, \cdots, b_k, D: |D| \leq T} \alpha_{z, b_1\cdots, b_k, D} \ket z \ket {b_1, \cdots, b_k} \otimes \ket {D}
    \end{align*}
    where each $\ket {b_i}$ is going to be measured by $\As$. Therefore, it is easy to see that conditioned on any outcomes of the measurements (of non-zero probability), the resulting state has only non-zero amplitude on $\ket z \ket {D}$ whose $|D| \leq T$.
\end{proof}

\begin{customlemma}{\ref{lem:zhandry-lemma5}}[{\cite[Lemma 5]{zhandry2019record}}]
    \statelemzhandry
\end{customlemma}
\begin{proof}
    Without conditioning on intermediate measurements, this is the exact statement of \cite[Lemma 5]{zhandry2019record}.
    Observe that their proof works as long as the final state of the algorithm and the oracle can be written as,
    \begin{align}
        \label{eq:lemma5-final-state}
        \sum_{\substack{\mathbf{x}, \mathbf{y}, z, D \\ \forall i, D(x_i) = \bot}} \sum_{I \subseteq [k]} \sum_{\mathbf{r}_I\in \{1, ..., M-1\}^{|I|} }  \alpha_{\mathbf{x}, \mathbf{y}, z, D, \mathbf{r}_I}\, \frac{1}{\sqrt{M^{|I|}}} \sum_{\mathbf{y}'_I \in (\Z/M\Z)^{|I|}} \omega_M^{\langle \mathbf{y}'_I, \mathbf{r}_I\rangle} \left| \mathbf{x}, \mathbf{y}, z, D \cup (\mathbf{x}_I, \mathbf{y}'_I) \right\rangle,
    \end{align}
    where $\mathbf x \in [N]^k, \mathbf y \in (\Z/M\Z)^k$.
    That is, if $D(x) \ne \bot$, then the superposition of $D(x)$ decomposed in the Fourier basis does not contain a uniform superposition (corresponding to $\mathbf{r}_I \in \{1, ..., M-1\}^{|I|}$ in the formula, or in other words, no $0$ entry).  
    
    Zhandry showed the final state satisfies the condition above. It is indeed easy to prove.
    Recall that $\csto = \stddecomp \circ \csto' \circ \stddecomp$.
    As the second $\stddecomp$ will map any uniform $\sum_y \ket y$ to $\ket \bot$, the decomposition in the Fourier basis does not contain uniform superposition. 

    Now consider the case of having intermediate measurements.
    We claim that the state has the above form, conditioned on arbitrary outcomes (with non-zero probability) of $\As$'s intermediate measurements.
    Without loss of generality, we defer all measurements until the end of the computation.
    Before making these measurements, we know the overall state can be written as
    \begin{align*}
        \sum_{\substack{\mathbf{x}, \mathbf{y}, z, \omega, D \\ \forall i, D(x_i) = \bot}} \sum_{I \subseteq [k]} \sum_{\mathbf{r}_I\in \{1, ..., M-1\}^{|I|} }  \alpha_{\mathbf{x}, \mathbf{y}, z, \omega, D, \mathbf{r}_I}\, \frac{1}{\sqrt{M^{|I|}}} \sum_{\mathbf{y}'_I \in (\Z/M\Z)^{|I|}} \omega_M^{\langle \mathbf{y}'_I, \mathbf{r}_I\rangle} \left| \mathbf{x}, \mathbf{y}, z, \omega, D \cup (\mathbf{x}_I, \mathbf{y}'_I) \right\rangle ,
    \end{align*}
    where $\omega \in \Omega$ is the register to be measured. It is easy to see that after conditioning on $\omega \in E$ where $E \subseteq \Omega$ being an arbitrary event, the final state can still be written as \eqref{eq:lemma5-final-state}, as conditioning simply zeroes out some $\alpha$'s for every $\omega \not\in E$ and re-normalizes.
    The rest of the proof simply follows Zhandry's proof for \cite[Lemma 5]{zhandry2019record}.
\end{proof}

\fi
\fi

\end{document}